\documentclass[12pt]{article}
\usepackage{authblk}
\usepackage{geometry}
\usepackage{amsthm}
\usepackage{mathrsfs}
\usepackage{customcommands}
\usepackage{bm}
\usepackage{Chrisstyle}

\usepackage{amsmath}
\usepackage{amssymb}
\usepackage{amsfonts}
\usepackage{braket}
\usepackage[normalem]{ulem}
\usepackage{color}
\usepackage{ulem}
\usepackage{mathtools}
\usepackage{tensor} 
\usepackage{overpic}
\usepackage{breqn} 
\usepackage{subcaption}
\usepackage{float}
\usepackage[export]{adjustbox}
\newtheorem{theorem}{Theorem}

\DeclareMathOperator{\tr}{tr}

\graphicspath{{./fig/}{./fig/sec3/}{./fig3/}}
\newtheorem{corollary}[theorem]{Corollary}
\newtheorem{lemma}[theorem]{Lemma}
\theoremstyle{definition}
\newtheorem{definition}{Definition}[section]

\theoremstyle{remark}

\title{Reflected entropy in random tensor networks}

\author[1]{Chris Akers,}
\author[2]{Thomas Faulkner,}
\author[2]{Simon Lin,}
\author[3]{Pratik Rath}

\affiliation[1]{Center for Theoretical Physics,\\
Massachusetts Institute of Technology, Cambridge, MA 02139, USA}
\affiliation[2]{Department of Physics, University of Illinois,\\ 1110 W. Green St., Urbana, IL 61801-3080, USA}
\affiliation[3]{Department of Physics, University of California, Santa Barbara, CA 93106, USA}

\emailAdd{cakers@mit.edu}
\emailAdd{tomf@illinois.edu}
\emailAdd{shanlin3@illinois.edu}
\emailAdd{rath@ucsb.edu}

\abstract{In holographic theories, the reflected entropy has been shown to be dual to the area of the entanglement wedge cross section. We study the same problem in random tensor networks demonstrating an equivalent duality. For a single random tensor we analyze the important non-perturbative effects that smooth out the discontinuity in the reflected entropy across the Page phase transition. By summing over all such effects, we obtain the reflected entanglement spectrum analytically, which agrees well with numerical studies. This motivates a prescription for the analytic continuation required in computing the reflected entropy and its R\'enyi generalization which resolves an order of limits issue previously identified in the literature. We apply this prescription to hyperbolic tensor networks and find answers consistent with holographic expectations. In particular, the random tensor network has the same non-trivial tripartite entanglement structure expected from holographic states.  We furthermore show that the reflected R\'enyi spectrum is not flat, in sharp contrast to the usual R\'enyi spectrum of these networks. We argue that  the various distinct contributions to the reflected entanglement spectrum can be organized into approximate superselection sectors.  We interpret this as resulting from an effective description of the canonically purified state as a superposition of distinct tensor network states. Each network is constructed by doubling and gluing various candidate entanglement wedges of the original network. The superselection sectors are labelled by the different cross-sectional areas of these candidate entanglement wedges.}

\begin{document}
\maketitle

\section{Introduction}\label{sec:intro}

Understanding the entanglement structure of holographic states has played a significant role in demystifying quantum gravity and the emergence of spacetime. In particular, quantum error correction has served as a useful paradigm for explaining various features of quantum gravity including the holographic principle \cite{tHooft:1993dmi,Susskind:1994vu}, the Ryu-Takayanagi (RT) formula \cite{Ryu:2006bv,Ryu:2006ef, Hubeny:2007xt,Engelhardt:2014gca} and subregion duality \cite{Almheiri:2014lwa,Jafferis:2015del, Dong:2016eik}. In particular, random tensor networks \cite{Hayden:2016cfa}, which are understood to model fixed-area states in holography \cite{Akers:2018fow,Dong:2018seb,Dong:2019piw}, have been particularly effective in capturing non-perturbative gravitational effects.

While calculations of the von Neumann entropy have led to most of the insight into holography, it is useful to consider other quantum information (QI) quantities in order to give a more complete picture of the emergence of spacetime from entanglement. For example, von Neumann entropy cannot discriminate between different kinds of multipartite entanglement. In this paper, we will focus on a quantity called the reflected entropy, $S_R(A:B)$ \cite{Dutta:2019gen}, defined for a mixed state $\rho_{AB}$ on two parties $A,B$. $S_R(A:B)$ is simply defined using a canonical purification $\left| \sqrt{\rho_{AB}} \right>$ interpreted as a pure state on the Hilbert space of linear matrices, a familiar operation in the context of going from the thermal density matrix to the thermofield double state \cite{Maldacena:2001kr}. This Hilbert space hosts a doubled/reflected copy of the $A,B$ parties and $S_R(A:B)$ is defined as 
\begin{equation}
	S_R(A:B) \equiv S_{vN}(\rho_{AA^\star})_{\ket{\sqrt{\rho_{AB}}}}, 
\end{equation}
the von Neumann entropy of the canonical purification reduced to $A A^\star$ where $A^\star$ is the reflection of $A$. We will review the definition in more detail and recall some basic properties of $S_R(A:B)$ in Section~\ref{sec:setup}. While an operational interpretation of $S_R(A:B)$ is lacking, and its broader status as a quantum information measure is not yet clear (see Ref.~\cite{Hayden:2021gno} for some recent progress), reflected entropy has the advantage that it is computable using various analytic and numerical techniques. 

\begin{figure}
	\centering 
	\begin{subfigure}
		[b]{0.4
		\textwidth} 
		\includegraphics[width=
		\textwidth]{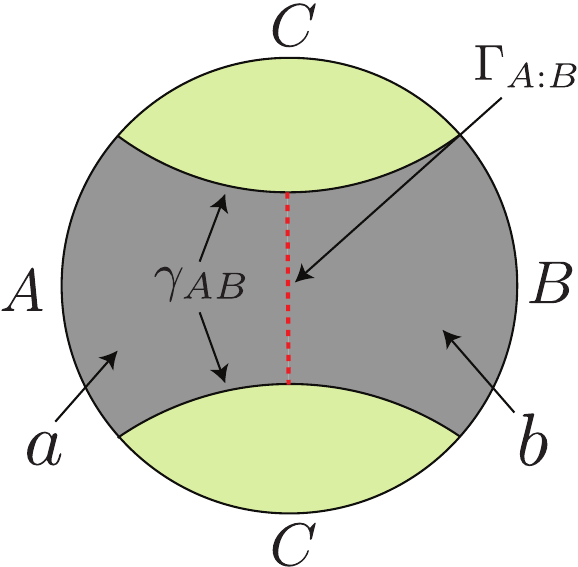} \subcaption{} 
	\end{subfigure}
	\hspace{5mm} 
	\begin{subfigure}
		[b]{0.4
		\textwidth} 
		\includegraphics[width=
		\textwidth]{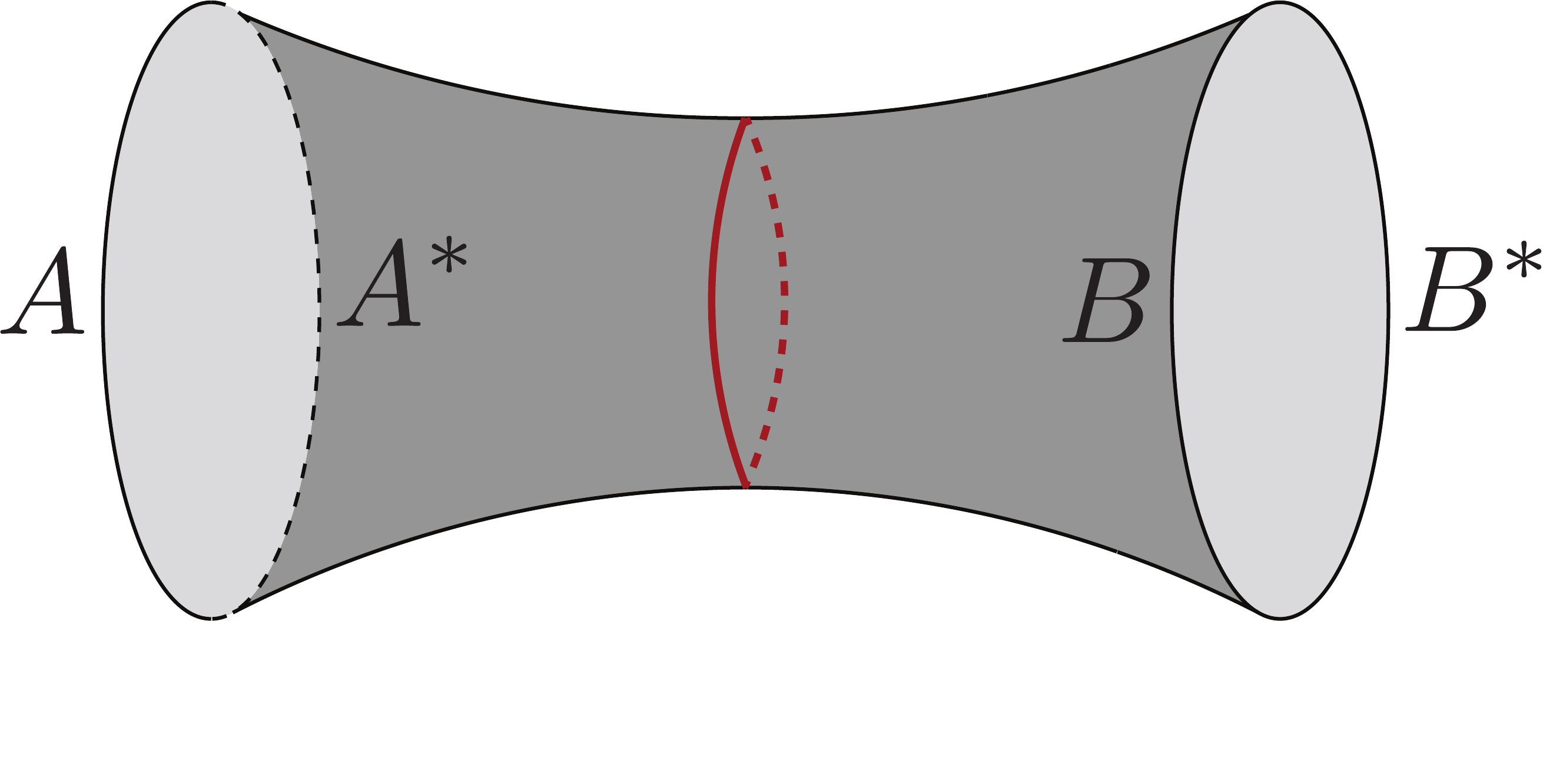} \subcaption{} 
	\end{subfigure}
	\caption{(a) A spatial slice of AdS with $A$ and $B$ chosen to be two intervals. The figure depicts the entanglement wedge of $AB$ (gray), the entanglement wedge of $C$ (green), the RT surface $\gamma_{AB}$ and the entanglement wedge cross section $\Gamma_{A:B}$. (b) A spatial slice of the proposed holographic dual to the canonically purified state $\ket{\sqrt{\rho_{AB}}}$. The RT surface for $A\cup A^*$ is given by a doubled copy of the entanglement wedge cross section.} 
\label{fig:EW} \end{figure}

Further, reflected entropy is of particular interest in holography, due to its relationship to the so-called \emph{entanglement wedge cross-section} \cite{Takayanagi:2017knl}, the minimal area surface $\Gamma_{A:B}$ in the dual bulk geometry that divides the entanglement wedge of $AB$ into two parts $a$ and $b$, one bounded by $A \cup \Gamma_{A:B}$, the other by $B \cup \Gamma_{A:B}$ as shown in \figref{fig:EW}. Using similar arguments to Lewkowycz-Maldacena (LM)\cite{Lewkowycz:2013nqa}, including a key assumption about the possible dominating saddles for the Euclidean gravitational path integral, it was shown in Ref.~\cite{Dutta:2019gen} that: 
\begin{equation}\label{eq:sr_equals_2ew} 
	S_R(A:B) = 2 EW(A:B) \equiv \min_{\Gamma_{A:B}} 2 \frac{\mathrm{Area}(\Gamma_{A:B})}{4G_N}~. 
\end{equation}
up to $\mathcal{O}(1)$ corrections in the small $G_N$ expansion, where $G_N$ is Newton's constant. Here, we have restricted to the time-independent setting although there is a simple generalization to time-dependent situations using a maximin formula \cite{Wall:2012uf,Akers:2019lzs}.

Why should one expect \Eqref{eq:sr_equals_2ew} to be true? An intuitive proof, also discussed in Ref.~\cite{Dutta:2019gen}, is that $\ket{\sqrt{\rho_{AB}}}$ is dual to two copies of the original entanglement wedge, glued along the RT surface as depicted in \figref{fig:EW}. The entropy $S_R(A:B)$ then equals $2 EW(A:B)$ by the RT formula and symmetry. 

The relationship \Eqref{eq:sr_equals_2ew} provides another example like the RT formula, relating a hard-to-compute information theoretic quantity to the area of a bulk surface. 
As such, it has already given new insight into holographic CFTs, for example demonstrating the large amount of tripartite entanglement required \cite{Akers:2019gcv}, a feat that could not be accomplished with the RT formula alone \cite{Cui:2018dyq}. Also because of \Eqref{eq:sr_equals_2ew}, the reflected entropy provides a novel measure of the connectedness of the entanglement wedge. If the entanglement wedge jumps discontinuously as a function of some parameter, the reflected entropy jumps accordingly by an amount $\mathcal{O}(1/G_N)$. This jump is in sharp contrast with the large-$N$ phase transitions of entanglement entropy that arise in AdS/CFT, which are continuous.

In this paper, we begin to explore the non-perturbative physics that describes such a discontinuity using toy models. That is, we compute the reflected entropy near such a phase transition, paying special attention to how non-perturbative effects smooth out the transition. In our random tensor network models, we will find a sharp, continuous transition analogous to the Page curve \cite{Page:1993df}, which reproduces the expected discontinuous transition in the limit that the bond dimension goes to infinity.

At the same time, we aim to test the assumptions that led to the relationship \Eqref{eq:sr_equals_2ew}. Such assumptions were recently shown to be delicate, even in the case of the RT formula, and leading order corrections to the RT formula were found in certain setups \cite{Akers:2020pmf}. Here we will find significant non-perturbative effects away from the EW phase transition, most notably when computing the $(m,n)$-R\'enyi reflected entropy (a generalization of reflected entropy which we define below in \secref{sec:setup}), where new saddles are important beyond the previously discussed, naive saddles. 
We will show that these effects go away as $n \rightarrow 1$, giving back \Eqref{eq:sr_equals_2ew}. 

We will start by working with a simple model of a single random tripartite tensor, which turns out to be perfectly tractable analytically. It models a three-boundary wormhole with horizon areas fixed to a small window \cite{Akers:2018fow,Dong:2018seb,Dong:2019piw} and a large interior. This results in an entanglement wedge (EW) phase transition as we tune the bond dimensions (the horizon areas). The ``Page curve'' for such a transition has the characteristic jump in reflected entropy shown in \figref{fig:analytic_phase_transition}. We then work with hyperbolic tensor networks, where the results are harder to control. Nevertheless, we extrapolate lessons from the simpler models (since they show many similarities with the larger networks) to compute the reflected entropy in these models as well.
\begin{figure}
			[h] \centering 
        	\includegraphics[width=0.95 \textwidth]{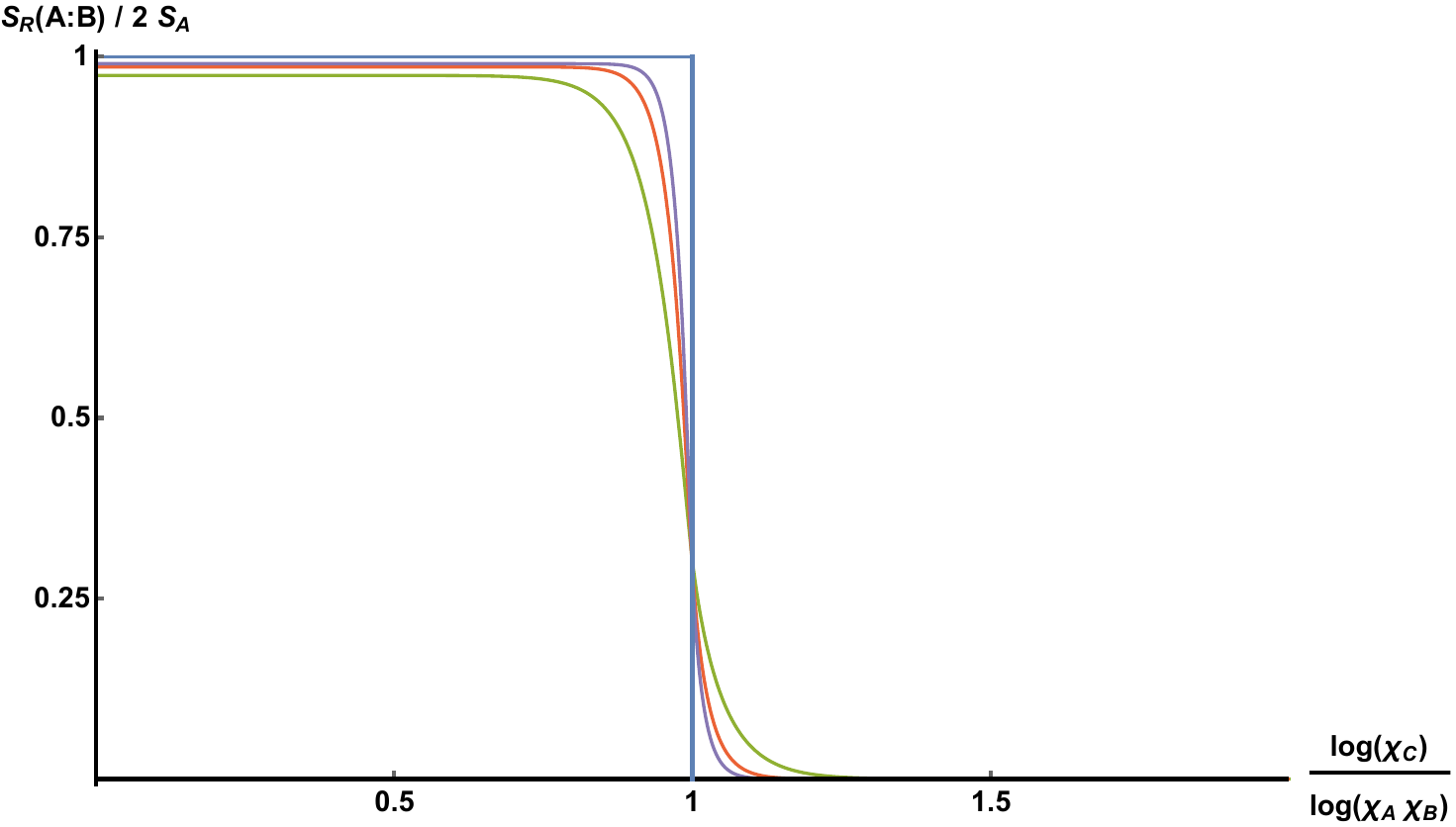}
			\caption{``Page curve'' of the reflected entropy, for a single tripartite random tensor with bond dimensions $\chi_A, \chi_B, \chi_C$. The blue curve corresponds to the infinite bond dimension limit, $\chi \to \infty$. The other curves correspond to large but finite bond dimension, correctly limiting to the blue curve as the bond dimension increases.} 
\label{fig:analytic_phase_transition} \end{figure}

In this paper we will also study the reflected entanglement spectrum, defined as the spectrum of $\rho_{AA^\star}$ as well as a two parameter family of generalizations called the $(m,n)$-R\'enyi reflected entropy: 
\begin{equation}
	S_R^{(m,n)}(A:B) \equiv - \frac{1}{n-1} \ln \tr(\rho^{(m)}_{AA^\star})^n, 
\end{equation}
where $\rho^{(m)}_{AA^\star}$ is the reduced density matrix on $A A^*$ associated to the state $\big| \rho_{AB}^{m/2} \big>$ after normalization. The $(m,n)$-R\'enyi reflected entropy for integer $n \geq 1$ and even integer $m \geq 2$ is computed at intermediate steps in the replica trick, and can also be used to extract the reflected entropy by analytic continuation: 
\begin{equation}\label{ordermn} 
	S_R(A:B) = \lim_{m \rightarrow 1} \lim_{n \rightarrow 1} S_R^{(m,n)}(A:B) 
\end{equation}
These detailed quantities give further insight into our results. 

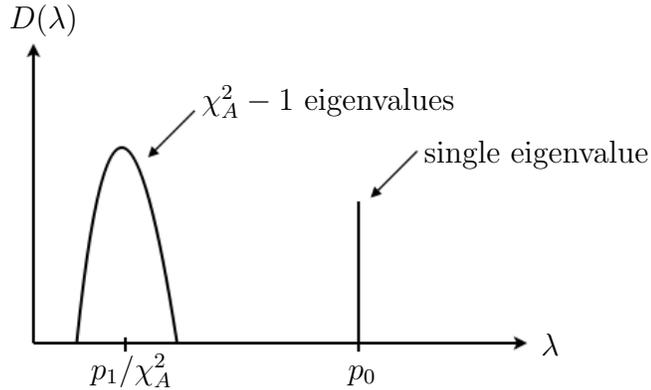
\begin{figure}
			[h] \centering 
			\begin{overpic}
				[width=0.5\linewidth]{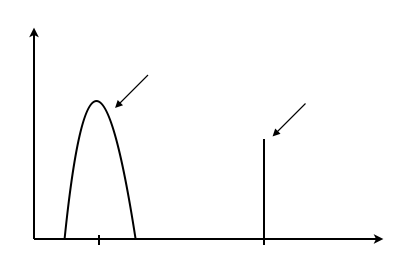} \put(4,62){$D(\lambda)$} \put(95,6){$\lambda$} \put(62,2){$p_0$} \put(18,2){$p_1/\chi^2_A$} \put(37,48){$\chi^2_A-1$ eigenvalues} \put(75,39){single eigenvalue} 
			\end{overpic}
			\caption{The reflected spectrum in the single tensor model has two superselection sectors corresponding to the disconnected and connected phase respectively.} 
		\label{fig:spectrum1} \end{figure}

We now summarize some of the main points in the paper: 
\begin{itemize}
		\item For the single tensor model, we prove some general results for the $(m,n)$-R\'enyi reflected entropy, including various continuity arguments as a function of $m$ which establishes independence of $m$ away from the phase transition. These will serve as checks on our later computations, and help motivate the prescription we use to analytically continue in $(m,n)$.  We then describe the phase diagram of dominant saddle points at integer $m/2,n$. See \secref{sec:ineq} and \secref{sec:phase_nm}.
		
	\item In the single tensor model, we use a Schwinger-Dyson re-summation method similar to \cite{Penington:2019kki,Shapourian:2020mkc} to compute the reflected spectrum. We make an assumption about the class of diagrams that dominates, and prove this to be correct in various limits. The reflected spectrum thus obtained has two main components: A single eigenvalue representing the disconnected EW phase, and a broad peak with $\exp(2 EW(A:B))$ eigenvalues representing the connected EW phase. These two contributions are always present irrespective of the relevant phase, but their respective weights move around as we tune the bond dimensions in the tensor network. They make dominant contributions to the reflected entropy in their respective phases. We show a cartoon of the reflected spectrum for the single tensor model in \figref{fig:spectrum1}. See \secref{sec:EE} and \secref{sec:SD}.
	
	\item Based on the above we argue that near the phase transition both sets of eigenvalues contribute. In the single tensor, the reflected entropy takes the form: 
	\begin{equation}
		S_R(A:B) = -p_0\ln p_0 -p_1\ln p_1 + p_1 \left( \ln \chi^2_A - \frac{\chi^2_A}{2\chi_B^2} \right) + O(\chi_{A,B}^{-2}) 
	\end{equation}
	where $\chi_X \gg 1$ are the bond dimensions for $X=A,B,C$ and the classical probabilities are: 
	\begin{equation}\label{eq:probRTN}
		p_0 = \frac{\chi_C}{\chi_A\chi_B} \;_2F_1 (1/2,-1/2;2; \frac{\chi_C}{\chi_A \chi_B})^2\,, \qquad p_1 + p_0 = 1,
	\end{equation}
	We argue that these results can be understood as arising from an effective description of the \emph{entanglement structure} for the canonically purified state as a superposition of tensor networks states. Here, one state is described by a doubled random tensor network, representing the connected phase, while the other tensor network is trivially factorized: 
	\begin{equation}\label{eq:CP_effective_1TN} 
		\includegraphics[width=0.8
		\textwidth]{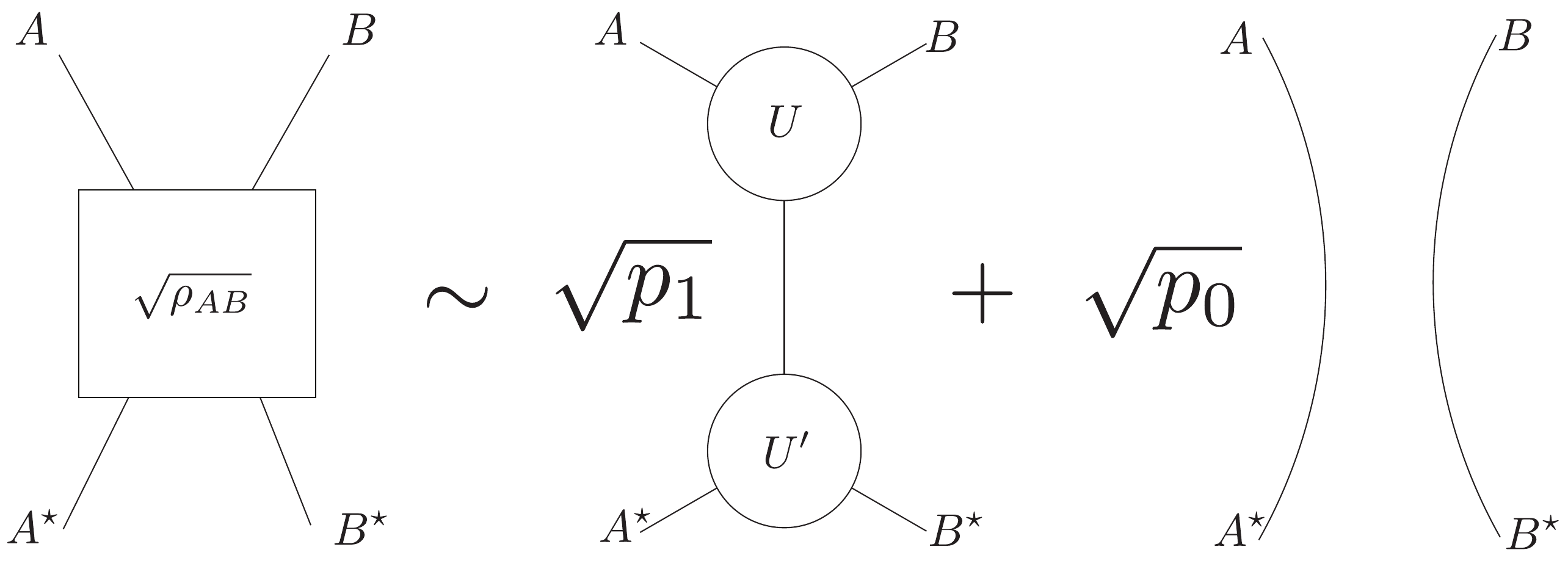} 
	\end{equation}
	The two terms are weighted by the associated probabilities $p_0$ and $p_1 = 1 - p_0$. The reduced density matrices on $AA^\star$ are approximately orthogonal, so these represent (approximate) superselection sectors\footnote{These are also sometimes referred to as $\alpha$-blocks or $\alpha$-sectors.} with an associated area operator: 
	\begin{equation}
		\mathcal{L}_{AA^\star} = \bigoplus_{s=0,1} S_{vN}(\zeta_s) 
	\end{equation}
	where $\zeta_1 \approx 1_{AA^\star}/\chi_A^2$ and $\zeta_0 = \left| 1_{A}/\chi_A \right> \left< 1_{A}/\chi_A \right|$.
	We also check the above results against numerics. See \secref{sec:spectrum} and \secref{sec:num}.
	
	\item For the hyperbolic networks we compute $S_R^{(m,n)}(A:B)$ at large bond dimension $\chi$ and at integer $m/2,n$ by looking for dominant saddle points in the effective spin model that computes the Haar random average. There are potentially many competing saddle points, so this turns out to be a difficult problem. However when the entanglement wedge is in the connected phase we conjecture, by analogy to the single tensor, that two types of saddles give the main contribution for all $n,m/2$. These phases were not accounted for in the original discussion of Ref.~\cite{Dutta:2019gen}. However, the difference to the naive approach does not show up in the reflected entropy $S_R$. The $(m,n)$-R\'enyi reflected entropy in the connected phase ($x>1/2$) can be written as: 
	\begin{equation}
	\label{ansrsr}
		S_R^{(m,n)}(A:B) =   {\rm min} \left. 
		\begin{cases}
			\frac{n}{n-1} I(A,B)
			
			\\
		\ln \chi \times	\min_{0 \leq \mathcal{A} \leq \mathcal{A}_x}\left[2 \mathcal{A} + \frac{4 n}{n-1} \ln 
			
			\cosh(\mathcal{A}_x- \mathcal{A}) \right] 
		\end{cases}
		\right\} 
	\end{equation}
	where $\mathcal{A}_x =\ln \frac{1 + \sqrt{1-x}}{\sqrt{x}}$ and $I(A,B) =2 \ln \chi \ln (1-x)/x$ is the mutual information.
	The first term arises from a disconnected saddle, that is still present in the connected phase, while the second term represents the connected saddles. The area $\mathcal{A}_x$ represents the cross-sectional area of the $AB$ entanglement wedge. While $\mathcal{A}$ represents the cross-sectional area of a ``pinched'' entanglement wedge which arises due to the non-trivial tension of the entanglement brane, when $n > 1$.
	
    We can again give an interpretation of this result using a superposition of networks, analogous to \Eqref{eq:CP_effective_1TN}:
	\begin{equation}\label{eq:CP_effective_hyperbolic} 
		\includegraphics[width=0.8
		\textwidth]{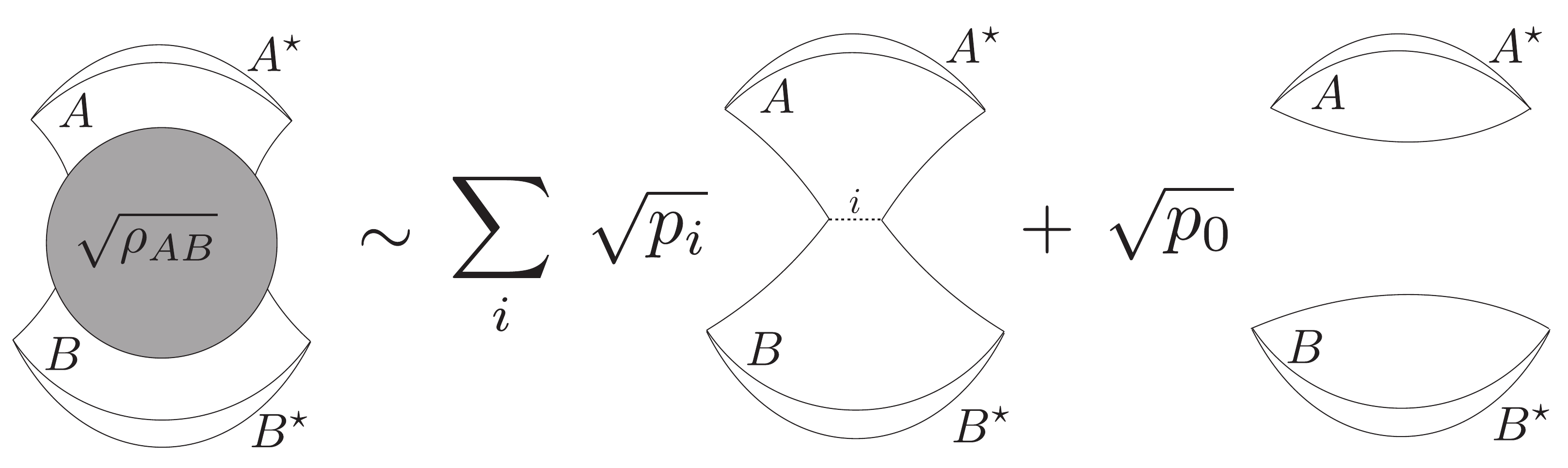} 
	\end{equation}
	where the state described by each side of this equation have the same $S_R^{(m,n)}$. The tensor networks on the right hand side are doubled and glued networks similar to those discussed in  \cite{Marolf:2019zoo}. The probabilities above can be read directly from \Eqref{ansrsr} and we will further expound upon the meaning of these pictures in \secref{sec:general_TN}.
	
	\item Our results will resolve the order of limits issue pointed out in Ref.~\cite{Kusuki:2019evw}. In particular if one applies the naive LM-like saddle point analysis in the limit $(m,n) \rightarrow (1,1)$ then the ``wrong'' phase can dominate if one takes the limit $m \rightarrow 1$ first and then $n \rightarrow 1$, leading to a formula different from \Eqref{eq:sr_equals_2ew}. This later formula fails several quantum information bounds, so one can rule out this order of limits on these grounds. However a better understanding of this issue and justification of \Eqref{ordermn} is desired. In fact we will show that no such order of limits issue occurs after we give a prescription for continuing the $(m,n)$-R\'enyi reflected entropies from $m \geq 2$ to $m=1$. We will give evidence for this prescription based on the single random tensor, and show how it is consistent with a more general set of QI bounds.

\end{itemize}

We start now with \secref{sec:setup}, where we review reflected entropy and random tensor networks and make a first pass at computing $S_R(A:B)$ in those networks, reviewing the usual large bond dimension, leading saddle computation, analogous to making the Lewkowycz-Maldacena (LM) assumption. We will then explain the deficiencies of this calculation. After this we move onto the results summarized above. 
\secref{sec:single_tensor} focuses on the single tensor case while \secref{sec:general_TN} moves to more complicated networks.
In \secref{sec:disc}, we discuss these results and future directions. Several appendices support our calculations, most notably Appendix~\ref{app:phase_proof} where we give a proof of the phase diagram for the single network case.

Our focus in this paper is on reflected entropy in random tensor networks. In a complementary paper \cite{westcoast}, we will explore the reflected entropy transition in the PSSY model of Jackiw-Teitelboim gravity and end-of-the-world branes.

\section{Setup and a first pass} \label{sec:setup}

In this section, we introduce the basic ingredients that go into our calculation. In \secref{sec:SR}, we review the definition of reflected entropy and its R\'enyi generalization. In \secref{sec:RTN}, we review the construction of random tensor networks which will be the setting for our main calculations. In \secref{sec:SR_RTN}, we review the replica trick calculation for reflected entropy, originally discussed in Ref.~\cite{Dutta:2019gen}. While their discussion was in the gravitational context, we rephrase it in terms of the random tensor network which have analogous features. This then leads us to \secref{sec:saddle}, where we attempt to calculate the reflected entropy using a naive saddle point approximation. This works in certain regimes, but isn't a completely satisfactory solution in the entire parameter space. We point out various issues with the calculation, which will then be resolved in \secref{sec:single_tensor}.

\subsection{Reflected entropy} \label{sec:SR}

The reflected entropy $S_R(A:B)$ is a function defined for a density matrix $\rho_{AB}$ on a bipartite quantum system $AB$. One first considers the \emph{canonical purification} of $\rho_{AB}$ on a doubled Hilbert space 
\begin{equation}
	\ket{\sqrt{\rho_{AB}}} \in \mathrm{End}(\mathcal{H}_A) \otimes \mathrm{End}(\mathcal{H}_B) = (\mathcal{H}_A \otimes \mathcal{H}^*_A) \otimes (\mathcal{H}_B \otimes \mathcal{H}^*_B)~, 
\end{equation}
where the space of linear maps $\mathrm{End}(\mathcal{H}_A)$ acting on $\mathcal{H}_A$ itself forms a Hilbert space with inner product $\braket{X|Y} = \tr_A(X^\dagger Y)$. In general $\mathrm{End}(\mathcal{H})$ is isomorphic to the doubled copy $\mathcal{H} \otimes \mathcal{H}^*$. In other words, define $\ket{\sqrt{\rho_{AB}}}$ by finding the unique positive matrix square root of $\rho_{AB}$, regarding the result as a state in $\mathrm{End}(\mathcal{H}_{A} \otimes \mathcal{H}_{B})$.

The reflected entropy is then defined as:
\begin{equation}
	S_R(A:B) = S(AA^*)_{\ket{\sqrt{\rho_{AB}}}} = S_{vN}(\rho_{AA^*})~, \qquad \rho_{AA^*} = {\rm Tr}_{BB^\star} \left| \sqrt{\rho_{AB}} \right> \left< \sqrt{\rho_{AB}} \right|, 
\end{equation}
where $S_{vN}$ is the von-Neumann entropy. 

For our purposes, it will be useful to consider a two parameter R\'enyi generalization, based on the following state:\footnote{Note that $\ket{\psi^{(m)}}$ is no longer a purification of $\rho_{AB}$.} 
\begin{equation}\label{eq:psim} 
	\left| \psi^{(m)} \right> = ({\rm Tr} \rho_{AB}^m)^{-1/2} \left| \rho_{AB}^{m/2} \right>. 
\end{equation}
The R\'enyi generalization is then given by 
\begin{equation}\label{eq:Renyimn} 
	S_R^{(m,n)}(A:B) = - \frac{1}{n-1} \ln {\rm Tr} (\rho_{AA^\star}^{(m)})^n \, , \qquad \rho_{AA^*}^{(m)} = \frac{1}{ {\rm Tr} \rho_{AB}^m}{\rm Tr}_{BB^\star} \left| \rho_{AB}^{m/2} \right> \left< \rho_{AB}^{m/2} \right|, 
\end{equation}
for $m \geq 0, n \geq 0$. We will refer to this as the $(m,n)$-R\'enyi reflected entropies. 

The reflected entropy satisfies various properties that make it a useful measure of correlation. Of these, an important condition that will be useful for us is continuity \cite{Akers:2019gcv}, i.e., 
\begin{equation}
	|S_R(A:B)_{\rho_{AB}}-S_R(A:B)_{\sigma_{AB}}|\leq 4 \sqrt{2 T_{AB}} \log \min(d_A,d_B)-2 \sqrt{2T_{AB}}\log(T_{AB}), 
\end{equation}
where $T_{AB}$ is the trace distance between the reduced density matrices $\rho_{AB}$ and $\sigma_{AB}$ and $d_{A,B}$ is the dimension of the Hilbert space of $A,B$ respectively. Another useful property is that the reflected entropy is upper and lower bounded by 
\begin{equation}
	I(A:B) \leq S_R(A:B) \leq 2 \min(S(A),S(B)), 
\end{equation}
where $I(A:B)$ is the mutual information between $A$ and $B$. 

\subsection{Random tensor networks} \label{sec:RTN}

Here we quickly review random tensor networks (RTNs). For more detail, refer to Ref.~\cite{Hayden:2016cfa}. 

Consider a graph $G=\{V,E\}$. For each vertex $x \in V$ we assign a rank-$k$ tensor $T_{x, \mu_1 \mu_2 \cdots \mu_k}$,where the indices $\mu_i$ label the edges connecting the point $x$. Thinking of each connecting edge as a Hilbert space $\mathcal{H}_i$ spanned by basis vectors $\ket{\mu_i}$, the tensor $T_x$ defines a state 
\begin{equation}
	\ket{T_x} = T_{x, \mu_1 \mu_2 \cdots \mu_k} \ket{\mu_1}\ket{\mu_2}\cdots \ket{\mu_k} 
\end{equation}
on the product Hilbert space of each edge $\bigotimes_i \mathcal{H}_i$. For simplicity, the bond dimensions $\chi$ of the individual edge Hilbert spaces are chosen to be the same. To define a tensor network, we contract all the adjacent vertices among their shared common edges 
\begin{equation}
	\left( \bigotimes_{\{xy\}\in E} \bra{xy} \right) \left( \bigotimes_{x\in V} \ket{T_x} \right) 
\end{equation}
Here $\ket{xy} = \delta_{\mu_i\mu'_j}\ket{\mu_i}\ket{\mu'_j}$ is an un-normalized maximally entangled state defined on the doubled Hilbert space of the edge $(xy)$ where the $i$-th index of tensor $x$ is contracted with the $j$-th index of tensor $y$.
\begin{figure}
	\centering 
	\begin{subfigure}
		[b]{0.4
		\textwidth} 
		\includegraphics[width=
		\textwidth]{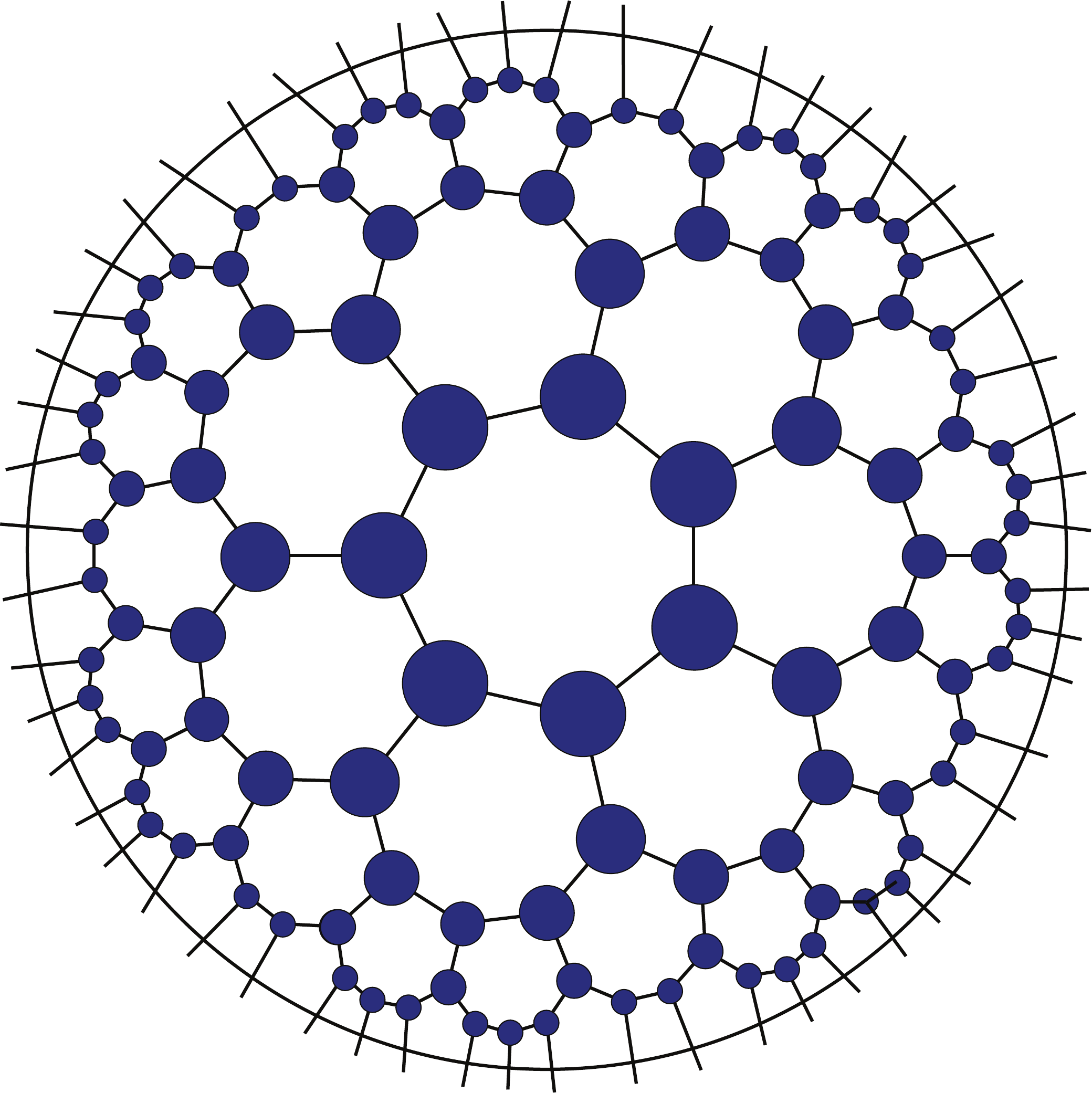} \subcaption{} 
	\end{subfigure}
	\hspace{5mm} 
	\begin{subfigure}
		[b]{0.4
		\textwidth} 
		\includegraphics[width=
		\textwidth]{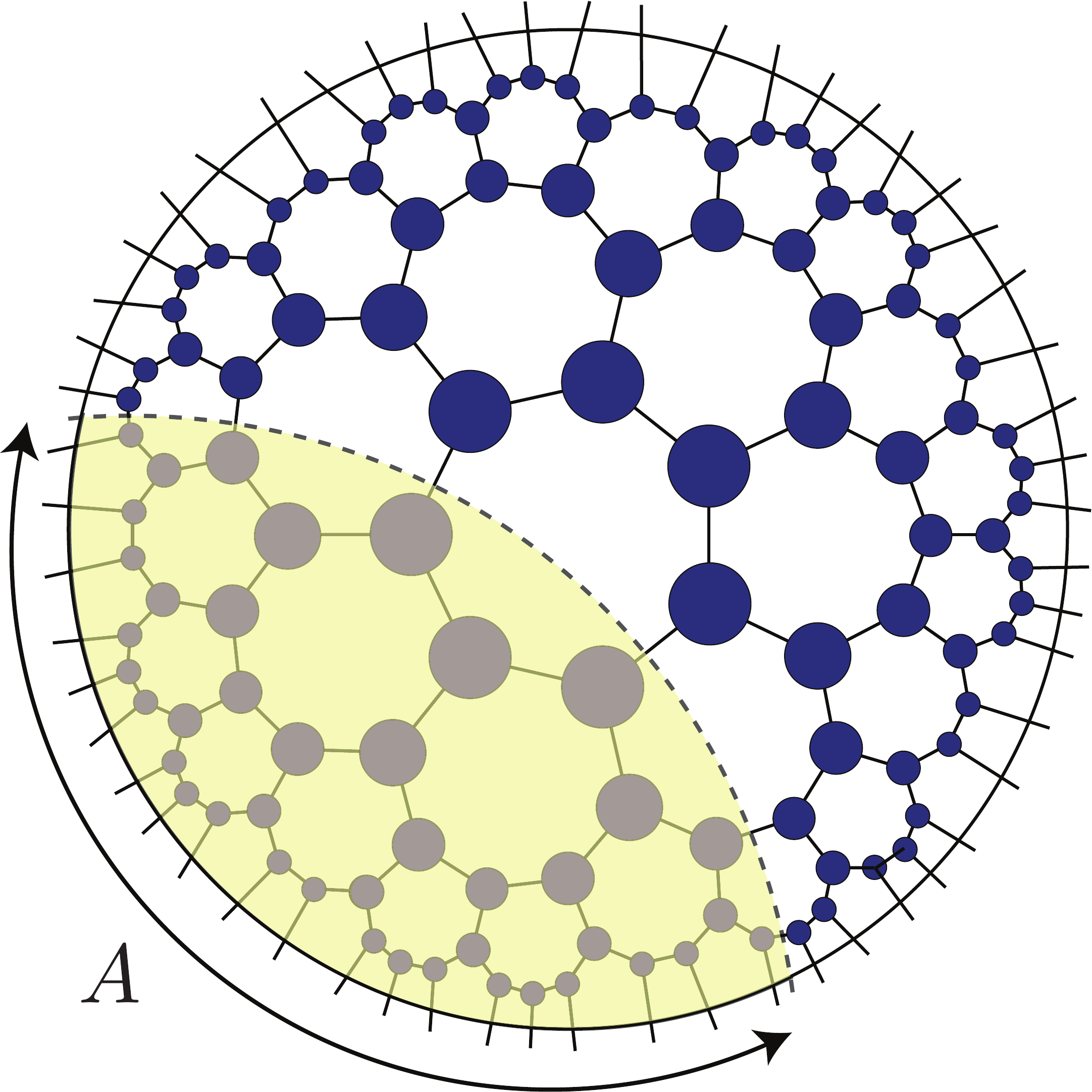} \subcaption{} 
	\end{subfigure}
	\caption{(a) We tile the hyperbolic disk isometrically by a graph. Each blue circle represents a random tensor and the connecting edges indicates contraction of the tensors. This tensor network defines a state on the Hilbert space spanned by the dangling legs on the boundary. (b) The Ising domain wall that arises in the calculation of R\'enyi entropies of region $A$ in the RTN. Minimizing the length of the domain wall gives rise to an entropy formula that corresponds to the RT formula in AdS/CFT.} 
\label{fig:TN} \end{figure}

While we can work with a general graph, we will often be interested in models of AdS/CFT, where it is natural to pick a triangulation of two dimensional hyperbolic space, thought to represent a fixed time slice of the $AdS_3$ spacetime, see Fig. \ref{fig:TN}. At the boundary of the network (at a cutoff region near the boundary of hyperbolic space) we leave the edges un-contracted and think of the resulting network as describing a pure state in a Hilbert space associated to the dangling legs $\mathcal{H}_\partial$. 
\begin{equation}
	\ket{\Psi} = \left( \bigotimes_{\{xy\}\in E\setminus \partial} \bra{xy} \right) \left( \bigotimes_{x\in V} \ket{T_x} \right) \in \mathcal{H}_\partial 
\end{equation}

For RTNs, we demand the tensor $\left| T_x \right>$ is sampled from a uniform random ensemble on the tensor Hilbert space. This can be achieved by acting with a unitary matrix $U_x$ picked from a Haar-random measure on some fixed anchor state $\ket{0_x}$. Then the average over such a measure can be readily done using Schur's lemma \cite{harrow2013church}: 
\begin{equation}
	\overline{(\ket{V_x}\bra{V_x})^{\otimes m}} = \int [DU_x] \left(U_x \ket{0_x}\bra{0_x} U_x^{\dagger} \right)^{\otimes m} \propto
\sum_{g_x \in S_m} g_x 
\end{equation}
where $g_x$ is a element of the symmetry group $S_m$ whose action is to permute the contraction edges of the $m$-replicas. The state $\ket{\Psi}$ is unnormalized. Normalization can be achieved ``on average'' by dividing by the average of the norm of $\ket{\Psi}$. At large bond dimension this procedure is sufficient for our purposes.

We now consider a factorization of the boundary Hilbert space as $\mathcal{H}_{\partial} = \mathcal{H}_A \otimes \mathcal{H}_{\bar{A}}$. The $m$-th R\'enyi entropy of subregion $A$ can be written as the expectation value of a twist operator $\Sigma_A(\tau_m)$ which acts on the tensor product of $m$ copies of $\mathcal{H}_A$. The twist operator cyclically permutes the different copies of $\mathcal{H}_A$ and we label this operator $\tau_m \in S_m$, the cyclic permutation element in the symmetric group $S_m$. After averaging and normalizing as described above one finds at large bond dimension: 
\begin{equation}\label{eq:IsingAction} 
	Z_m(A) \equiv \overline{{\rm Tr}_A \rho_{A}^m } = \overline{\left< \Psi \right|^{\otimes m} \Sigma_A(\tau_m) \left| \Psi \right>^{\otimes m}} = \sum_{\{ g_x \in S_m\}_{b.c.}} \exp \left[ - \ln \chi \sum_{\{xy\} \in E} d(g_x,g_y) \right] 
\end{equation}
which is a classical Ising-like model with nearest neighbor interaction where each ``spin" takes value in the symmetry group $S_m$. The interaction strength is given by the Cayley distance 
\begin{equation}
	d(g,h) = m - \# (gh^{-1}) 
\end{equation}
where $\#(\cdot)$ is a function that counts the total number of cycles of a group element, including trivial cycles that map a given element to itself. We summarize various useful results for the symmetric group in Appendix~\ref{app:perm}. The boundary conditions on the sum $\{ g_x \in S_m\}_{b.c.}$ are dictated by the presence of the twist operators, i.e. we impose $g_x = \tau_m$ on boundary region $x\in A \subseteq \partial$ and $g_x=e$ on $x \in \bar{A} = \partial \backslash A$. In the case of large bond dimension, the Ising model is in its low temperature limit and is localized to its ground state. The field configuration in this limit is given by domains of group elements $\tau_m$ and $e$, separated by a domain wall as shown in \figref{fig:TN}. Minimizing the energy of this domain wall, we recover the usual RT formula for entanglement entropy. At large bond dimension we do not need to worry about taking the logarithm since: 
\begin{align}
	\overline{S_m(\rho_A)} = \frac{1}{1-m} \overline{\left( \ln \frac{{\rm Tr} \rho^m_A}{({\rm Tr} \rho_A)^m}\right)} \approx \frac{1}{1-m} \ln \frac{\overline{{\rm Tr} \rho^m_A }}{\overline{({\rm Tr} \rho_A)^m}}~. 
\end{align}

Note that the above calculation gives the same results for the R\'enyi entropies $S_m$ for all $m \geq 0$. The entanglement spectrum for the RTN is therefore flat, which is clearly not true for generic holographic states. 
Instead, RTN have been interpreted as models of fixed-area states in such theories \cite{Akers:2018fow,Dong:2018seb,Dong:2019piw}. Generic holographic states can be obtained as a superposition of fixed-area states. Later we will see such superpositions naturally arise from the canonical purification without adding this structure in by hand. 

\subsection{Reflected entropy for RTNs} \label{sec:SR_RTN}

To set up the calculation for reflected entropy, we now divide the boundary into three different regions $\{A,B,C\}$. As a simple example, we take both $A$ and $B$ to be connected intervals and trace out $C= \overline{A \cup B}$ to get the reduced density matrix $\rho_{AB}$. The first step is to construct the canonical purification in the doubled Hilbert space 
\begin{equation}
	\rho_{AB} \to \ket{\sqrt{\rho_{AB}}}\bra{\sqrt{\rho_{AB}}} \equiv \rho_{ABA^*B^*} 
\end{equation}
where the $A^*~(\text{respectively } B^*)$ is the canonical purified counterpart of the original region. We will achieve this by first computing the $(m,n)$-R\'enyi reflected entropies defined in \Eqref{eq:Renyimn}. This is possible using a similar Ising-like model, for integer $n \geq 1$ and even integer $m \geq 2$. We introduce two group elements $g_A \in S_{mn}$ and $g_B \in S_{mn}$ below that act as twist operators needed to compute the $(m,n)$-R\'enyi reflected entropy. Using the same procedure as above, including an appropriate normalization that works at large bond dimension, we find: 
\begin{align}
	Z_{m,n}(A,B) &\equiv e^{- (n-1) \overline{S^{(m,n)}_R(A:B)}} = \frac{\overline{ \langle \Psi |^{\otimes nm} \Sigma_A(g_A) \Sigma_B(g_B) | \Psi \rangle^{\otimes nm}}}{ \left( \overline{\langle \Psi |^{\otimes m} \Sigma_{AB}(\tau_m) | \Psi \rangle^{\otimes m}} \right)^n } \\
	&= (Z_m(AB))^{-n} \sum_{\{ g_x \in S_{mn}\}_{b.c.'} } \exp \left[ - \ln \chi \sum_{\{xy\} \in E} d(g_x,g_y) \right] 
\end{align}
where the new boundary condition on the sum, denoted $b.c.'$, imposes $g_x = g_A $ for $x \in A$, $g_x = g_B $ for $x \in B$ and $g_x = e$ for $x \in C$.

We now describe the group elements of $S_{mn}$. For a detailed derivation of these twist operators see Ref.~\cite{Dutta:2019gen}. Denoting a specific replica by $(\alpha,\beta)$, where $\alpha= 1,\cdots,m$ and $\beta= 1,\cdots,n$, we represent the elements in $S_{mn}$ in the following notation: 
\begin{equation}
	\ell^{[i]_m} \in S_{mn} \qquad \ell^{[i]_m}\cdot (\alpha,\beta) =
	\begin{cases}
		(\alpha, \ell \cdot \beta) & \alpha = i \\
		(\alpha,\beta) & \alpha \neq i 
	\end{cases}
	\qquad \ell \in S_n 
\end{equation}
and similarly: 
\begin{equation}
	\kappa^{[j]_n} \in S_{mn} \qquad \kappa^{[j]_n}\cdot (\alpha,\beta) =
	\begin{cases}
		(\kappa\cdot \alpha, \beta) & \beta = j \\
		(\alpha,\beta) & \beta \neq j 
	\end{cases}
	\qquad \kappa \in S_m 
\end{equation}
The group elements that permute the regions $A,B$ can thus be written as 
\begin{equation}
	g_B = \prod_{j=1}^{n} \tau_m^{[j]_n}\,, \qquad g_A = (\tau_n^{-1})^{[m/2]_m} \left( \prod_{j=1}^{n} \tau_m^{[j]_n} \right) \tau_n^{[0]_m} 
\end{equation}
where $\tau_m, \tau_n$ are the respective full cyclic twist operators on $S_m, S_n$. If we define 
\begin{equation}
	\gamma_\tau = \prod_{i=1}^{m/2} \tau_n^{[i]_m} 
\end{equation}
as the group element corresponding to the full n-cyclic permutation in the lower half replicas $\beta>m/2$, we can express $g_A$ as a conjugation of $g_B$ such that    
\begin{equation}
	g_A = \gamma_\tau g_B \gamma_\tau^{-1}, 
\end{equation}
where it is useful to note that conjugation relabels the elements while preserving the cycle structure.
\begin{figure}
	\centering 
	\includegraphics[width=0.9
	\textwidth]{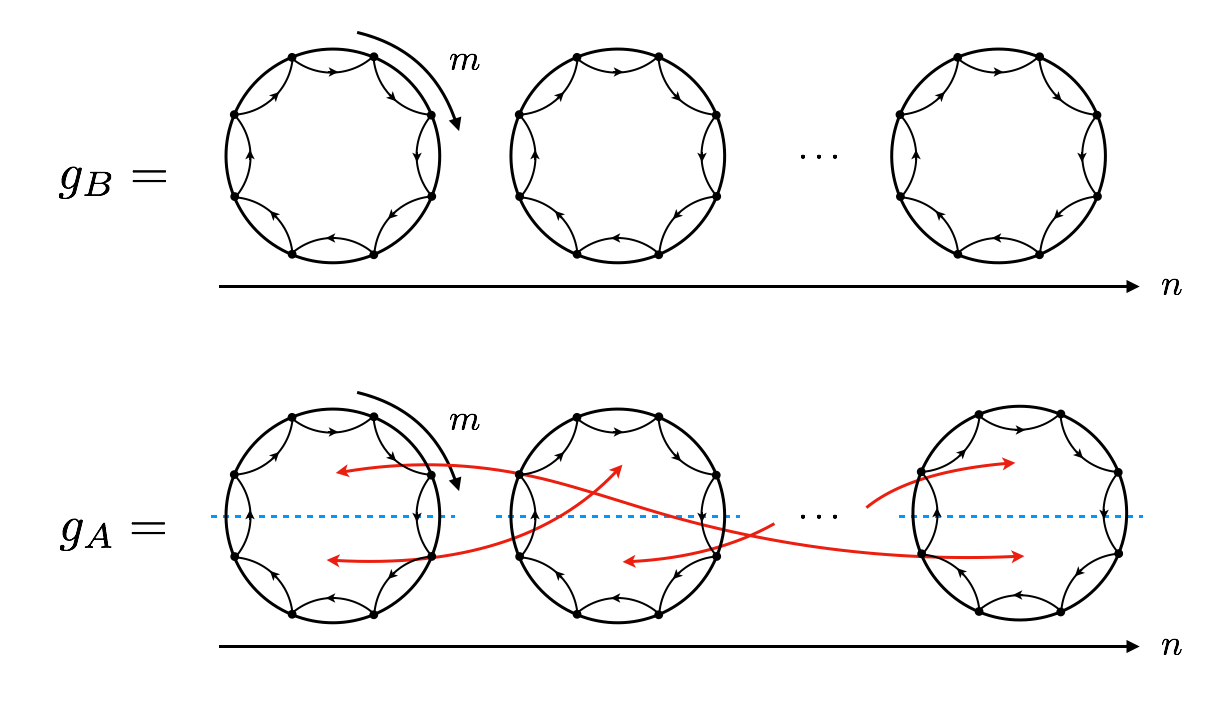} 
	\caption{A graphical representation of $g_A$ and $g_B$. The individual circles each represents the $m$ replicas of the original tensor and each of the circles is further replicated $n$ times. Going clockwise in each circle increases the $m$ replica number and going to the next circle on the right increases the $n$ replica number. A cycle of the permutation is represented by a closed directed loop. The element $g_A$ can be thought of as cutting open the $m$-circles of $g_B$ from the middle, shifting the bottom half cyclically in the $n$ direction (as the red arrows) and then gluing them back together.} 
\label{fig:g_A_g_B} \end{figure}

Since these definitions are a bit complicated to unpack, we introduce a graphical visualization for elements in $S_{mn}$ in \figref{fig:g_A_g_B}. We represent each column of replicas with the same $n$ index by a circle. We will sometimes refer to this as a \emph{puddle}. The same circle is then repeated $n$ times to make the full $mn$ copies for $g$ to act on. Each individual replica (including both a bra and a ket) corresponds to a dot on the circles - there will be $m$ dots on the circle. The action of $g\in S_{mn}$ is represented by drawing directed closed loops corresponding to the cycle decomposition of $g$. For each circle we denote the upper half to be the $\beta=1,\cdots,m/2$ replicas of the column and the lower half to be $\beta=m/2+1,\cdots,m$. The action of conjugation by $\gamma_\tau$ is simply cutting open each circle along its middle line, making a cyclic permutation on the lower halves, and then gluing back together.

The intuition behind these specific elements is that $g_B$ cyclically permutes the elements in each puddle, computing ${\rm Tr} \rho_{AB}^m$. These puddles have a specific $\mathbb{Z}_2$ symmetry that exchanges the lower half and upper half. It is useful to think of this symmetry as a time reflection symmetry. Cutting open the traces at fixed points of this $\mathbb{Z}_2$ action will give the canonical purification $\left| \rho_{AB}^{m/2} \right>$. Then relative to $g_B$ we introduce into $g_A$ two $n$-twist operators, $\tau_n^{[0]_m}$ and $(\tau_n^{-1})^{[m/2]_m}$, that compute the $n$-R\'enyi entropies and live at the fixed points of this $\mathbb{Z}_2$. The two twist operators correspond to $A$ and $A^\star$ respectively, so putting this together it is clear that we are computing ${\rm Tr} (\rho^{(m)}_{AA^\star})^n$.

\subsection{The naive saddles} \label{sec:saddle}

We now consider two natural saddle points, i.e. domain wall configurations for the Ising problem that arose in computing the $(m,n)$-R\'enyi reflected entropy. These saddles are the tensor network analogs of the saddle points that were considered in Ref.~\cite{Dutta:2019gen} and used to prove the correspondence in \Eqref{eq:sr_equals_2ew}. These were also interpreted as arising from the gluing construction of Refs.~\cite{Engelhardt:2017aux,Engelhardt:2018kcs}, which has a natural analog in the tensor network case \cite{Marolf:2019zoo}. \footnote{This was one of several methods suggested for proving the correspondence in Ref.~\cite{Dutta:2019gen}. The other method is studying the modular flowed correlators, which doesn't obviously extend to the tensor network case.} In particular, these saddles are analogs of the LM saddles \cite{Lewkowycz:2013nqa} used to compute R\'enyi entropies in AdS/CFT. 

With the above results at hand, we are now ready to make a first-pass calculation of the reflected entropy in RTN using the saddle-point approximation. We remind our readers that the presentation in this subsection will be schematic and brief, merely touching on the issues that arise during the calculation. We return to those issues with a more detailed treatment in \secref{sec:single_tensor} after our study of the one-site tensor model.
\begin{figure}
	\centering
	
	\begin{subfigure}
		[t]{.4\linewidth} 
		\caption{The naive disconnected phase} 
		\begin{overpic}
			[width=1
			\textwidth]{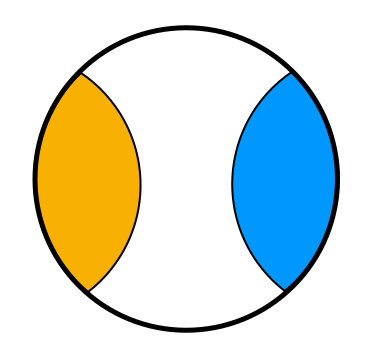} \put(20,45){\large$g_A$} \put(75,45){\large$g_B$} 
		\end{overpic}
		
	\end{subfigure}
	\begin{subfigure}
		[t]{.4\linewidth} 
		\caption{The naive connected phase} 
		\begin{overpic}
			[width=1
			\textwidth]{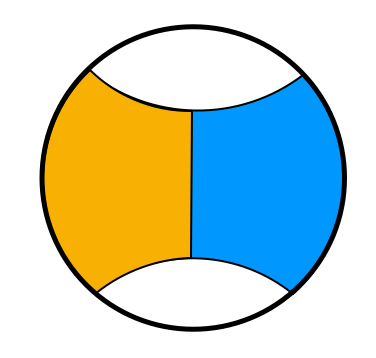} \put(25,45){\large$g_A$} \put(70,45){\large$g_B$} 
		\end{overpic}
		
	\end{subfigure}
	\caption{Two naive solutions for the bulk spin model. The white region indicates the identity permutation $e$.} 
\label{fig:naive_bulk_sol} \end{figure}

Imposing the boundary elements to be $g_A$ on region $A$ and $g_B$ on region $B$ and allowing the corresponding domains to propagate into the bulk, we immediately identify two possible geometrical saddles (see \figref{fig:naive_bulk_sol}). The first saddle is characterized by each region bounded by its minimal homology surface (we call this the \textit{naive disconnected saddle}). The second saddle features a non-trivial entanglement wedge, with the minimal wedge cross-section marked as the domain wall for $g_A\leftrightarrow g_B$ (we call this the \emph{naive connected saddle}). The tensions of these domain walls are found from the Cayley distances: 
\begin{align}
	d(g_A, e) = d(g_B, e) = n (m-1), \qquad d(g_A, g_B) = 2 (n-1) 
\end{align}
Assuming the naive disconnected saddle dominates and doing a simple analytic continuation to $(m,n) \rightarrow (1,1)$ one finds that the reflected entropy vanishes, while a similar procedure for the naive connected saddles gives a reflected entropy proportional to the EW cross-section. It is then expected that as we vary the size of regions $A,B$, one of the two solutions gains dominance over the other, with a phase transition at some critical size. Furthermore, one can verify that the transition point occurs right at the point when the RT surface of $A\cup B$ jumps. Hence the conjecture ``$S_R(A:B)=2EW(A:B)$" seems to be true in RTN.

We do expect that our construction of the two semi-classical saddles works well when the bulk is far from the transition point, where the sum in the partition function is strongly dominated by a single saddle. However there are several issues that arise from this proposal:
\begin{itemize}
	\item This construction fails to correctly account for the $n$-R\'enyi reflected entropies at $m=1$ even away from the phase transition. To see this consider the difference in free energies of the two saddles: 
	\begin{align}
		\Delta F &= F_{disconn.} - F_{conn.} = \ln \chi \times ( n(m-1) \ell_{mi} - 2 (n-1) \ell_{ew} )\, \\
		\ell_{mi} &= 2\ln \frac{1-x}{x} \, \qquad \ell_{ew} = \ln \frac{ 1+ \sqrt{1-x}}{\sqrt{x}} 
	\end{align}
	where $x$ is the conformal cross ratio of the end points of the boundary intervals and $x<1/2$ in the connected phase. Here $\ell_{mi}$ is the difference of lengths for the RT surfaces that compute the mutual information in the connected phase. Also $\ell_{ew}$ is the length of the cross section. We have approximated the discrete network by a continuum in order to write simple formulas. If we treat these saddles seriously for all $(m,n)$ then we find the R\'enyi entropies for $n>1$ at $m=1$ are dominated (have smaller free energy) by the disconnected phase and thus vanish to leading order in $\ln \chi$. Taking the limit $n \rightarrow 1$ after setting $m=1$ gives $0$ for the reflected entropy even in the connected phase. In particular this indicates an order of limits issue that was first pointed out in Ref.~\cite{Kusuki:2019evw}. This answer is not obviously inconsistent, however we will argue that it is nevertheless incorrect. 
	
	\item While the two bulk solutions are legitimate saddles at integer $(m/2,n)$ and each carries physical significance, it is not obvious that other semi-classical bulk configurations do not dominate over these. In fact as we will see Section 4, there will be a new class of solutions featuring a new group element $X\in S_{mn}$ that can dominate the partition function. We will see how this new element interacts with the existing solutions, and we will argue that phases involving $X$ are the correct ones for computing the R\'enyi entropies at $m=1$.\footnote{While this is somewhat similar to the new ``replica symmetry breaking'' group elements that are important for negativity computations \cite{Shapourian:2020mkc,Dong:2021oad,Vardhan:2021npf}, the details differ significantly. } These solutions do not have the same order of limits issue above.
	
	\item At the phase transition there is considerable uncertainty with how to proceed. So far the calculations we discussed were based on some fixed integer replica number $n,m>1$. To obtain the reflected entropy we need to analytically continue $n,m\to 1$. While the solutions we obtained above work w independently, we need to include both of them to make a phase transition. However simply summing the two saddles is known to lead to incorrect results \cite{Fischetti:2014zja}. The simplest issue that arises is from the normalization. Sending $ n \rightarrow 1$ we expect: 
	\begin{equation}
		1 = \tr \rho_{AA^*} = \tr \left( \tr_{BB^*} \ket{\sqrt{\rho_{AB}}}\bra{\sqrt{\rho_{AB}}} \right) 
	\end{equation}
	However since we view the two bulk solutions as independent from one another, each of them contributes unity to the trace exactly at the phase transition. Adding those up we obtain a paradoxical result $\tr \rho_{AA^*} = 2$. It is clear that we should somehow treat the two seemingly different saddles as limits of a more general class of solution.
	
	\item Indeed, around the phase transition point we expect there will be other previously sub-leading saddles that will smooth out the sharp transition. A notable example for such effect is the calculation of entanglement entropy for two disjoint regions \cite{Dong:2020iod,Marolf:2020vsi} (See also Ref.~\cite{Penington:2019kki}), where a summation over a larger class of group elements called the non-crossing permutations is performed. In our case the problem is harder due to the appearance of the domain wall $g_A\leftrightarrow g_B$ in our bulk solution. However when we include a summation over a larger set of saddles, we also expect to find a smoothing out of the reflected entropy phase transition. 
\end{itemize}

The above issues suggest we need to sum over a larger set of non-trivial bulk saddles. However, the bulk geometry is complicated, and it is not a priori obvious how to construct these solutions. 
To make our life easier and also to improve our understanding of the underlying mechanics of the RTN calculation, we will warm up in the next section with a toy model consisting of only a single random tensor. Being a simpler model, we are able to perform an analytic computation of the full entanglement spectrum of the canonical purified density matrix. We will see that this simple model exhibits interesting behaviors that capture and resolve many of the issues we listed above. We will return to the more complicated network states in \secref{sec:general_TN}.

\section{Single random tensor}\label{sec:single_tensor}
\begin{figure}
	[h] \centering 
	\includegraphics[width=0.4
	\textwidth]{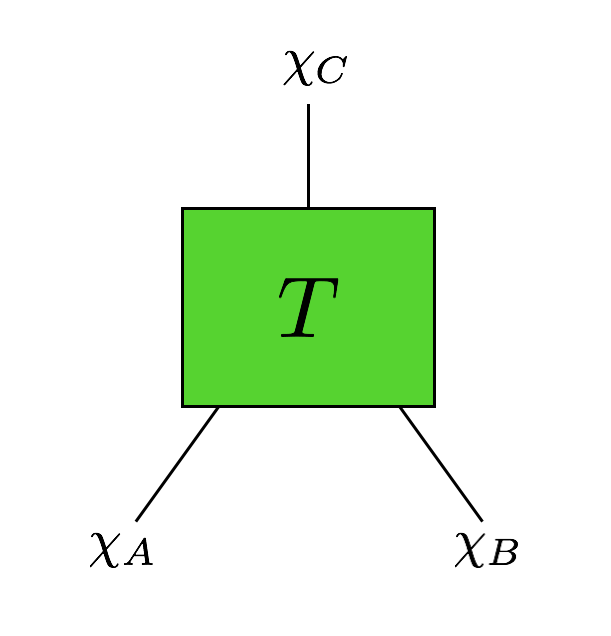} 
	\includegraphics[width=0.5
	\textwidth]{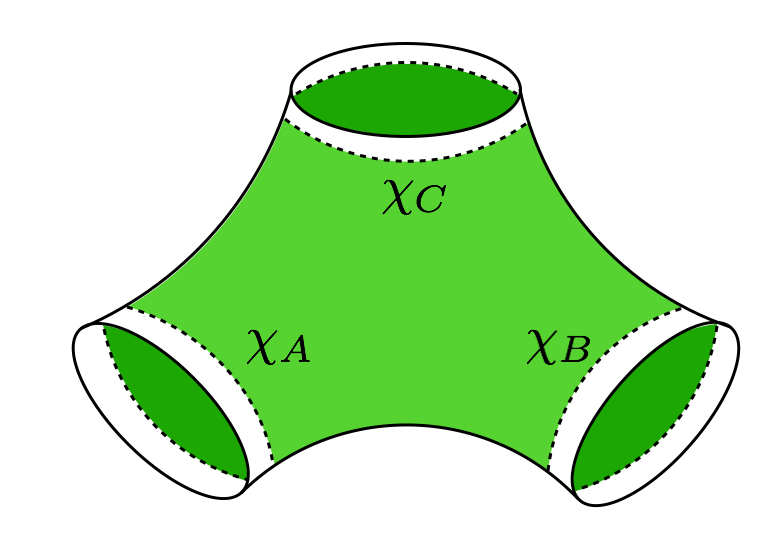} 
	\caption{(left) A state on the Hilbert space $\mathcal{H}_{ABC}$ is prepared by a single tripartite random tensor with bond dimensions $\chi_A$, $\chi_B$ and $\chi_C$. (right) The three-boundary wormhole solution, which is modeled by the single random tensor when its horizon areas are fixed to a small window. The horizon areas of the wormhole corresponds to the bond dimensions of the single tensor.} 
\label{fig:single} \end{figure}

In order to resolve the issues we faced in the previous section, we will consider a much simpler version of the problem. Consider a single tripartite random tensor $T$ with bond dimensions $\chi_A$, $\chi_B$ and $\chi_C$ that prepares a state $\ket{\psi}_{ABC}$ on the Hilbert space $\mathcal{H}_{ABC}$. This setup can be thought of as a toy model for the holographic setting of a multiboundary wormhole with three asymptotic boundaries, where the area of the mouths $\mathcal{A}_{A,B,C}$ are related to the bond dimensions $\chi_{A,B,C}$ as 
\begin{align}
	\log \chi_{A,B,C} &= \frac{\mathcal{A}_{A,B,C}}{4G_N} 
\end{align}
by the standard relation between the bulk spacetime and the tensor network that discretizes it \cite{Balasubramanian:2014hda,Hayden:2016cfa}, as shown in \figref{fig:single}. More so, when computing the R\'enyi entropy, there is an exact correspondence (including non-perturbative effects) between the single random tensor and the multiboundary wormhole when the horizon areas are fixed to a parametrically small window \cite{Akers:2018fow,Dong:2018seb,Dong:2019piw}.\footnote{Note that by non-perturbative effects we mean effects of $O(e^{-1/G_N})\sim \chi^{k}$ for some $k$, i.e., power law corrections in the bond dimension.} 

Note that in order for the connection between the single tensor and the three boundary wormhole to be sharp, we require that the interior region is large enough that there are no other competing entanglement wedge cross sections except the horizons of $A$ and $B$ since there is no analog of such surfaces in the single tensor. It seems plausible that this can be arranged for by adding sufficient matter in the interior to support a long wormhole like in Ref.~\cite{Shenker:2013yza}.

This setup is much more tractable than the general tensor network and allows us to understand phase transitions in the reflected entropy in detail. We will use the resolvent trick described in \cite{Penington:2019kki} to compute the exact entanglement spectrum of $\rho_{AA^*}$ in the limit of large bond dimension. This allows us to compute the reflected entropy $S_R(A:B)$, along with other generalizations such as its R\'enyi versions, as a function of the bond dimensions. We will find a phase structure that is consistent with the holographic proposal \cite{Dutta:2019gen} for the reflected entropy at least away from phase transitions, and the spectrum will smoothly interpolate between the different phases.

We summarize now the various parameter ranges we will be interested in, each some limit with $\chi_A, \chi_B, \chi_C \gg 1$. We consider the following fixed parameters as we take $\chi_i$ large: 
\begin{equation}
	x_A = \frac{\ln \chi_A}{\ln \chi_C} \, \qquad x_B = \frac{\ln \chi_B}{\ln \chi_C} 
\end{equation}
with $0 < x_A , x_B < \infty$. The ``entanglement wedge'' phase transition, corresponding to the Page phase transition that we are mostly interested in, lives along the line $x_A + x_B = 1$. There are other phase boundaries along $x_A = 1 + x_B$ and $x_B = 1 + x_A$ where the derivative of the mutual information jumps. However, the reflected entropy $S_R(A:B)$ does not have interesting behavior here so we are less interested in them. 

In terms of these parameters the mutual information approaches
\begin{equation}
	\frac{I(A:B)}{\ln \chi_C } \rightarrow 
	\begin{cases}
		0 \,, & x_A + x_B < 1 \\
		x_A + x_B -1 \,, & x_A + x_B > 1 \,\,\, \&\quad x_A - 1 < x_B < x_A + 1~,
	\end{cases}
\end{equation}
in the large bond dimensions limit.
Also, as we will demonstrate below, the reflected entropy behaves like
\begin{equation}
	\frac{S_R(A:B)}{\ln \chi_C } \rightarrow 
	\begin{cases}
		0 \,, & x_A + x_B < 1 \\
		2 \,{\rm min}(x_A,x_B) \,, & x_A + x_B > 1 \,\,\, \&\quad x_A - 1 < x_B < x_A + 1~.
	\end{cases}
\end{equation}
Note the value of the reflected entropy in the connected phase matches the holographic expectation; it is the entanglement wedge cross-section in this single-tensor model. 
Our goal will be to compute this rigorously and to also study $S_R$ near the phase boundaries. 

Consider now the phase boundary near $x_A + x_B = 1$. Set: 
\begin{equation}
	q = \chi_A \chi_B / \chi_C\, , \qquad y = \ln \chi_A/ \ln \chi_B 
\end{equation}
We will explore the parameter range $ 0 < q < \infty$ and $ 0 < y < \infty$ fixed as we send $\chi_C \rightarrow \infty$. The mutual information is well studied in this limit as it can be computed from the entanglement entropies \cite{Page:1993df}. The behavior across the phase transition is generally referred to as the Page result/transition. 

There is also another limit of interest at the edge of this phase boundary which lives near $x_B = 0$ and $x_A = 1$ where we fix $\chi_B \gg 1$ and send $\chi_C \rightarrow \infty$ holding $p = \chi_A/\chi_C$ fixed and $\chi_B$ fixed. This limit is mostly of interest because the analytic calculation is more tractable. 

We start by proving certain results rigorously for the reflected entropy in different regimes of parameter space in \secref{sec:ineq}. We then compute the phase diagram for the R\'enyi reflected entropy in the saddle point approximation in \secref{sec:phase_nm}. We then analyze the phase transitions in reflected entropy by using the resolvent trick. Before doing so, in \secref{sec:EE} we describe the phase transition in entanglement entropy which reproduces the Page curve as a warm up problem. This helps us establish the formalism and ingredients required to describe the reflected entropy phase transition in \secref{sec:SD}. We discuss the results of this calculation in \secref{sec:spectrum} and check these against numerics in \secref{sec:num}. 

\subsection{General considerations} \label{sec:ineq}

Let us start by proving some rigorous results on reflected entropies and the R\'enyi generalizations in the large bond dimension limit. These results will help guide our way through the replica treatment of this problem.

The intuition behind the following result is that for a random state, the reduced density matrix $\rho_{AB}$ has an approximately flat spectrum, implying it is approximately proportional to a projector at large bond dimension. This implies that the states $\left| \rho_{AB}^{m/2} \right>$ for different $m$ and after normalization, are close in the Hilbert space norm. This implies a certain continuity in $m$ for the $(m,n)$-R\'enyi reflected entropies. While continuity results for R\'enyi entropies are weak, the states in question are sufficiently close at large bond dimension to imply a useful result. We prove the following lemma: 
\begin{lemma}\label{conn_cont} 
	The $(m,n)$-R\'enyi reflected entropy for a tripartite random state with $\chi_C < \chi_A \chi_B$ satisfies a continuity bound as a function of $m$ for $1 \leq m \leq 2$ and $n>1$: 
	\begin{equation}\label{rbound} 
		\overline{\left| S_R^{(1,n)}(A:B) - S_R^{(m,n)}(A:B) \right|} \leq \frac{2n}{n-1} \left(\frac{ (m-1) \chi_C}{\chi_A \chi_B} \right)^{1/2} \min \{ \chi_A^{2 (n-1)},\chi_B^{2 (n-1)} \} 
	\end{equation}
	Also the reflected entropy (at $n=1$) satisfies: 
	\begin{equation}\label{ebound} 
		\overline{\left|S_R(A:B) - S_R^{(m,1)}( A:B) \right|} < 2\left( \frac{(m-1) \chi_C}{\chi_A \chi_B} \right)^{1/2} \min \{ \ln \chi_A, \ln \chi_B \} + \ln 2 
	\end{equation}
\end{lemma}
\begin{proof}
	We will also need the continuity bound for reflected entropy \cite{Akers:2019gcv}, which follows from applying the Fannes-Audenaert inequality \cite{fannes1973continuity,audenaert2006sharp} to the reflected entropy. In fact we will need a R\'enyi generalization of this inequality, which is usually considered to be a weak bound \cite{rastegin2011some}: 
	\begin{align}
		\left|S_n( \rho_{AA^\star}^{(1)} ) - S_n( \rho_{AA^\star}^{(m)} ) \right| &\leq \frac{ \chi_A^{2 (n-1)}}{n-1} \left(1-(1- T)^n -(\chi_A^2-1)^{1-n} T^n \right)\, \qquad n \geq 1 \\
		& \leq \frac{ n\chi_A^{2 (n-1)}}{n-1} T 
	\end{align}
	where $T$ is the trace distance: 
	\begin{align}
		T = \| \rho_{AA^\star}^{(1)} - \rho_{AA^\star}^{(m)} \|_1 &\leq 2 \sqrt{ 1 - |\big< \psi^{(1)} \big|\psi^{(m)} \big>|^2 } = 2 \sqrt{ 1 - \frac{\left({\rm Tr} \rho_{AB}^{(1+m)/2}\right)^2}{{\rm Tr} \rho_{AB}^{m}}} 
	\end{align}
	The later inequality follows from the bound of the trace norm by the fidelity, as well as monotonicity of the fidelity under tracing from $AB A^\star B^\star$ to $AA^\star$. The overlap $|\big< \psi^{(1)} \big|\psi^{(m)} \big>| $ is the fidelity on $AB A^\star B^\star$. The last equality uses the Hilbert-Schmidt norm. We estimate: 
	\begin{align}
		T &\leq 2 \sqrt{ 1 - \exp\left( -(m-1) (S_{(m+1)/2}(\rho_{AB}) - S_{m}(\rho_{AB})) \right) } \\
		& \leq 2((m-1))^{1/2} (S_{(m+1)/2}(\rho_{AB}) - S_m(\rho_{AB}))^{1/2} \\
		& \leq 2((m-1))^{1/2} (\ln \chi_C - S_2(\rho_{AB}))^{1/2}
	\end{align}
	where we used some elementary bounds and in the last line we used monotonicity of the R\'enyi entropies as a function of $m$: $S_m \geq S_2$ for $m\leq 2$.
	We also used $S_{(m+1)/2}(\rho_{AB}) \leq \ln \chi_C$ which follows from Schmidt rank and the assumption $\chi_C < \chi_A \chi_B$. Averaging over the Haar random unitaries we have: 
	\begin{equation}
		\overline{T} \leq 2((m-1))^{1/2} \overline{(\ln \chi_C - S_2(\rho_{AB}))^{1/2}} 
	\end{equation}
	We now apply the Lubkin result \cite{lubkin1978entropy}, see Ref.~\cite{liu2018entanglement}: 
	\begin{equation}
	\label{lubkin}
		 \overline{\left< S_2(\rho_{AB}) \right>} \geq \ln \chi_C - \frac{\chi_C}{\chi_A \chi_B} 
	\end{equation}
	So we have: 
	\begin{equation}
		\overline{\left( \ln \chi_C - S_2(\rho_{AB}) \right)^{1/2}} \leq \left( \overline{ \ln \chi_C - S_2(\rho_{AB}) } \right)^{1/2} \leq \left( \frac{\chi_C}{\chi_A \chi_B} \right)^{1/2} 
	\end{equation}
	by Jensen and concavity of the square root. Putting all these inequalities together we find that: 
	\begin{equation}
		\overline{T} \leq 2( (m-1))^{1/2} \left(\frac{\chi_C}{\chi_A \chi_B} \right)^{1/2} 
	\end{equation}
	which gives \Eqref{rbound} after using the symmetry between $A$ and $B$. A similar analysis for the $n=1$ case using now the Fannes continuity result and the bound on the Shannon correction term $ - T \ln T - (1-T) \ln (1-T) \leq \ln 2$ gives \Eqref{ebound}. 
\end{proof}

While we focused on $1 \leq m \leq 2$ there is no obstruction to finding similar results for later $m >2$, say by generalizing the Lubkin bound \Eqref{lubkin} for other integer R\'enyi entropies. The most efficient way to do this for general $m$ would involve taking the large bond dimension limit, whereas the result above is for any bond dimensions. We expect the conclusions are the same for $m>2$.

As an application of these results, consider sitting on the $(x_A, x_B)$ phase diagram above $x_A + x_B \geq 1$ and assume that $x_A < x_B$. We see that there is always some window $ 1 \leq n \leq n_{\rm c}$ where the leading large $\chi$ $(m,n)$-R\'enyi entropies are provably independent of $m$: 
\begin{align}
	\overline{|S_n( \rho_{AA^\star}^{(1)} ) - S_n( \rho_{AA^\star}^{(m)} ) |} 	
	\leq \frac{ 2n ( (m-1))^{1/2}}{n-1} \chi_C^{ 1/2(1 -x_A -x_B) + 2 x_A(n-1) } 
\end{align}
which implies that the change in the $(m,n)$-R\'enyi reflected entropy as a function of $m$ is non-perturbatively small for $n>1$ and $n < n_c$ where $n_c =  5/4 - (1 -x_B) /(4x_A) \geq 1$. Additionally at $n=1$ the shift as a function of $m$ can be order $1$. Based on explicit calculations we will argue that this is true for all $n>1$, although these results will be less rigorous, relying on certain assumptions about the analytic continuation. In this way Lemma~\ref{conn_cont} will serve as a check on our method for analytic continuation.

We can make a similar argument for $\chi_A \chi_B \leq \chi_C$, but now we expect the reflected density matrix to be close to a factorized density matrix $\rho_{AB} \approx 1_A \otimes 1_B/(\chi_A \chi_B)$. This allows us to make a stronger statement since the reflected entropy of this density matrix vanishes. 
\begin{lemma}
\label{lem2}
	For a tripartite random state with $\chi_A \chi_B \leq \chi_C$ 
	\begin{equation}
		\overline{S^{(m,n)}_R(A:B) } \leq \frac{2n}{n-1} \left( \frac{\chi_A \chi_B}{\chi_C} \right)^{1/2} \min \{ \chi_A^{2 (n-1)},\chi_B^{2 (n-1)} \} 
	\end{equation}
	for $1 \leq m \leq 2$ and $n > 1$. Similarly for the entropy: 
	\begin{equation}
		\overline{S^{(m,1)}_R(A:B) } \leq 2 \left( \frac{\chi_A \chi_B}{\chi_C} \right)^{1/2} \min \{ \ln \chi_A, \ln \chi_B \} + \ln 2 
	\end{equation}
\end{lemma}
\begin{proof}
	We use the same proof idea as in Lemma~\ref{conn_cont}, except now one of our states in the Fannes-like inequalities is the reduced state $\tilde{\rho}_{AA^\star}$ associated to the canonical purification for the maximally mixed state $\widetilde{\rho}_{AB} = 1_{AB}/(\chi_A\chi_B)$. Of course the reflected entropy vanishes for this state, and thus $\tilde{\rho}_{AA^\star}$ is a minimal projection. Again the trace distance can be estimated in terms of the overlap $|\braket{ \psi^{(m)} | \widetilde{\rho}_{AB}^{1/2}}|^2 = ({\rm Tr} \rho_{AB}^{m/2})^2/( \chi_A \chi_B {\rm Tr} \rho_{AB}^m)$. For $1 \leq m \leq 2$ we apply monotonicity $S_{m/2} \geq S_2$ and $S_m \geq S_2$ and finally we use the Lubkin bound (for the case $\chi_A \chi_B > \chi_C$). This gives the stated results. 
\end{proof}

Note that if we simply set $m=1$ we see this constrains the \emph{answer} that we are interested in. Both the reflected entropy and the R\'enyi reflected entropies are provably non-perturbatively small, i.e. power law suppressed in $\chi_C$, for $ x_A + x_B < 1$ and for $n > 1$ (using monotonicity with $n$).

\subsection{Phase diagram and R\'enyi reflected entropies} \label{sec:phase_nm} In order to proceed we will use the replica trick at $(m/2,n)$ integer to evaluate the $(m,n)$-R\'enyi reflected entropy. We will find a phase diagram as a function of the bond dimensions and then give some arguments for how to analytically continue this phase diagram away from the integers. 

The integer moments of the density matrix for the reflected entropy reads (up to the usual small corrections from the treatment of the normalization): 
\begin{align}
	\overline{{\rm Tr}( \rho^{(m)}_{AA^\star})^n} = \frac{\sum_{g\in S_{mn}} \exp \left( - \ln \chi_C ( x_A d( g, g_A) + x_B d(g, g_B) + d(g,e) )\right)} {\left(\sum_{g\in S_{m}} \exp\left( - \ln \chi_C ( (x_A + x_B) d(g, \tau_m) + d(g,e) ) \right)\right)^n} 
\end{align}
We now simply find the phase diagram (as a function of the various bond dimensions) at fixed integer $(m/2,n)$. At each point in the phase diagram some set of $g$ dominates.
We will find a picture of the phase diagram where at each point one of four possible elements dominates: $e, g_A, g_B$ and a new element we call $X$. 

To set the stage we are interested in minimizing the following free energy: 
\begin{equation}
	f(g) = x_A d(g,g_A) + x_B d(g,g_B) + d(g,e)\, \qquad g \in S_{mn} 
\end{equation}
where we have factored out $\ln \chi_C$. Since $f$ is a linear function of $x_A$ and $x_B$, any region of the phase diagram in the $(x_A \geq 0 ,x_B \geq 0)$ plane must be convex.\footnote{See Lemma \ref{lem:convexity} in Appendix \ref{app:phase_proof} for a proof.}
The Cayley distance, as a metric on $S_{mn}$, satisfies the triangle inequality $d(g_1,g_2)+d(g_2,g_3)\geq d(g_1,g_3)$.
It will be useful to introduce the geodesics in the Cayley that pass between two group elements $g_1, g_2$ as the set that saturates this inequality. We denote these elements using: 
\begin{equation}
	\Gamma(g_1,g_2) = \{ g \in S_{mn}: d(g_1,g) + d(g,g_2) = d(g_1,g_2) \} 
\end{equation}
See Appendix \ref{app:perm} for a discussion of elements on this geodesic. 

We now determine the two-dimensional phase diagram in terms of $(x_A,x_B)$. 
As a first step we proceed as follows. We solve for the dominating elements at various special points in the phase diagram. Then we fill in the undetermined regions by the convexity of the phase diagram. 
As a limiting case, it is easy to show that the identity element $e$ dominates at $x_A=x_B=0$ and $g_A(g_B)$ dominates at $x_A(x_B)\to\infty$. For the other regions:
\begin{itemize}
	\item Along the line $x_A + x_B = 1$: 
	\begin{align}
		f &= x_A ( d(g,g_A) + d(g,e)) + (1- x_A) ( d(g,g_B) + d(g,e)) \\
		& \geq x_A d(g_A,e) + (1- x_A) d(g_B,e) = n (m-1) 
	\end{align}
	with saturation is achieved for $ g \in \Gamma(g_A,e) \cap \Gamma(g_B,e)$. 
	
	\item Along the line $x_A = x_B + 1$: 
	\begin{align}
		f &= x_B ( d(g,g_A) + d(g,g_B)) + ( d(g,g_A) + d(g,e)) \\
		& \geq x_B d(g_A,g_B) + d(g_A,e) = 2(n-1) x_B + n (m-1) 
	\end{align}
	with saturation for $ g = g_A$ since this is the single element in $\Gamma(g_A,g_B) \cap \Gamma(g_A,e)$. 
	\item Along the line $x_B = x_A + 1$ is the same as above with $A \leftrightarrow B$. 
\end{itemize}
After applying convexity to the above results the only region left is for $x_A + x_B > 1$ and $1- x_B < x_A < 1 + x_B$. 
This is of course the main region of interest where the entanglement wedge is connected. 
Because there is no common geodesic element between all three elements (that is $\Gamma(g_A,e) \cap \Gamma(g_B,e) \cap \Gamma(g_A,g_B)$ is empty), there is no simple argument that determines the phase in this region.\footnote{If there was a common geodesic element then this element would sit on all three edges of the region under consideration - so we could then fill the phase diagram inside the region with this common element using convexity.} This is the main difference between our calculation and the negativity computations in \cite{Shapourian:2020mkc,Dong:2021oad,Vardhan:2021npf}.
Instead we must seek other methods for determining the phase in this region.
Indeed, there exists a proof for this main region, providing a complete picture of the phase diagram.
However, because the proof is relatively involved, we present it in Appendix \ref{app:phase_proof} for interested readers.
For now we just summarize the results we find:

\begin{itemize}
    \item There exists a new phase with a non-trivial element, which we denote $X$, in the triangular region $x_A + x_B \geq 1, x_A < 1+x_B (1-2/n), x_B < 1+ x_A (1-2/n)$.
    $X$ is a group element with the special property that it lies on the joint geodesic $ \Gamma(g_A,e) \cap \Gamma(g_B,e)$ while being closest to $g_A$ or $g_B$.
    It has the following form:
    \begin{equation}
        X = \prod^n_{i=1} (\tau^{U}_{m/2})^{[i]_n}(\tau^{D}_{m/2})^{[i]_n}
    \end{equation}
    where $\tau^U_{m/2}$ is the cyclic permutation that acts on the upper half of each \emph{puddle} (that is it cyclically permutes elements $\beta = 1,\cdots,m/2$ within each fixed puddle) and $\tau^D_{m/2}$ is the full cyclic permutation on the lower half (cyclically permuting elements $\beta = m/2+1,\cdots,m$ within each fixed puddle). See the discussion around \figref{fig:g_A_g_B} to unpack this notation.
    We give a graphical representation for the element $X$ in \figref{fig:X}.
    \begin{figure}[h] 
        \centering 
    	\includegraphics[width=.7
    	\textwidth]{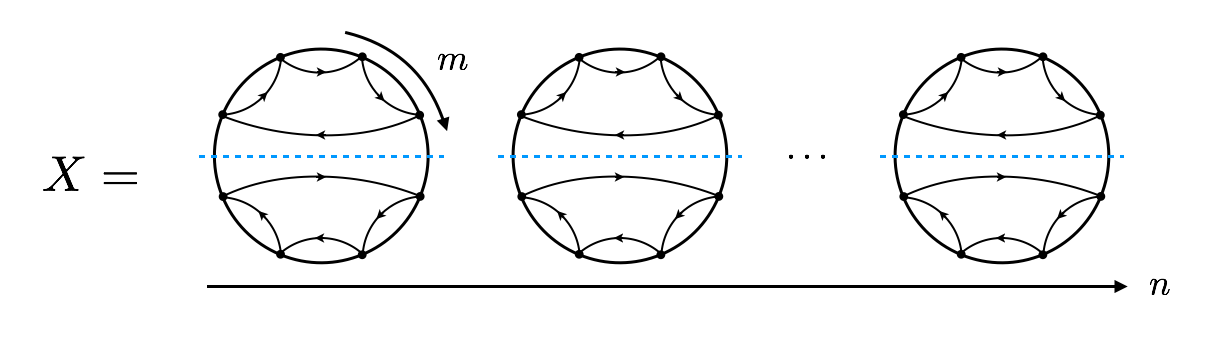} 
    	\caption{A graphical representation of the special element $X$, using the same conventions as Fig~\ref{fig:g_A_g_B}.} 
      \label{fig:X} 
    \end{figure}
    \item The phase of $g=g_A$ smoothly extends into the main region and occupies the upper-diagonal wedge $x_A>x_B, x_A > 1+x_B (1-2/n)$.
    The same is also true for $g_B$, where it occupies the lower-diagonal wedge.
    At the boundary where the two phases meet (i.e. $x_A = x_B$ and $x_A>n/2$) we have a large degeneracy shared by a complicated set of elements.\footnote{For detailed description of this set please see Appendix \ref{app:2nd_resum}.}
    \item  No other elements are dominant in the space between the co-existence boundaries. 
\end{itemize}
We give an example of the phase diagram in \figref{fig:1site_phase}.
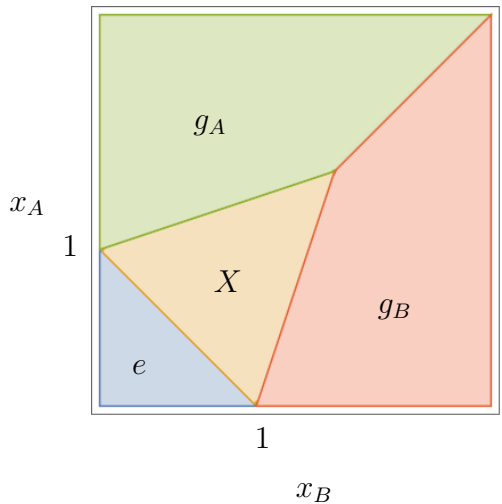
\begin{figure}
	[h] \centering 
	\begin{overpic}
		[width=0.35\linewidth]{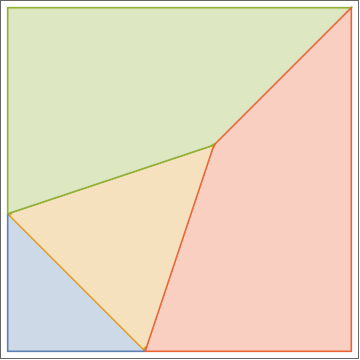} \put(-20,50){$x_A$} \put(50,-20){$x_B$} \put(-7,39){$1$} \put(40,-8){$1$} \put(10,10){$e$} \put(30,30){$X$} \put(25,70){$g_A$} \put(70,25){$g_B$} 
	\end{overpic}
	\vspace{1.5cm} 
	\caption{An example of the phase diagram of the single tripartite tensor. The tip of the $X$ triangle lies at $(n/2,n/2)$.} 
\label{fig:1site_phase} \end{figure}

We now extract the consequences for $S_R$ and its R\'enyi versions. For the disconnected entanglement wedge with $x_A + x_B < 1$, there is only ever one phase for all $(n,m)$ so it is reasonable that this remains the case upon analytic continuation. Indeed this always gives $S_R^{(m,n)} = 0$ upon correctly normalizing the free energy. This is of course consistent with Lemma~\ref{lem2}.

From now on we consider only $x_A + x_B > 1$. The free energies of the four possible phases, including the normalization subtraction (which is different to the disconnected phase) are: 
\begin{align}
\begin{split}
    F(e)/\ln \chi_C &= (x_A + x_B-1) (m-1) n\,, \quad  F(X)/\ln \chi_C = - n + (x_A + x_B) n \\
	F(g_A)/\ln \chi_C &= 2 x_B (n-1) \,, \qquad\qquad\qquad  F(g_B)/\ln \chi_C = 2 x_A (n-1) 
\end{split}
\end{align}

We make some observations. Firstly, at fixed $(x_A,x_B)$ in certain regions of the phase diagram there is a new novel behavior where there is a phase transition as a function of $n$. This indicates that the entanglement spectrum is \emph{not flat} for the reflected density matrix, as compared to the usual reduced density matrix of a random pure state. Secondly, if we do a very naive analytic continuation away from the integers, where we include all four phases and re-minimize at each $(n,m)$ we arise at some puzzling results. For example at $m=1$ and fixed $n>1$ the $e$ phase has minimal free energy and: 
\begin{equation}
\label{naivee}
	\left. S_R^{(1,n)}\right|_{\rm naive} = 0\,, \qquad n > 1 
\end{equation}
Which seems to imply that taking $\lim n \rightarrow 1$, $S_R(A:B) = 0$ for $x_A + x_B > 1$ which cannot be the case. In particular it violates the mutual information bound $S_R \geq I$. Previously this was attributed to an order of limits issue \cite{Kusuki:2019evw} where one must take $\lim_{m \rightarrow 1} \lim_{n \rightarrow 1}$ to get the right answer - indeed here that resolution does work. However this result gives a very puzzling set of R\'enyi entropies, which were naively still consistent. However we can now observe using Lemma~\ref{conn_cont} that they are in fact not consistent, since for $m = 2$, we see that either $F(X)=F(e)$ or $F(g_{A,B})$ dominates and neither gives the same reflected entropy as the naive $m=1$ phase with $g=e$ given in \Eqref{naivee}. This is inconsistent with the continuity bounds of Lemma~\ref{conn_cont}, at least in some window of parameters. 

Note that one way to understand the origin of this latter issue is as follows. The Cayley distance between $X$ and $e$ reads 
\begin{equation}
	d(X,e) = n(m-2) 
\end{equation}
For $m<2$ this is negative. The effect of this negative tension is to exchange the two phases $e\leftrightarrow X$ in the phase diagram \figref{fig:1site_phase}. This leads to the dominance of the $e$ phase even deep in the phase diagram where the entanglement wedge is connected. 

The resolution that we land on is to simply discard the $F(e)$ saddle for $x_A + x_B > 1$. This is reasonable since it never appears in the $x_A + x_B > 1$ part of the phase diagram for integer $(n,m/2)$. This will now give a different answer for the R\'enyi entropies as compared to the naive analysis.
We motivate this prescription as follows.   Rather than analytically continue the R\'enyi entropies we will analytically continue the \emph{spectrum} of $\rho_{AA^\star}^{(m)}$ from $m>2$ down to $m=1$. This has the same effect of discarding the $e$ saddle. In fact the analytic continuation in $m$ is trivial since there is simply no dependence on $m$ in the free energies of $X,g_A,g_B$. This is now consistent with the bounds we derived in Lemma~\ref{conn_cont}. The R\'enyi entropies read: 
\begin{equation}
\label{possbl}
	S_R^{(m,n)} = \ln \chi_C \times \min\left\{ 2 x_B, 2x_A, \frac{n}{n-1} (x_B + x_A -1) \right\}\, \qquad n \geq 1 
\end{equation}
and the entanglement spectrum defined via: 
\begin{equation}
	{\rm Tr} (\rho_{AA^\star}^{(m)})^n = \int d \lambda \lambda^{n} D(\lambda) 
\end{equation}
is given by: 
\begin{equation}
\label{D1D0}
	D(\lambda) = D_1 \delta (\lambda - \lambda_1)+ \delta( \lambda - \lambda_0) \,, \qquad D_1 = {\rm min}\{ \chi_C^{2 x_A}, \chi_C^{2 x_B}\} 
\end{equation}
where $\lambda_0 = \max \{\chi_C^{ 1- x_A - x_B},\lambda_1\}$ and $\lambda_1 = {\rm max}\{ \chi_C^{-2 x_A}, \chi_C^{-2 x_B}  \}$. \footnote{Note that for $x_A < x_B - 1$ then $n(x_B+x_A-1)> 2 n x_A > 2 (n-1) x_A$  so the last possibility in \Eqref{possbl} is never dominant for $n\geq1$. We interpret this to say that the single eigenvalue $\lambda_0$ gets absorbed by the large set of $D_1$ eigenvalues at $\lambda_1$ in this particular part of the phase diagram. We have written the spectrum in a way that is consistent with this.} 

This is a crude approximation to the actual spectrum that we will compute in the next section. One obvious crude feature, for example, is that it has one more than the allowed number of non-zero eigenvalues, ${\rm min}\{ \chi_A^2, \chi_B^2\}$, although this is clearly a small correction in the large $\chi_C$ limit.

We now turn to an examination of the phase transition near $x_A + x_B = 1$. This will give some further evidence for the prescription given above, and will also give a more complete picture of the reflected entropy and spectrum. 
It is clear from the above discussion that there are potentially many $g$ elements that we need to sum up in order to compute the cross-over corrections. While we might attempt to directly sum over a class of $g$ this approach turns out to be somewhat ad-hoc.
The main difficulty is that certain classes of elements can be more important at different values of $n$, so the issue of which elements to include is mixed up with the question of how to analytically continue in $n$. This is in turn related to the fact that the spectrum itself does not have a uniform limit as $\chi \rightarrow \infty$, being made up of distinct contributions.
Instead, in the next few sections, we will develop a diagrammatic approach that attempts to directly extract the spectrum. Having said this we briefly discuss an explicit sum over elements in Appendix~\ref{app:2nd_resum}, where we do find consistent results.

\subsection{Schwinger-Dyson for entanglement entropy}
\label{sec:EE}

In this section we will present a standard diagrammatic technique for computing the entanglement spectrum of a Haar random pure state on a bipartite Hilbert space $AB$ reduced to $A$ \cite{Page:1993df}. This is of course the setting for the Page curve, so these results are very well known. We review it here just to setup notation used in the next section for the reflected entanglement spectrum problem.

The approach we present was first proposed by Ref.~\cite{Jurkiewicz08,BREZIN1993613,BREZIN1995531}. It was recently applied to find the entanglement entropy of JT gravity \cite{Penington:2019kki} and generic fixed area states \cite{Akers:2020pmf}. Similar techniques were used to calculate the negativity \cite{Kudler-Flam:2021efr,Dong:2021oad}. This powerful approach enables us to write down the Schwinger-Dyson equation for the \textit{resolvent} of $\rho_{AB}$, which then gives full information about the entanglement spectrum. 

A density matrix $\ket{\psi}\bra{\psi}$ is represented as a four-legged tensor 
\begin{equation}
	\includegraphics[width=0.33
	\textwidth]{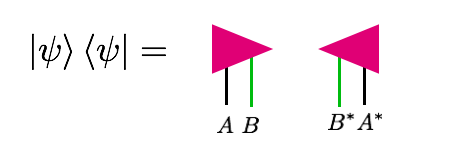} 
\end{equation}
where we use black (green) lines to represent the subspace associated to $A$ ($B$). We can form the reduced density matrix by tracing over the $B$ indices 
\begin{equation}
	\includegraphics[width=0.42
	\textwidth]{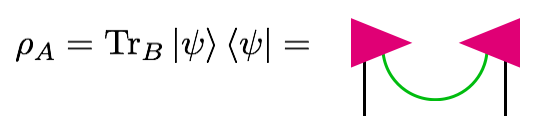} 
\end{equation}
Averaging over random states is accomplished by means of pair contractions between bras and kets. In \secref{sec:RTN} we reviewed the fact that an average over (the $m$-th power of) a Haar random state can be recast as a summation over permutation group elements $g\in S_m$. In terms of diagrams these permutations become pairwise contractions between the bras and kets: 
\begin{equation}
	\includegraphics[width=0.7
	\textwidth]{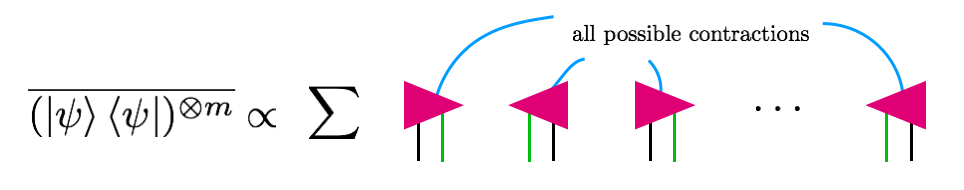} 
\end{equation}
A single pair contraction is defined by simply connecting the corresponding legs: 
\begin{equation}
	\includegraphics[width=0.45
	\textwidth]{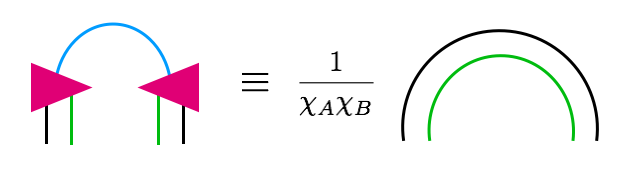} 
\end{equation}
A solid black (green) line represents the Kronecker delta.
The purpose of the $(\chi_A\chi_B)^{-1}$ factor associated to each contraction is to maintain the correct normalization of the density matrix. The element $g\in S_m$ is recovered by tracing out the contractions from bras to kets. This allows one to recover the Haar averaged expression \Eqref{eq:IsingAction}.

As a simple example, let us compute the second R\'enyi entropy (purity) of the reduced state using this method. Connecting the lower legs appropriately we have 
\begin{equation}
	\includegraphics[width=0.7
	\textwidth]{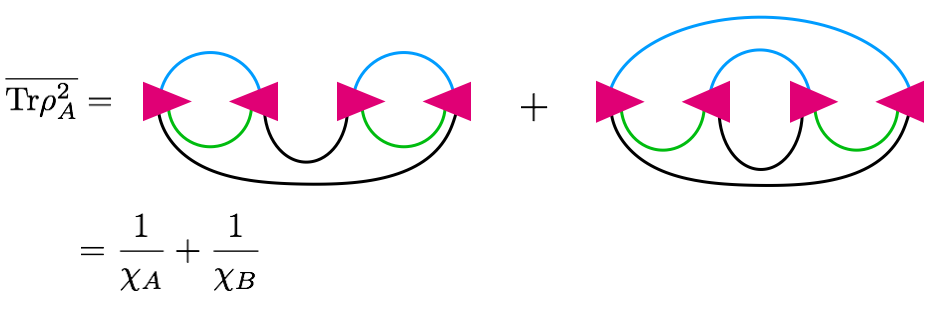} 
\end{equation}
where we have used the fact that a full contraction of black (green) loop gives a factor of $\chi_A$ ($\chi_B$), the corresponding Hilbert space dimension.

To get the full eigenvalue distribution of $\rho_{A}$ we will make use of the resolvent trick. The resolvent $R(\lambda)$ for the density matrix $\rho_{A}$ is defined formally as 
\begin{align}
	R(\lambda)=\tr\left( \frac{1}{\lambda-\rho_A} \right) 
\end{align}
from which one can extract the eigenvalue distribution function $D(\lambda)$ using 
\begin{align}
	D(\lambda) &= -\frac{1}{\pi}\lim_{\epsilon \rightarrow 0} \text{Im} R(\lambda+i\epsilon) 
\end{align}
To evaluate $R(\lambda)$, we expand the matrix inverse around $\lambda=\infty$: 
\begin{equation}
	R_{ij}(\lambda)=\frac{\chi_{A}}{\lambda}\delta_{ij}+\sum_{m=1}^{\infty} \frac{ (\rho_{A}^m)_{ij}}{\lambda^{m+1}}. 
\end{equation}
In terms of diagrams, this is 
\begin{equation}
	\includegraphics[width=0.85
	\textwidth]{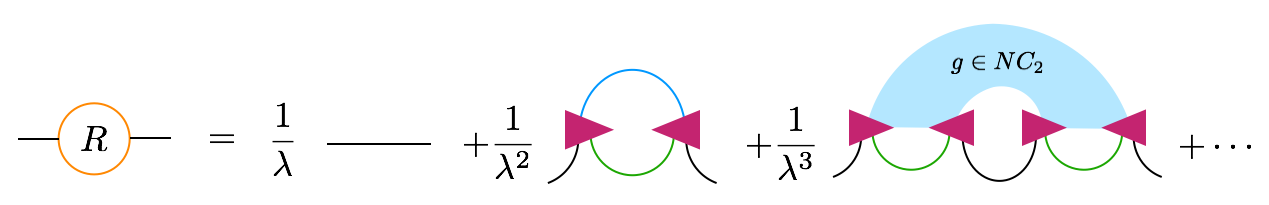}, 
\end{equation}
where we have restricted our summation to planar diagrams only, or equivalently, over all the corresponding non-crossing permutations. The contribution of non-planar diagrams are always suppressed in the limit of the large bond dimension. 
\begin{figure}
	[h] \centering 
	\includegraphics[width=0.8
	\textwidth]{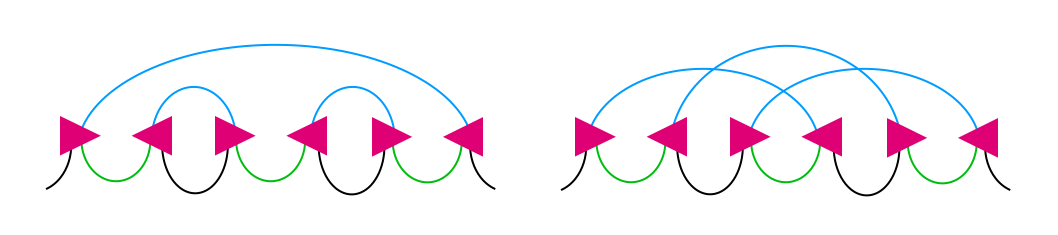} 
	\caption{(left) An example of the planar diagram in the calculating $\rho^3$, featuring a non-crossing permutation for the contraction. (right) An example of a non-planar diagram. The contribution is always suppressed by powers of $\chi$.} 
\end{figure}

Now denote $F_{ij}(\lambda)$ to be the connected part of the resolvent, defined by 
\begin{equation}\label{eq:resolv2} 
	\includegraphics[width=0.8
	\textwidth]{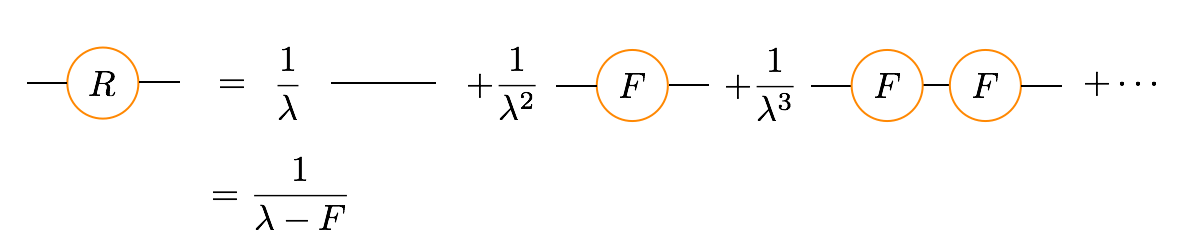} 
\end{equation}
It allows us to write down a Schwinger-Dyson equation for $R_{ij}(\lambda)$ and $F_{ij}(\lambda)$: 
\begin{equation}
	\includegraphics[width=
	\textwidth]{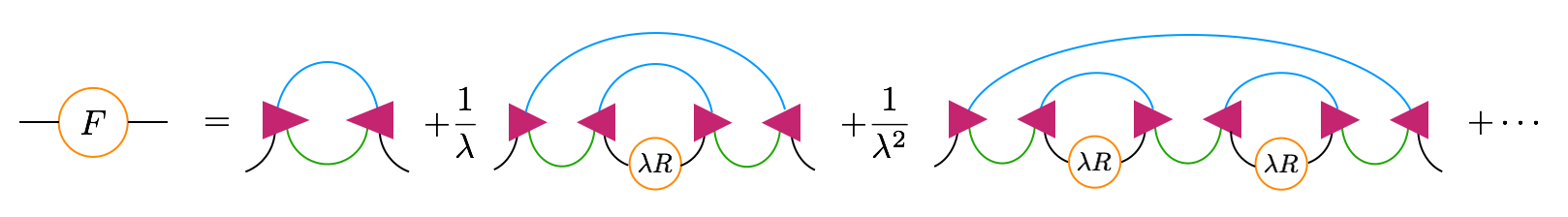} 
\end{equation}
which after taking the trace gives rise to the following algebraic equation 
\begin{align}
	\tr F_{ij}(\lambda) = 1 + \sum^{\infty}_{m=1}\frac{R(\lambda)^m}{(\chi_A\chi_B)^m} 
	= \frac{1}{1-\frac{R(\lambda)}{\chi_A\chi_B}} 
\end{align}
Combining this with \Eqref{eq:resolv2}, and using the fact that both $R$ and $F$ are proportional to the identity matrix (as implied by contraction rules) we obtain a quadratic equation for $R(\lambda)$: 
\begin{align}
	\frac{\lambda R^2(\lambda)}{\chi_A\chi_B}+ \left( \frac{1}{\chi_A}-\frac{1}{\chi_B}-\lambda\right)R(\lambda) + \chi_A = 0 
\end{align}
Solving for $R(\lambda)$ and picking the square root sign by demanding the correct behavior at $\lambda\to \infty$ we find 
\begin{align}
	R(\lambda) = \frac{\chi_A\chi_B}{2\lambda}\left( \lambda-\frac{1}{\chi_A}+\frac{1}{\chi_B} - \sqrt{(\lambda-\lambda_-)(\lambda-\lambda_+)} \right), \quad \lambda_\pm = \left( \frac{1}{\sqrt{\chi_A}}\pm \frac{1}{\sqrt{\chi_B}} \right)^2 
\end{align}
and 
\begin{equation}
	D(\lambda) = \frac{\chi_A\chi_B}{2\pi\lambda}\sqrt{(\lambda-\lambda_-)(\lambda_+-\lambda )} + \delta(\lambda)(\chi_A-\chi_B)\theta(\chi_A-\chi_B). 
\end{equation}
Thus, we have found that the eigenvalue spectrum consists of a large peak within the range $[\lambda_-,\lambda_+]$ and a number of extra zero eigenvalues when $\chi_A>\chi_B$. This is the famous \emph{Marchenko-Pastur (MP) distribution}. It describes the singular value distribution of a random $m\times n$ matrix when both $m,n\gg 1$. This result was first shown by Page \cite{Page:1993df}. Using the spectrum, the entanglement entropy is given by 
\begin{align}
	S(A) = -\tr \rho_{A} \ln \rho_A = 
	\begin{cases}
		\log(\chi_A) - \frac{\chi_A}{2\chi_B}, \quad &\chi_A < \chi_B \\
		\log(\chi_B) - \frac{\chi_B}{2\chi_A}, \quad &\chi_B < \chi_A. \\
	\end{cases}
\end{align}

\subsection{Schwinger-Dyson for
reflected entropy} \label{sec:SD} We now move on to our main problem. We want to find the resolvent for the reduced density matrix $\rho_{AA^*}$ obtained from the R\'enyi generalization of the canonically purified state $\ket{\psi^{(m)}}$ defined in \Eqref{eq:psim}, i.e., 
\begin{equation}
	R_m(\lambda) = \frac{\chi^2_{A}}{\lambda} + \sum^\infty_{n=1} \frac{\tr (\rho^{(m)}_{AA^*})^n}{\lambda^{n+1}} 
\end{equation}
where the integer moments of the (normalized) density matrix are given by
\begin{equation}\label{eq:reflected_moments} 
	\overline{\tr(\rho^{(m)}_{AA^*})^n} = \frac{Z_n}{(Z_1)^n}, \quad Z_n = \frac{\sum_{g\in S_{mn}}\chi_A^{
	\#gg^{-1}_A}\chi_B^{
	\#gg^{-1}_B}\chi_C^{\#(g)}}{\sum_{g\in S_{mn}}(\chi_A\chi_B\chi_C)^{\#(g)}} 
\end{equation}
We would like to evaluate this sum in the limit where the bond dimensions $\chi_A,\chi_B,\chi_C$ are taken to be large and the ratio $q\equiv \chi_A\chi_B/\chi_C$ and $r\equiv \chi_A/\chi_B$ are held fixed. We will do so by a diagrammatic technique, keeping track of the important diagrams. Ultimately we are interested in the analytic continuation $m\to 1$, where the normalized factor $Z_1 = \tr (\rho^{m}_{AB}) \to 1$, so it suffices to only consider the numerator partition function $Z_n$ for now. We will later discuss how to restore correct normalizations in our calculation for general $m$. Most of the notation in the previous section carries over, except that our state $\ket{\psi}$ is now a three-legged tensor.
\begin{figure}
	[h!] \centering 
	\includegraphics[width=.8
	\textwidth]{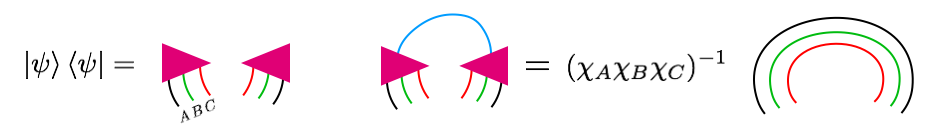} 
\end{figure}
Note that the denominator in $Z_n$ will be approximated as $(\chi_A \chi_B \chi_C)^{mn}$ from the identity group element, which dominates for large bond dimensions. We can easily account for this factor in the diagrammatic expansion by dividing by $\chi_A \chi_B \chi_C$ for each contraction as shown in the figure above.

We now setup a slightly more general problem. Consider the following $2 \times 2$ ``matrix'' of resolvents:\footnote{All powers of the (involuted) density matrix $\rho^{(m)}_{AA^*}, (\rho^{m/2}_{AB})^\Gamma, etc...$ appearing in the resolvent calculation, unless otherwise noted, refer to the Haar averaged version $\overline{\rho^{(m)}_{AB}}, \overline{(\rho^{m/2}_{AB})^\Gamma} $ etc.} 
\begin{equation}
	\mathbb{R}(\lambda) = \sum_{k=0}^\infty \lambda^{-1-k/2} 
	\begin{pmatrix}
		0 & (\rho^{m/2}_{AB})^{\Gamma^\dagger} \\
		(\rho^{m/2}_{AB})^\Gamma & 0 
	\end{pmatrix}
	^k , 
\end{equation}
where $\rho \rightarrow \rho^\Gamma$ is an involution defined to take a linear operator $\rho$ on $\mathcal{H}_{AB}$ to a linear operator from $\mathcal{H}_{AA^\star} \rightarrow \mathcal{H}_{BB^\star}$. It is defined by re-arranging the incoming/outgoing legs in the obvious (and canonical/basis independent) way. In fact it is the same as the correspondence between the Choi state $\tau$ of a channel (here, the would be channel maps $\mathcal{B}(\mathcal{H}^\star_A)$ to $\mathcal{B}(\mathcal{H}_B)$) and the linear representation of the channel via a transfer matrix $\tau^\Gamma$. Similarly $\rho \rightarrow \rho^{\Gamma^\dagger}$ spits out an operator from $\mathcal{H}_{BB^\star} \rightarrow \mathcal{H}_{AA^\star}$ which is also the adjoint operator. Each insertion of $(\rho^{m/2}_{AB})^{\Gamma}$ (and $(\rho^{m/2}_{AB})^{\Gamma^\dagger}$) has $m/2$ replicas of each of the bra and ket of the original random state $\ket{\psi}$. Note that $\mathbb{R}(\lambda) \in \mathcal{B}( \mathcal{H}_{AA^\star} \oplus \mathcal{H}_{BB^\star} )$. These involuted density operators are related to the canonical purified density matrix by 
\begin{equation}
	(\rho^{m/2}_{AB})^{\Gamma^\dagger} (\rho^{m/2}_{AB})^\Gamma = \rho^{(m)}_{AA^\star} \,, \qquad (\rho^{m/2}_{AB})^\Gamma (\rho^{m/2}_{AB})^{\Gamma^\dagger} = \rho^{(m)}_{BB^\star} .
\end{equation}
up to an overall normalization factor that we will correct for later. 
In terms of diagrams, this is (e.g. for $m=6$) 
\begin{equation}
	\includegraphics[width=0.7
	\textwidth]{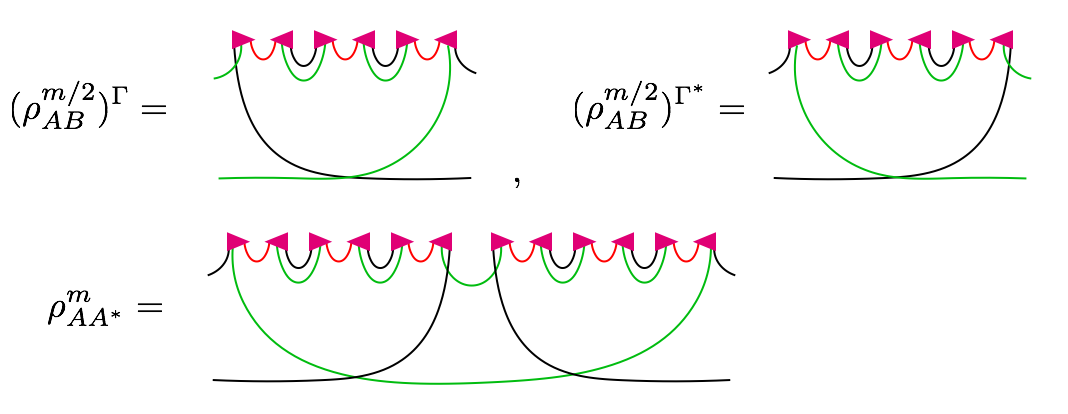}. 
\end{equation}
In the following calculations we will represent them diagrammatically using a short-hand: 
\begin{equation}
	\includegraphics[width=0.6
	\textwidth]{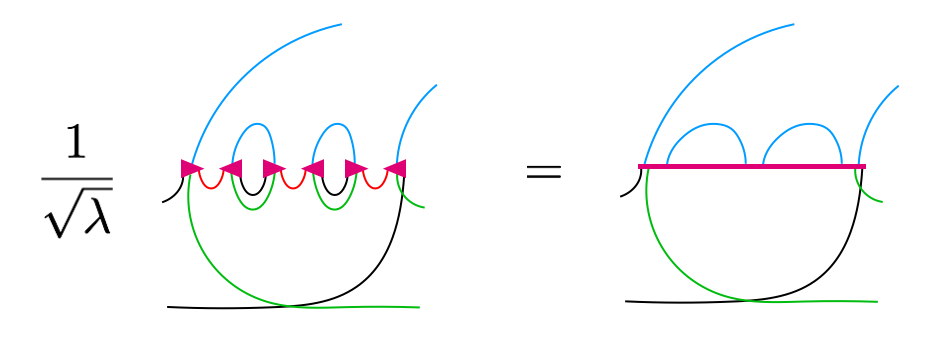} . 
\end{equation}
From these results we can infer an alternative expression for matrix $\mathbb{R}(\lambda)$: 
\begin{align}
	\mathbb{R}(\lambda) &= \sum_{n=0}^\infty \lambda^{-1-n} 
	\begin{pmatrix}
		(\rho_{AA^\star}^{(m)})^n & 0 \\
		0 & (\rho_{BB^\star}^{(m)})^n 
	\end{pmatrix}
	 \\ & \qquad \qquad +\sum_{n=1}^\infty \lambda^{-1/2-n} 
	\begin{pmatrix}
		0 & (\rho_{AA^\star}^{(m)})^{n-1} (\rho^{m/2}_{AB})^{\Gamma^\dagger} \\
		(\rho_{BB^\star}^{(m)})^{n-1} (\rho^{m/2}_{AB})^\Gamma & 0 
	\end{pmatrix}
	, 
\end{align}
which diagrammatically looks like 
\begin{equation}
	\includegraphics[width=.9
	\textwidth]{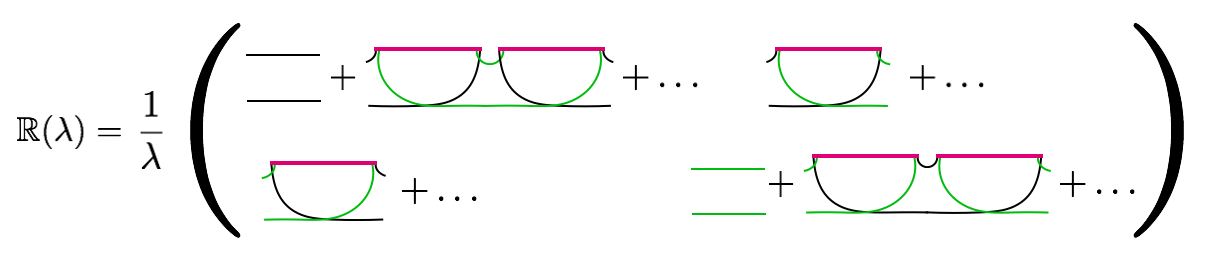}. 
\end{equation}
Thus the actual resolvent we are interested in is 
\begin{equation}
	R(\lambda) = {\rm Tr}_{AA^\star} ( \mathbb{R}_{11}(\lambda) ). 
\end{equation}

Now, as in the calculation of the entanglement entropy, the Haar averaging over the states is done by summing over all possible pair contractions over bras and kets. We call each pattern of contractions a \emph{diagram}. We call a single copy of $(\rho^{(m/2)}_{AB})^{\Gamma}$ a \emph{puddle}. Each diagram corresponds to an element in $g \in S_{mk/2}$ where $k$ is the total number of puddles. A sub-diagram is a subset of puddles and associated contractions that act only inside this subset.
We say a diagram is \emph{connected} if the diagram cannot be split into more than one sub-diagrams each made up of a \emph{contiguous} set of puddles. Otherwise the diagram is \emph{disconnected}. 

Consider the connected part of $\mathbb{R}$, which we call $\mathbb{F}$. This corresponds to a sum over a subset of diagrams which are connected. We have 
\begin{align}\label{eq:F_def} 
	\mathbb{R}(\lambda) = \frac{1}{\lambda} + \frac{\mathbb{F}}{\lambda^2} + \frac{\mathbb{F}^2}{\lambda^3} + \ldots = \frac{1}{\lambda-\mathbb{F}} 
\end{align}
Note that $\mathbb{F}$ and $\mathbb{R}$ are constrained to take the following form: (all possible contractions give rise to this form, for example see \figref{fig:proj_id})
\begin{equation}
	\mathbb{F}(\lambda) = 
	\begin{pmatrix}
		G_{11} (1_{AA^\star} - e_A) + F_{11} e_A & F_{12} \left| \epsilon_A \right> \left< \epsilon_B \right| \\
		F_{21} \left| \epsilon_B \right> \left< \epsilon_A \right| & G_{22} (1_{BB^\star} -e_B) +F_{22} e_B 
	\end{pmatrix}
	, 
\end{equation}

\begin{figure}
    \centering
    \includegraphics[width=\textwidth]{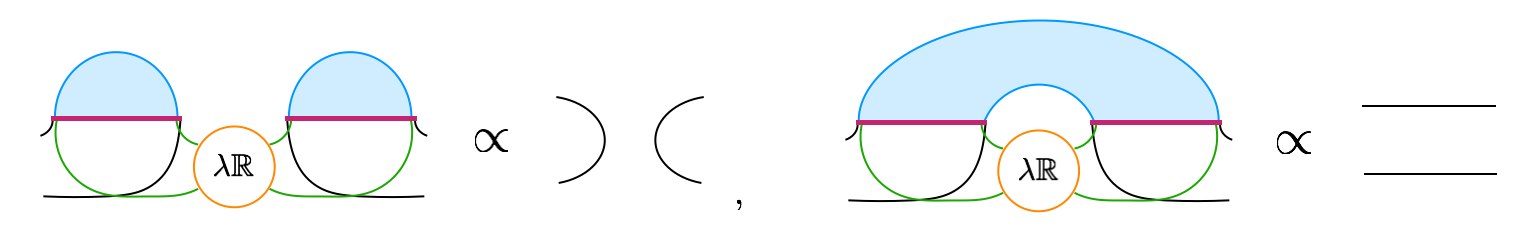}
    \caption{Example contractions of the diagram that give rise to the projector or the identity, where it should be clear that all the possible diagrams are either proportional to identity operator $1_{AA^*}, 1_{BB^*}$, or one of the four projectors $\ket{\epsilon_{A,B}}\bra{\epsilon_{A,B}}$.}
    \label{fig:proj_id}
\end{figure}
where $\ket{\epsilon_A} = \chi^{-1/2}_A\ket{1_A}$ is the maximally mixed state on $AA^*$ (same for $\ket{\epsilon_B}$), $e_A = \left| \epsilon_A \right> \left< \epsilon_A \right|$ and $e_B = \left| \epsilon_B \right> \left< \epsilon_B \right|$ are normalized minimal projectors and $F$ is a $2\times 2$ matrix of scalars (not to be confused with $\mathbb{F}$). Also $1_{AA^\star}$ is the identity acting on $\mathcal{H}_{AA^\star}$ etc. A similar form applies to 
\begin{equation}
	\mathbb{R}(\lambda) = 
	\begin{pmatrix}
		S_{11} (1_{AA^\star} -e_A) + R_{11} e_A & R_{12} \left| \epsilon_A \right> \left< \epsilon_B \right| \\
		R_{21} \left| \epsilon_B \right> \left< \epsilon_A \right| & S_{22} (1_{BB^\star}-e_B) +R_{22} e_B 
	\end{pmatrix}
	. 
\end{equation}
Meanwhile we must have $F_{12}^\star = F_{21}$ and $R_{12}^\star = R_{21}$ (for real $\lambda$). Then from \Eqref{eq:F_def} we obtain: 
\begin{align}\label{eq:SD-defining} 
	S_{11} = (\lambda - G_{11})^{-1}\,, \qquad S_{22} = (\lambda - G_{22})^{-1}\,, \qquad R_{ij}=(\lambda-F)^{-1}_{ij} , 
\end{align}
and the resolvent of interest is 
\begin{equation}
	R(\lambda) = \tr \mathbb{R}_{11} = (\chi^2_A-1)S_{11} + R_{11}. 
\end{equation}
We now seek a Schwinger-Dyson equation for $\mathbb{F}(\lambda)$. There will be many different diagrammatic contributions to $\mathbb{F}$. We can organize them by the number of $\mathbb{R}$ insertions: $k-1$ for $k=1,2 \ldots$. That is $\mathbb{F} = \sum_k \mathbb{F}^{(k)}$. Note that $k$ is also the number of external puddles in each diagram: 
\begin{equation}
	\includegraphics[width=.85
	\textwidth]{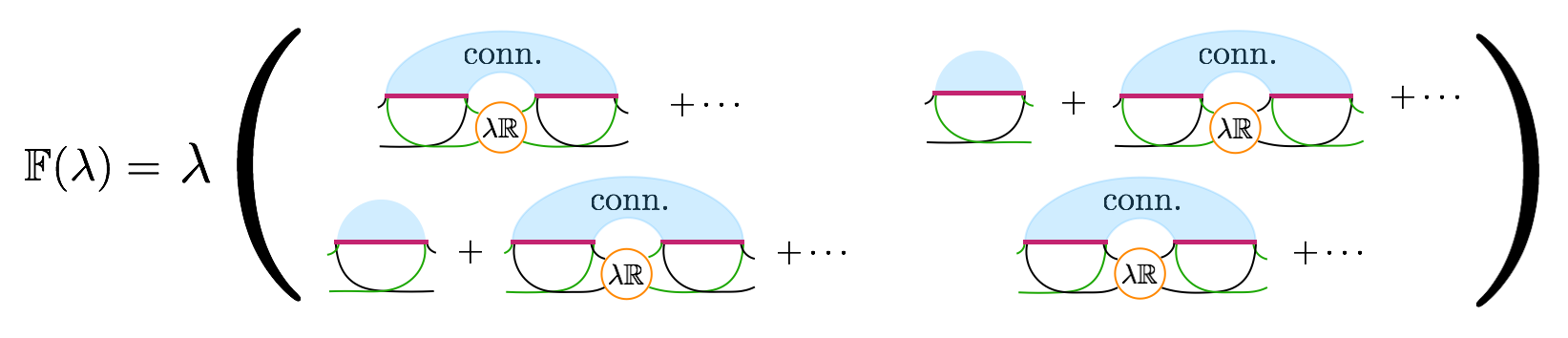} 
\end{equation}
where we have used the shorthand \textit{conn.} to indicate all the connected contractions.

The lowest order contribution ($k=1$) for $\mathbb{F}$ only features diagrams with a single disconnected puddle: 
\begin{equation}\label{eq:SD-order0} 
	\mathbb{F}^{(1)} = D_m\sqrt{\lambda} 
	\begin{pmatrix}
		0 & \ket{\epsilon_A}\bra{\epsilon_B} \\
		\ket{\epsilon_B}\bra{\epsilon_A} & 0 
	\end{pmatrix}
	, 
\end{equation}
where the number $D_m$ is defined as 
\begin{equation}\label{eq:dm} 
	\includegraphics[width=.8
	\textwidth]{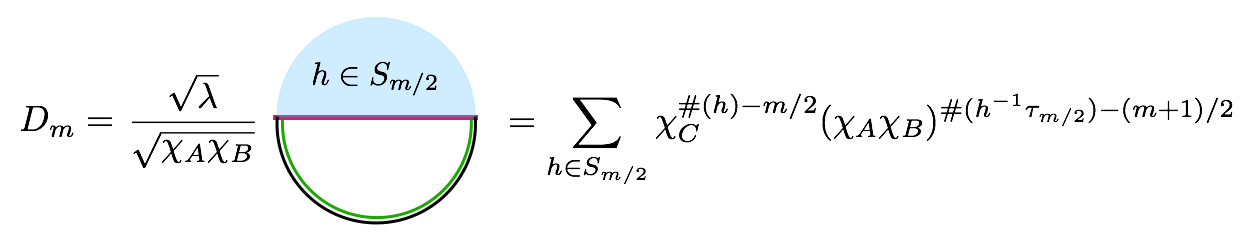}. 
\end{equation}
At large $\chi$, this sum is dominated by non-crossing permutations and we have $D_m\simeq \chi_C^{(m+1)/2} \sqrt{q} C_{m/2}(1/q) \sim O(\chi^{(m+1/2)}_C)$. Using \Eqref{eq:SD-defining} we can find the lowest order solution for the matrix $R$: 
\begin{align}\label{eq:R_sol_leading} 
	R_{ij} = \frac{1}{\lambda-D^2_m} 
	\begin{pmatrix}
		1 & \lambda^{-1/2}D_m \\
		\lambda^{-1/2}D_m & 1 
	\end{pmatrix}_{ij} + \ldots
\end{align}
This represents the single eigenvalue contribution we had originally around \Eqref{D1D0} by summing elements.  We will find significant corrections to this pole from the next order.

At the second order ($k=2$), we have diagrams that contain exactly one $\mathbb{R}$ insertion. Their contributions split into two parts according to the pattern of contractions between the first and last legs of each puddle. That is if we follow the contractions in the resulting diagram - following lines above and below puddles - the pattern of outer contractions tracks where the first and last legs connect to in the rest of the diagram.

For example, let's take a look at the diagonal component $\mathbb{F}^{(2)}_{11}$: 
\begin{equation}
	\includegraphics[width=.95
	\textwidth]{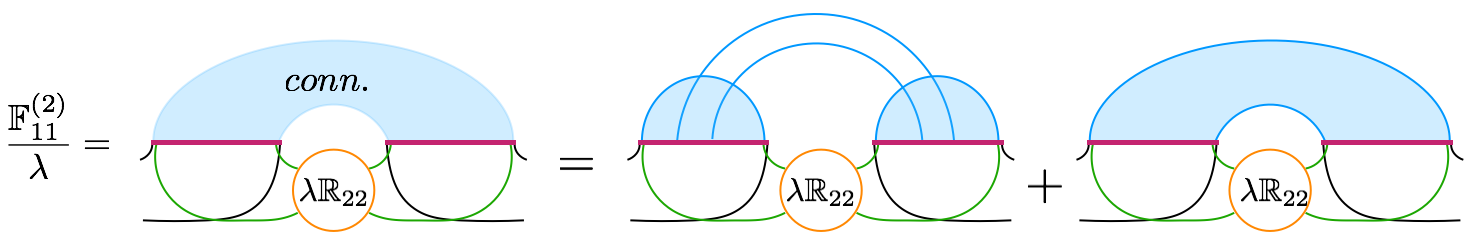}. 
\end{equation}
where the bold blue lines after the second equality represent the pattern/topology of outer contractions. By using this pattern to follow the green and black curves below the puddle, we can see if the contribution should be either proportional to the identity and thus contributing to $(G_{11})$ or a projector, contributing to $(F_{11})$. On the other hand, for the off-diagonal part of $\mathbb{F}^{(2)}$, both permutation types only contribute to the projector, e.g., 
\begin{equation}
	\includegraphics[width=
	\textwidth]{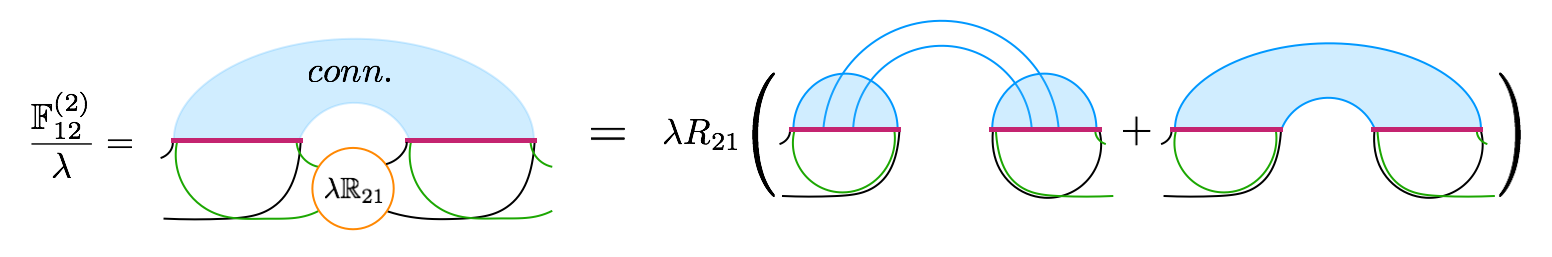} 
\end{equation}

For both $\mathbb{F}^{(2)}_{11}$ and $\mathbb{F}^{(2)}_{12}$, the first diagram shown above necessarily involve elements $h \in S_m$ that are crossing or higher genus, so these are naturally suppressed by powers of external bond dimensions.

We can compute these diagrams to find $\mathbb{F}^{(2)}$ written as: 
\begin{align}
	\begin{split}
		(\chi_A \chi_B)^2 G_{11}^{(2)} &= \lambda ( S_{22} (\chi_B^2-1) + R_{22}) B_m \,, \\
		(\chi_A \chi_B)^2 G_{22}^{(2)} &= \lambda ( S_{11} (\chi_A^2-1) + R_{11}) B_m\,, 
	\end{split}\label{eq:SD_order1} 
	\\
	\frac{(\chi_A \chi_B)^2}{\lambda} 
	\begin{pmatrix}
		F_{11} & F_{12} \\
		F_{21} & F_{22} 
	\end{pmatrix}
	^{(2)} &= 
	\begin{pmatrix}
		R_{22} & R_{21} \\
		R_{12} & R_{11} 
	\end{pmatrix}
	E_m + 
	\begin{pmatrix}
		S_{22} (\chi_B^2-1) & 0 \\
		0 & S_{11} (\chi_A^2-1) 
	\end{pmatrix}
	B_m , 
\label{eq:SD_order2} \end{align}
where the numbers $B_m$ and $E_m$ are defined by the following diagrams 
\begin{equation}\label{eq:bm} 
	\includegraphics[width=0.9
	\textwidth]{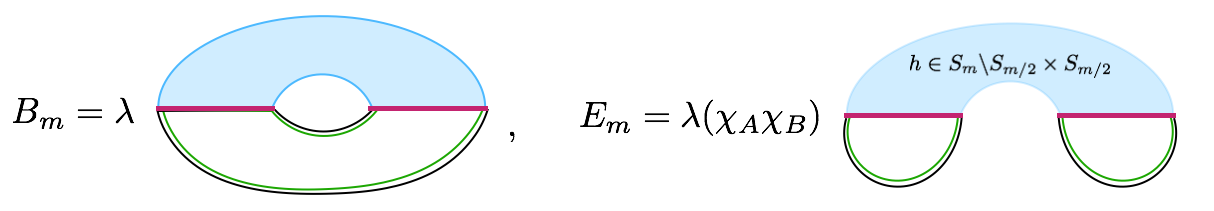} 
\end{equation}
More explicitly, we have 
\begin{align}
	B_m &= \frac{1}{(\chi_A\chi_B)^2-1} \left( -E_m + \sum_{h\in S_m\setminus S_{m/2}\times S_{m/2}} \chi_C^{\#(h)-m}(\chi_A\chi_B)^{\#(h^{-1}\tau_m)-m+2}\right) \\
	E_m &= \sum_{h\in S_m\setminus S_{m/2}\times S_{m/2}} \chi_C^{\#(h)-m}(\chi_A\chi_B)^{\#(h^{-1}\tau_{m/2}\times\tau_{m/2})-m+1}, 
\end{align}
where we have used the identity 
\begin{equation}
	\includegraphics[width=0.6
	\textwidth]{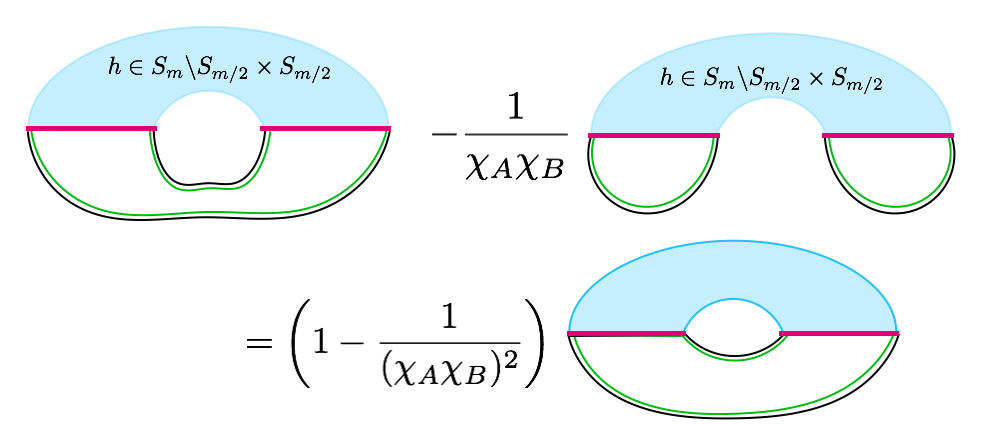}, 
\end{equation}
which follows from the fact that crossing permutations are projected out from the LHS.

We now conjecture that the main contribution to the matrix $\mathbb{F}$ comes from lowest two orders in $k$, i.e. $\mathbb{F}\simeq \mathbb{F}^{(1)}+\mathbb{F}^{(2)}$, in large bond dimensions. Unfortunately this truncation does not obviously follow from a genus counting argument.  We base our conjecture on four pieces of evidence: First, we explicitly calculated the contribution from $k=3$ and a special class of $k=4$ diagrams and found that they only give corrections to higher orders of $\chi$. Secondly we power counted a particular class of diagrams for general $k$ (which we believe to be the dominant contributions for $k$) and find it to be suppressed by powers of $\chi$.
Thirdly, although our diagrammatic approach is not entirely the same as the direct summation of group elements in App.~\ref{app:2nd_resum}, 
the two approaches give results that differ only by sub-leading $\chi$ corrections. 
Finally, we numerically evaluated the reflected entropy and eigenvalue spectrum for a Haar random state and we indeed find good match to the analytical results obtain from this truncation. The numerical results will be discussed in Section~\ref{sec:num}.

We rescale various quantities to restore the correct normalization from \Eqref{eq:reflected_moments} 
\begin{equation}
	G \rightarrow G Z_1\,, \qquad F \rightarrow F Z_1\,, \qquad S \rightarrow S /Z_1 \,, \qquad R \rightarrow R /Z_1 \,, \qquad \lambda \rightarrow \lambda Z_1 
\end{equation}
in which all the above equations take the same form but with $\widehat{B}_m = B_m/Z_1$, $\widehat{E}_m = E_m/Z_1$ and $\widehat{D}_m = D_m/Z_1^{1/2}$ and these hatted quantities are now all $\mathcal{O}(1)$ at large $\chi_C$ with $q$ fixed. We have the sum rule: 
\begin{equation}
	\widehat{D}_m^2 + \left( 1-\frac{1}{(\chi_A\chi_B)^2} \right) \widehat{B}_m + \frac{1}{(\chi_A\chi_B)^2} \widehat{E}_m = 1 
\end{equation}

We now attempt to solve for the system of matrix equations \Eqref{eq:SD-defining}, \Eqref{eq:SD_order1}-\Eqref{eq:SD_order2}. We can completely solve this system if we make one approximation, which we will later check is self-consistent. We will assume that we can drop the ``back-reaction'' of $R$ on \Eqref{eq:SD_order1} that determined $S_{11}, S_{22}$. We have 
\begin{align}
	(\chi_A \chi_B)^2 (\lambda - S_{11}^{-1}) = \lambda S_{22} (\chi_B^2-1) \widehat{B}_m \,, \quad (\chi_A \chi_B)^2 (\lambda - S_{22}^{-1}) = \lambda S_{11} (\chi_A^2-1) \widehat{B}_m 
\end{align}
The solution is 
\begin{equation}\label{eq:MP-main} 
	(\chi_A^2 - 1) S_{11} , (\chi_B^2 - 1) S_{22} = \frac{ \chi_A^2 \chi_B^2}{2 \widehat{B}_m \lambda } \left( \lambda - \sqrt{ (\lambda - \lambda_+) (\lambda - \lambda_-)} \right) \pm \frac{ \chi_A^2 - \chi_B^2}{2 \lambda} 
\end{equation}
where 
\begin{equation}
	\lambda_{\pm} = \frac{\widehat{B}_m}{\chi_A^2 \chi_B^2} \left( \chi_A^2 + \chi_B^2 - 2 \pm 2 \sqrt{ (\chi_A^2 - 1)(\chi_B^2 -1)} \right) 
\end{equation}
This is a Marchenko-Pastur distribution with support between $\lambda_{\pm}\simeq\widehat{B}_m(\chi_A^{-1}\pm \chi_B^{-1})^2$ and $\min(\chi_A^2-1,\chi^2_B-1)$ eigenvalues. We can check that our initial assumption for solving $S_{11}, S_{22}$ is indeed valid since $R_{11}\sim O(1)$ and $(\chi^2_A-1)S_{11}\sim \chi^2_B\chi^2_A$. The validity of this assumption breaks down at $\lambda \gg \chi^{-2}_A, \chi^{-2}_B$, but as we have seen already, this is well outside the spectral weight of the MP distribution and so does not effect the spectrum that we find.

The second equation \Eqref{eq:SD_order2} determines correction to the leading order solution of matrix $R_{ij}$ \Eqref{eq:R_sol_leading} we obtained earlier. Together with \Eqref{eq:SD-defining} and \Eqref{eq:SD-order0}, this yields a quadratic equation of $2\times 2$ matrices. We can solve this equation completely, although the algebra is a bit more involved and the it only features corrections of $O(\chi^{-2})$ and higher orders. For this reason we only summarize our findings here: 
\begin{enumerate}
	\item The position of the single pole at $\lambda=\widehat{D}^2_m$ is now shifted to 
	\begin{equation}
		\lambda=\widehat{D}^2_m + \left( \frac{1}{\chi^2_A} + \frac{1}{\chi^2_B}\right) \widehat{B}_m + O(\chi^{-4}) 
	\end{equation}
	\item The same pole is now resolved to a small peak of width 
	\begin{equation}
		\delta\lambda \sim \frac{8\widehat{D}_m \sqrt{\widehat{E}_m}}{\chi_A\chi_B}+O(\chi^{-4}) 
	\end{equation}
\end{enumerate}
For the detailed form of $R$ and its derivation please refer to appendix \ref{app:SD-R}. As a final remark, we stress here that the various techniques we used here, such as the truncation at $k=2$ and various approximations we made to solve the SD equation, rely on the limit of large bond dimension $\chi_A,\chi_B,\chi_C\to \infty$. However we do not require all the external bond dimensions to be large for our analysis to work. In fact there are two interesting parameter ranges to take: 
\begin{enumerate}
	\item $\chi_A,\chi_B,\chi_C\to\infty$ with $q=\chi_A\chi_B/\chi_C$ and $r=\chi_A/\chi_B$ fixed: \\
	This is the main parameter range of our interest in this paper. The resolvent has two interesting regimes. For $\lambda \sim 1/\chi_C$ we simply find the MP distribution 
	\begin{equation}
		R (\lambda) = \frac{(\chi_A\chi_B)^2}{2\widehat{B}_m\lambda} \left( \lambda-\sqrt{(\lambda-\lambda_+)(\lambda-\lambda_-)} \right) + \frac{\chi^2_A-\chi^2_B}{2\lambda} 
	\end{equation}
	And for $\lambda =\widehat{D}_m \sim O(\chi_c^0)$ we find a simple pole. This pole is resolved into a mini-peak each with width $\delta \lambda \simeq 8\widehat{D}_m\sqrt{\widehat{E}_m}\chi_C^{-1}$ for finite $\chi_C$.
	
	\item $\chi_A,\chi_C\to \infty$ with $\chi_A/\chi_C$ and $\chi_B$ fixed: \\
	The main MP distribution becomes narrow, approximating a pole centered at $\lambda_{\pm} = \chi_B^{-2}$ with weight $\chi_B^2 -1$. The other single pole is located at $\lambda = \widehat{B}_m \chi_B^{-2} + \widehat{D}_m^2$: 
	\begin{equation}
		R(\lambda) = \frac{\chi_B^2-1}{\lambda- \widehat{B}_m \chi_B^{-2} } + \frac{1}{ \lambda - \widehat{B}_m \chi_B^{-2} - \widehat{D}_m^2} 
	\end{equation}
	The first, smaller peak, is resolved into a peak with width $\delta\lambda \simeq 8\widehat{D}_m\sqrt{\widehat{E}_m}(\chi_A\chi_B)^{-1}$, and the larger peak is resolved into an MP distribution of width: 
	\begin{equation}
		\delta \lambda =\frac{1}{\chi_A} \frac{2 \widehat{B}_m \sqrt{\chi_B^2 -1}}{\chi_B^2} 
	\end{equation}
	This regime is not relevant to the random tensor model we study in this section. However it will prove to be useful for a 2-site random tensor model that aims to better model the holographic phase transition, which we will aim to discuss in our upcoming work \cite{pagereflected}. 
\end{enumerate}

\subsection{Spectrum and reflected entropy} \label{sec:spectrum} We have seen how to use diagrammatic methods to construct the resolvent of the reflected entropy. We summarize what we have found here:
\begin{equation}\label{eq:ref_spectrum} 
	R(\lambda) = \frac{(\chi_A\chi_B)^2}{2\lambda}\left( \frac{\lambda}{\widehat{B}_m} + \left( \frac{1}{\chi^2_B}-\frac{1}{\chi^2_A} \right) - \sqrt{(\lambda/\widehat{B}_m-\lambda_+ )(\lambda/\widehat{B}_m-\lambda_-)} \right) + \frac{1}{\lambda-\widehat{D}^2_m} 
\end{equation}
with 
\begin{equation}
	\lambda_\pm = \frac{1}{(\chi_A\chi_B)^2}\left( \sqrt{\chi^2_A-1} \pm \sqrt{\chi^2_B-1} \right)^2 
\end{equation}
We have given evidence that this form is asymptotically correct under the limit $\chi_C\to \infty$ with $q=\chi_A\chi_B/\chi_C$ and $r=\chi_A/\chi_B$ fixed. For higher order corrections to the resolvent refer to the previous section and Appendix~\ref{app:SD-R}. Under the limit of large bond dimension, the permutation sums $\widehat{B}_m$ and $\widehat{D}_m$ are dominated by non-crossing permutations and can be expressed in terms of q-Catalan numbers\footnote{ The q-Catalan numbers generalize the Catalan numbers away from $q=1$. They can be written in terms of Hypergeometric functions. For the complete definition see Appendix~\ref{app:perm}.}. 
\begin{equation}\label{eq:weights} 
	\widehat{D}^2_m \simeq \frac{C_{m/2}(q^{-1})^2}{C_m(q^{-1})} \equiv p_0, \quad \widehat{B}_m \simeq \frac{C_m(q^{-1})-C_{m/2}(q^{-1})^2}{C_m(q^{-1})} \equiv p_1 
\end{equation}
We now have the simplified sum rule: 
\begin{equation}
	p_0 + p_1 = 1 
\end{equation}
We will see that these quantities act like classical probabilities, which control the phase transition between the connected and disconnected entanglement wedges, with $q$ being the tuning parameter. We have the following asymptotic expressions: 
\begin{equation}
	(p_0,p_1) \sim 
	\begin{cases}
		(1-m^2q/4,m^2q/4), \quad &q \ll 1 \\
		(1/q, 1-1/q), \quad & q \gg 1 
	\end{cases}
\end{equation}
and around $q\sim 1$ the two probabilities are both order $1$.
\begin{figure}
	[h] \centering 
	\begin{overpic}
		[width=0.5\linewidth]{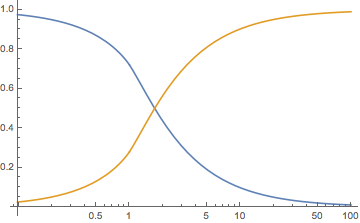} \put(50,-3){$q$} \put(80,50){$p_1(q)$} \put(80,10){$p_0(q)$} 
	\end{overpic}
	\vspace{0.5cm} 
	\caption{The plot showing the trend of $p_i(q)$. We have taken $m=1$ here.} 
\label{fig:p_i} \end{figure}

The reflected entropy that follows from this resolvent is given by (assuming $\chi_A<\chi_B$) 
\begin{align}
	S_R 	&= -p_0\ln p_0 -p_1\ln p_1 + p_1 \left( \ln \chi^2_A - \frac{\chi^2_A}{2\chi_B^2} \right) + O(\chi^{-2}) 
\label{eq:main_result} \end{align}
This is the main result of this section.
Note that the entropy has a contribution from  the classical Shannon entropy of a single bit $p_i$ plus a term proportional to what seems to be the entropy of a Haar random state on $AA^\star BB^\star$ reduced to $AA^\star$ (the standard Page result). This contribution however is multiplied by $p_1$.

We can better understand the physics by looking at the eigenvalue spectrum:
\begin{align}
	D(\lambda) = \frac{(\chi_A\chi_B)^2}{2\pi \lambda}\sqrt{(\lambda_+-\lambda/p_1)(\lambda/p_1-\lambda_-)} + \delta(\lambda)(\chi^2_A-\chi^2_B)\theta(\chi^2_A-\chi^2_B) + \delta(\lambda-p_0) 
	\label{evalspec}
\end{align}

\begin{figure}
	[h] \centering 
	\begin{overpic}
		[width=0.5\linewidth]{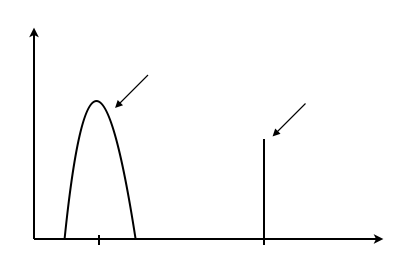} \put(4,62){$D(\lambda)$} \put(95,6){$\lambda$} \put(62,2){$p_0$} \put(18,2){$p_1/\chi^2_A$} \put(37,48){$\chi^2_A-1$ eigenvalues} \put(75,39){single eigenvalue} 
	\end{overpic}
	\caption{The spectrum obtained from the resolvent \Eqref{evalspec}. We take $\chi_A<\chi_B$ here so the MP distribution has no zero eigenvalues.} 
\label{fig:spectrum} \end{figure}

We see that the  spectrum is given by a shifted Marchenko-Pastur distribution with a total of $\min(\chi_A^2-1,\chi^2_B-1)$ eigenvalues, plus a single eigenvalue located at $\lambda=p_0$  
The R\'enyi entropies can be reconstructed by the sum of the moments of the two pieces in resolvent 
\begin{equation}
	S_n(AA^*) = \frac{1}{1-n} \ln \left( p^n_0 + \chi^{2(1-n)}_A r^{-2}C_n(r^2) p^n_1 \right) 
\end{equation}
where $r=\chi^2_A/\chi^2_B$. Note that the $(m,n)$ dependence has essentially factorized: the only $m$ dependence is through the probabilities $p_0(m),p_1(m)$.

The R\'enyi entropies have three different approximate behaviors based on the relative ratio of the external bond dimensions. We have: 
\begin{equation}
	S_n(AA^\star) \approx 
	\begin{cases}
		\frac{n}{1-n} \ln p_0 \approx 0, \, & \chi_{AB} \ll \chi_C \\
		\frac{n}{1-n} \ln p_0 \approx\frac{n}{n-1} \ln \frac{\chi_{AB}}{\chi_C}, \, & \chi_C \ll \chi_{AB} \ll \chi_C \chi_A^{ \frac{2(n-1)}{n}} \\
		\approx 2 \ln \chi_B, \, & \chi_{AB} \gg \chi_C \chi_B^{ \frac{2(n-1)}{n}} 
	\end{cases}
\end{equation}
and we have assumed $\chi_A \leq \chi_B$ and $n > 1$. These three different behaviors exactly match the three phases we found in \secref{sec:phase_nm}, and their transitions match the phase boundary lines in \figref{fig:1site_phase}. As we take $n\to 1$ the middle regime vanishes and we get back the single transition in the reflected entropy as expected, as we move from disconnected to connected entanglement wedges.

\subsection{Effective description}
\label{sec:eff}

Based on the results of the previous subsection, we now give an alternative effective description of the canonically purified state and its $m$ generalization:  $\ket{\rho^{m/2}_{AB}}\bra{\rho^{m/2}_{AB}}$.
Consider the following pure state on the Hilbert space $AA^\star BB^\star$: a superposition of a factorized state (with probability $p_0$) plus a random maximally entangled state (with probability $p_1$):
\begin{equation}
\label{psieff}
\left| \psi_{\rm eff} \right>
= p_0^{1/2} \left| \chi_0 \right> + p_1^{1/2} \left| \chi_1 \right>
\end{equation}
Here $\left| \chi_0\right> =  \left| 1_{A}/\chi_A \right>_{AA^\star} \otimes \left| 1_{B}/\chi_B \right>_{BB^\star}$ is the factorized state 
and we construct $\left| \chi_1 \right>$ using a random tensor network state as follows. We use a new network geometry that comes from doubling across the connected entanglement wedge:
\begin{equation}
\label{newnetwork}
\left| \chi_1 \right> \propto \left< 1_{C} \right| U_{ABC} U'_{A^\star B^\star C^\star} \left| 0\right>_{ABC} \left| 0\right>_{A^\star B^\star C^\star}
\end{equation}
where $\left| 1_{C} \right>$ is the (canonical) un-normalized maximally entangled state on the $CC^\star$ Hilbert space.
It is natural to pick $U_{ABC} \left|0\right>_{ABC} $ to be the same  random state that we started with, but we choose $U'_{A^\star B^\star C^\star}$ to be an independent Haar random unitary.  We will comment on this choice further below. See \figref{fig:sumsum} for a picture of these states.

\begin{figure}
\centering
		\includegraphics[width=0.8
		\textwidth]{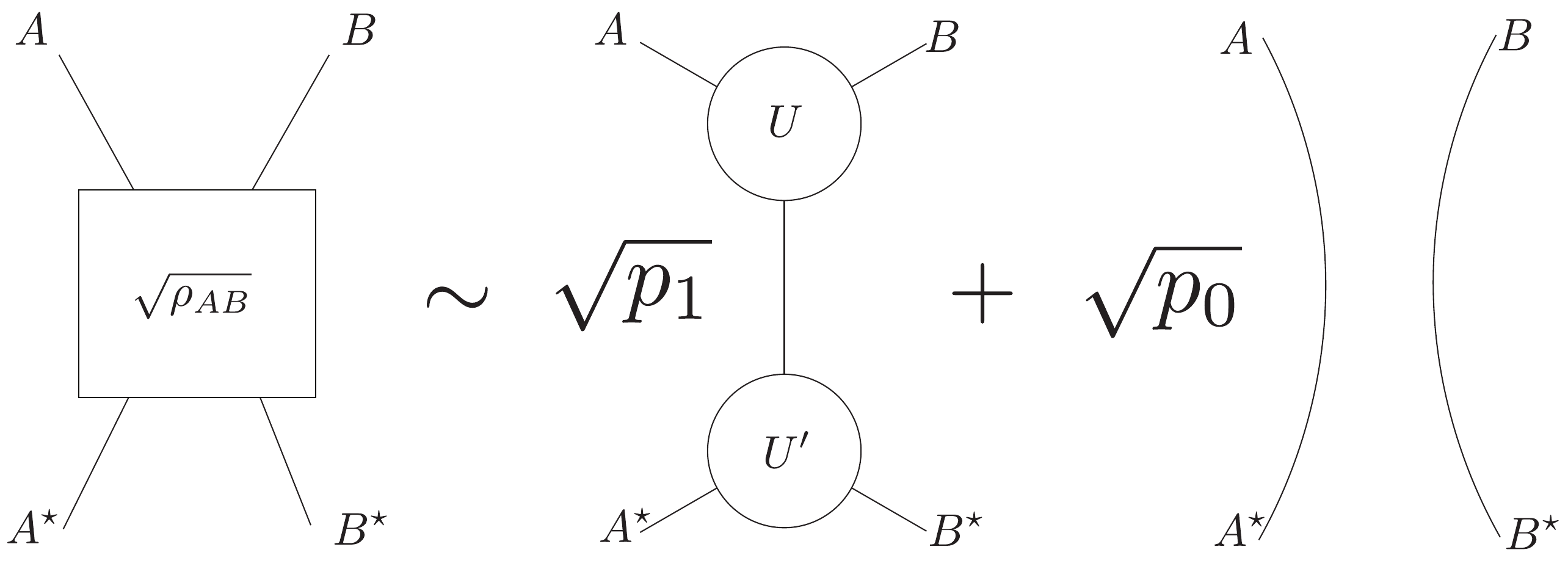} 
	\caption{\label{fig:sumsum} An illustration of \Eqref{psieff}. We argue that the canonical purification is effectively described by a sum over tensor network states. 
	The first state comes from the natural doubling procedure outlined in Ref.~\cite{Marolf:2019zoo} for the connected entanglement wedge of the original network in \figref{fig:single} and the second is the trivial factorized network which is obtained by doubling the trivial/disconnected entanglement wedge of the original network.} 
\end{figure}

Note that $\ket{\chi_0}$ is
the canonical purification of a maximally mixed state $\rho_{AB} = 1_{AB}/(\chi_A \chi_B)$ - this is the density matrix $\rho_{AB}$ that usually arises from a Haar random state reduced to $AB$ when $\chi_C > \chi_A \chi_B$. 
And, $\ket{\chi_1}$ is the natural guess for the canonical purification of the original Haar random state reduced to $AB$ when $\chi_C < \chi_A \chi_B$: that is, when the entanglement wedge is connected.
The doubling procedure is motivated by the AdS/CFT discussion in \cite{Dutta:2019gen}.

The claim is that $\ket{\psi_{\rm eff}}$ has the same entanglement structure as $\left| \rho_{AB}^{m/2} \right>$, at large bond dimension and for all $m$ (note that $p_{0,1}$ depend on $m$). In particular we will see that it has the same $AA^\star$ entanglement spectrum. Similarly, observe that the two states $\ket{\chi_{0,1}}$, when working in their respective phases,  give the same expectation values for operators in $AB$ as the original density matrix $\rho_{AB}$. Thus the entanglement structure for sub-regions of $AB$ is also maintained, at least away from the phase transition.
Since we think of the entanglement structure as being closely linked to the emergent bulk geometry, we can think of $\ket{\psi_{\rm eff}}$ as capturing the effective geometry of the canonical purification.

Such superpositions \Eqref{psieff} over different tensor networks, have been postulated as models of gravitational states in AdS/CFT \cite{Akers:2018fow, Dong:2018seb, Harlow:2016vwg,Bao:2018pvs,Qi:2017ohu}, see also \cite{Almheiri:2016blp}. In particular these states allow for non-trivial fluctuations in the area of the RT surface, and do not suffer from the issue of a flat entanglement spectrum that is not expected in typical holographic states. Each wavefunction in the superposition is then thought of as a fixed area state, where the gravitational state is projected onto approximate eigenstates of the area operator. 

More specifically, on the physical Hilbert space the reduced density matrices on $AA^\star$ 
will have approximately orthogonal support:
\begin{equation}
{\rm supp} ( (\rho_0)_{AA^\star} ) \perp {\rm supp}( (\rho_1)_{AA^\star})\, \qquad {\rm supp} ( (\rho_0)_{BB^\star} ) \perp {\rm supp}( (\rho_1)_{BB^\star})
\end{equation}
So the resulting states behave like approximate superselection sectors with respect to $AA^\star$: the phase between the two components in the wavefunction is unobservable when restricting to $AA^\star$ or $BB^\star$. This turns out to be  approximately true for $\ket{\chi_0},\ket{\chi_1}$, in particular because of how we chose $U'_{A^\star B^\star C^\star}$ in \Eqref{newnetwork}: we picked an independent random unitary not equal to $(U_{ABC})^\dagger$. This latter choice might have seemed more natural, considering the $m=2$ case $\left| \rho_{AB} \right>$ exactly gives such a network, albeit without the disconnected wave-function $\ket{\chi_0}$. Note that the choice $U' = U^\dagger$ leads to correlations between different tensors in the doubled tensor network. Our claim here is that for the effective state, such correlations have already been taken into account by $\ket{\chi_0}$ so we should not double count this effect. This leads us to \Eqref{newnetwork}.

Given the discussion above, we can easily compute the entropy $S({AA^\star})$ of $\ket{\psi_{\rm eff}}$ and indeed it agrees with \Eqref{eq:main_result}. Similarly the R\'enyi entropies also agree.

Let us push this interpretation a little further, and give an quantum error correction interpretation of this superposition of tensor networks states, and emergent area operator.  
For simplicity let us assume the density operators have exactly orthogonal support. We introduce an area operator that labels the different superselection sectors.
The superselection sectors in
this case are described by a single bit $s=0, 1$. The area operator is: 
\begin{equation}\label{eq:superselect} 
	\mathcal{L}_{AA^\star}= \mathop{\oplus}_{s =0,1} \overline{S(\chi_s)} 
\end{equation}
accounting for the entropy of each sector.
The two orthogonal subspaces, determined by projections $\pi_0, \pi_1$, are defined using the supports: 
\begin{equation}
\pi_0 = {\rm supp}(\rho^0_{AA^\star})
\otimes {\rm supp}(\rho^0_{BB^\star})
\end{equation}
and similarly for $\pi_1$. Then:
\begin{equation}
\mathcal{L}_{AA^\star} \left| \chi_s \right> = \overline{S(\chi_s)} \left| \chi_s \right>
\end{equation}
In this case we can define a simple quantum/classical error correcting code, that protects a single classical bit from quantum errors. Define the code subspace as a single bit and the isometry $V = \left| \chi_0 \right> \left< 0 \right| + \left| \chi_1 \right> \left< 1 \right|$. This code protects from errors on either $AA^\star$ or $BB^\star$, which can be confirmed by computing:
\begin{equation}
V^\dagger (1_{AA^\star} \otimes \mathcal{O}_{BB^\star}) V
 =\left| 0 \right> \left< 0 \right| \left< \chi_0 \right| \mathcal{O}_{BB^\star} \left| \chi_0 \right>
 + \left| 1 \right> \left< 1 \right| \left< \chi_1 \right| \mathcal{O}_{BB^\star} \left| \chi_1 \right>
\end{equation}
So the error, represented by an arbitrary operator on $BB^\star$: $\mathcal{O}_{BB^\star}$, does no damage to the classical information. This is the Knill-Laflamme condition stated for operator algebra error correction with complementary recovery and a center, see Refs.~\cite{Harlow:2016vwg,Faulkner:2020hzi} (in this case the center is everything in the code.) 

We know how to compute the entropy for states on such a code, following \cite{Harlow:2016vwg} we find:
\begin{equation}
S_{\psi_{\rm eff}}(AA^\star)
= - p_0 \ln p_0 - p_1 \ln p_1 + \left< \psi_{\rm eff} \right| \mathcal{L}_{AA^\star} \left| \psi_{\rm eff} \right>
\end{equation}
as expected. 
Note that, in reality the supports of the reduced density matrices are not exactly orthogonal. This is because the error correcting code is not exact. 

\subsection{Numerical results} 
\label{sec:num}
Our main result \Eqref{eq:main_result} corrects the naive holographic reflected entropy, from a step function to a smooth transition. In this subsection, we corroborate these corrections, comparing our answer to numerical results, showing they indeed capture the details of the phase transition. All numerical results are obtained by generating a random tripartite state of appropriate bond dimension and numerically computing its reflected entropy.

First, in \figref{fig:spec_numerics} we plot a histogram of (the logarithm of the) eigenvalues of $\rho_{AA^\star}$, corroborating \Eqref{evalspec} and \figref{fig:spectrum}. We choose bond dimensions $\chi := \chi_A = 23, \chi_B = 25, \chi_C = 11,$ and $50$ trials.

\begin{figure}[h!] \centering 
	\includegraphics[width=0.90
	\textwidth]{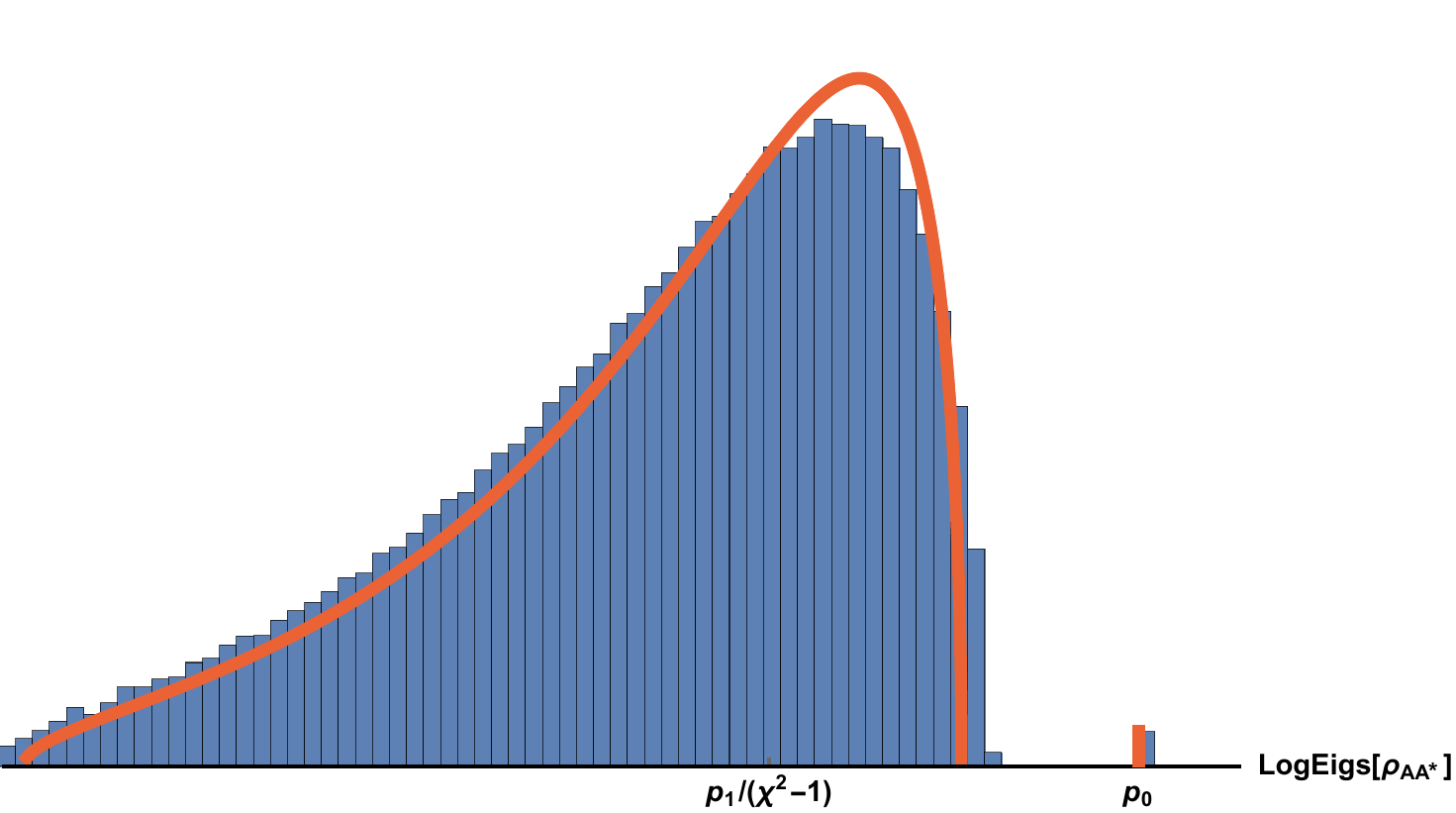} 
	\caption{The blue histogram is a numerical plot of the spectrum of $\rho_{AA^\star}$ for a tripartite tensor. The orange curve is our analytic prediction for the spectrum \Eqref{evalspec}. Note that the tick values on the x-axis are not exactly $p_0$ from \Eqref{eq:probRTN}, but instead include a small correction, replacing $p_0$ with its shifted version \Eqref{eq:shifted_pole}, differing only at $\mathcal{O}(\log(\chi)/\chi^2)$. This gives the value in the plot of $p_0 = 0.226.$} 
\label{fig:spec_numerics} \end{figure}

\begin{figure}[h!] \centering 
	\includegraphics[width=0.90
	\textwidth]{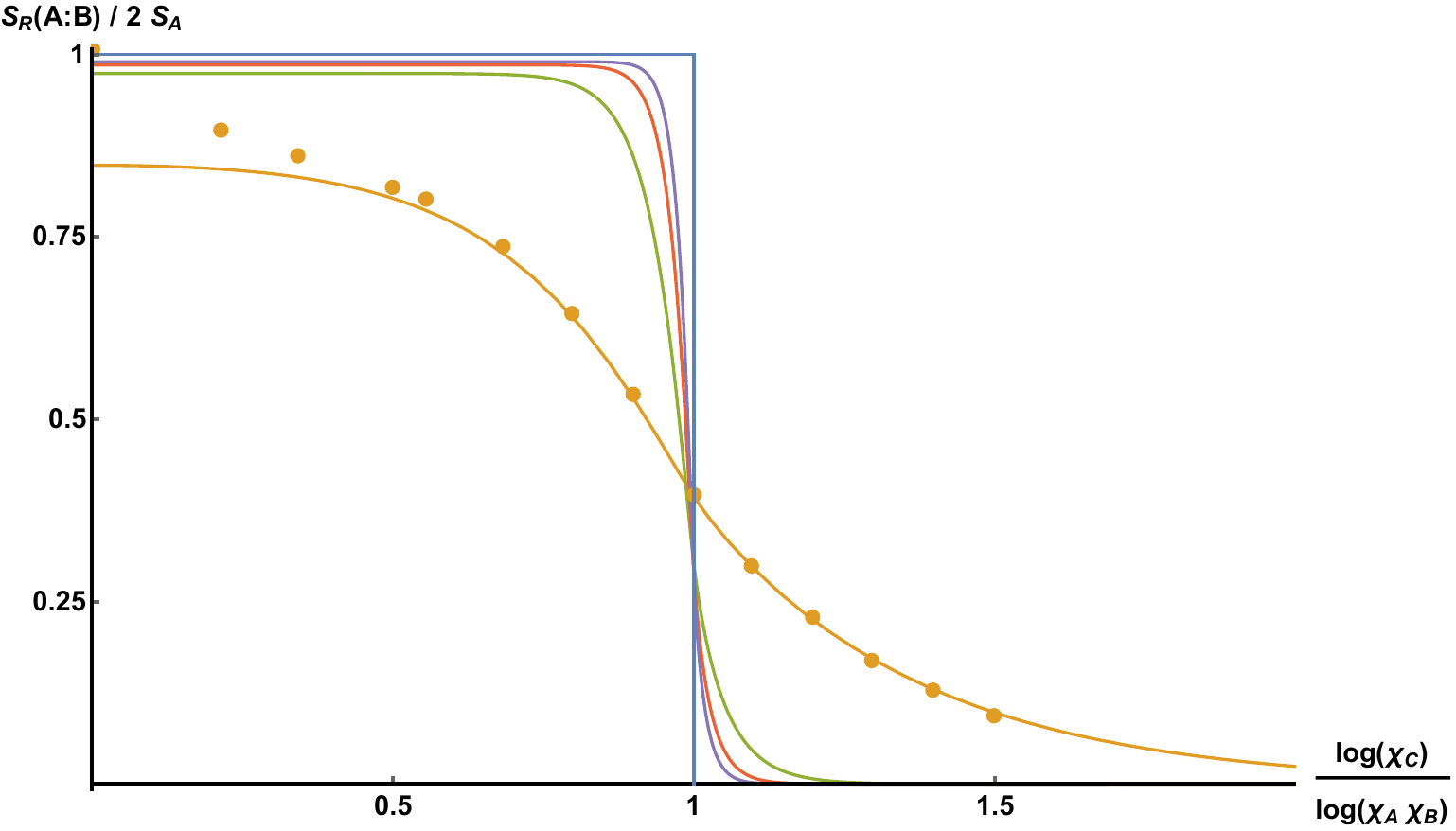} 
	\includegraphics[width=0.90
	\textwidth]{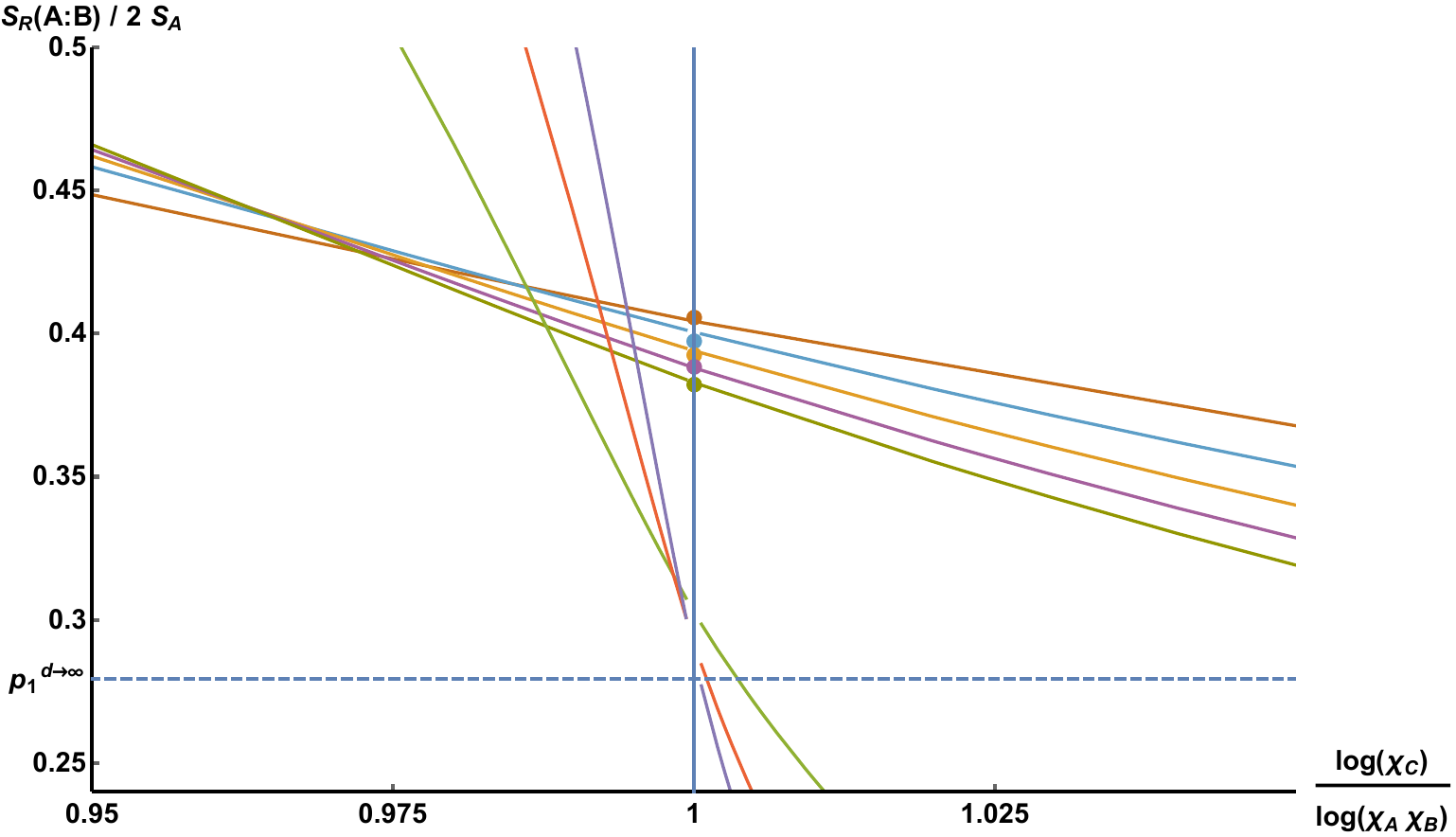} 
	\caption{Plots of the ``Page curve'' for reflected entropy. Our results capture the phase transition and agree with numerics. The disagreement in the top figure at small $\chi_C$ is to be expected from \Eqref{eq:main_result} not being valid at small bond dimension. Note that the analytic functions plotted are not exactly \Eqref{eq:main_result}, but instead include a small correction, replacing $p_0$ with its shifted version \Eqref{eq:shifted_pole}, differing only at $\mathcal{O}(\log(\chi)/\chi^2)$.} 
\label{fig:numerics} \end{figure}

In \figref{fig:numerics}, we present two plots of the ``Page curve'' for reflected entropy. Solid lines are analytic results \Eqref{eq:main_result}. Dots are numerical results (only obtained for small bond dimension). All values are normalized by twice the entanglement entropy of $A$, the upper bound on $S_R(A:B)$. In the top figure, the blue step function is the large-$\chi$ limit, with $x_A + x_B = \log(\chi_A \chi_B)/\log(\chi_C)$ held fixed. The other dimensions are $\log_5 \chi_A = \log_5 \chi_B = x_A \log_5 \chi_C = \{16, 11, 6, 1 \}$. The numerics agree, with larger deviation at large $x_A$. This is as expected, because large $x_A$ at fixed $\chi_A$ means small $\chi_C$, and \Eqref{eq:main_result} was derived in the large bond dimension limit. Besides illustrating that \Eqref{eq:main_result} agrees with small bond-dimension numerics and limits to the correct semiclassical answer in the limit $\chi \to \infty$, this plot also illustrates the novel prediction for $S_R(A:B)$ very near the phase transition.

Our results are depicted more precisely in the bottom figure: where we zoom into a narrower range of $x_A$ around the phase transition point. The horizontal blue dashed line depicts the predicted value of $S_R(A:B)/2 S_A$ in the limit $\chi \to \infty$. To demonstrate agreement with numerics, we have included curves and numerics corresponding to smaller dimensions, $\chi_A = \chi_B = \chi_C^{x_A} = \{3,4,6,7\}$.
It is evidence that \Eqref{eq:main_result} is sufficient to capture many of the non-trivial features of the transition.

\section{More general tensor networks}\label{sec:general_TN}

We have seen how the single tensor model, despite being relatively simple, still exhibits complicated behaviors such as a sharp jump in the reflected entropy. We also saw how this model should be thought of as describing the emergence of superselection sectors associated to an area operator for the entanglement wedge cross section. 
In this section we revisit the tensor network calculation for reflected entropy with the insights we gained from the single tensor. 
We focus on a regime away from the phase transition. 
We do not investigate the phase transition itself.

To begin our discussion, recall that in Section 2 we identified two possible bulk solutions (\figref{fig:naive_bulk_sol}). And indeed, we found their counterpart in the single tensor model, which corresponded to a phase with permutation element $g = e$ (the disconnected phase) and $g = g_A$ or $g=g_B$ (the connected phase.) One of the obvious issues we must confront for more general tensor networks is the possible appearance of new phases associated to $g=X$. 
In addition to this there are several new complications that we have to deal with relative to the single tensor case:

\begin{enumerate}
	\item \emph{Kinks in the domain walls:} \\
	In the EW phase, the tension of the EW cross-section $d(g_A,g_B) = 2(n-1)$ is non-zero whenever $n>1$. This will cause the cross-section to contract, resulting in two ``kinks'' where the two regions meet. The sharpness of the kink depends on both $m$ and $n$. See for example:
\begin{equation}
		\centering 
		\begin{overpic}
			[width=0.35
			\textwidth]{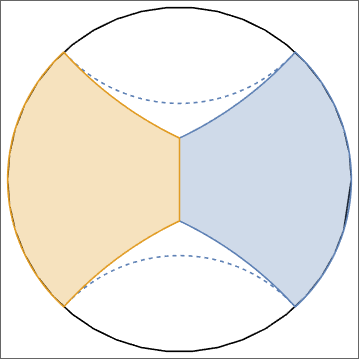} \put(25,50){$g_A$} \put(70,50){$g_B$} 
		\end{overpic}
\end{equation}
Since the tension $\propto (n-1)$, the larger $n$ is, the sharper the kink is. When $n=1$ the kink disappears and we recover the EW phase solution. 
	
	\item \emph{Multiple phases with element $X$:} \\
	Certainly we expect $X$ to make an appearance as it did in the single tensor calculation, see \secref{sec:phase_nm}. We will actually find several new phases where $X$ makes an appearance. To get some feeling for how $X$ can show up, assume that we have the following bulk configuration:
	\begin{equation}
	     \centering 
		\begin{overpic}
			[width=0.3
			\textwidth]{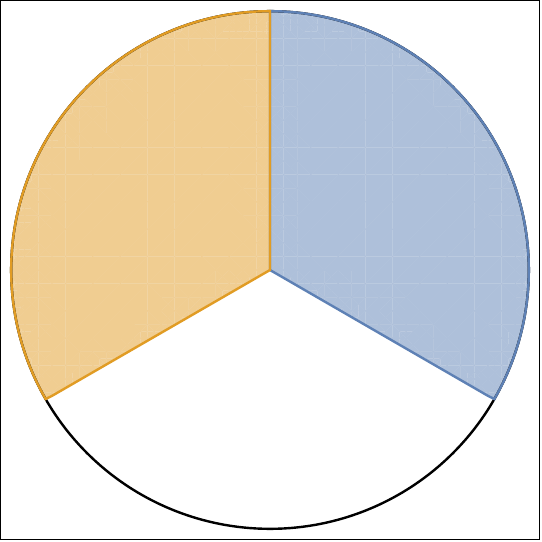} \put(25,60){$g_A$} \put(70,60){$g_B$} \put(48,20){$e$} 
		\end{overpic}
		\begin{overpic}
			[width=0.3
			\textwidth]{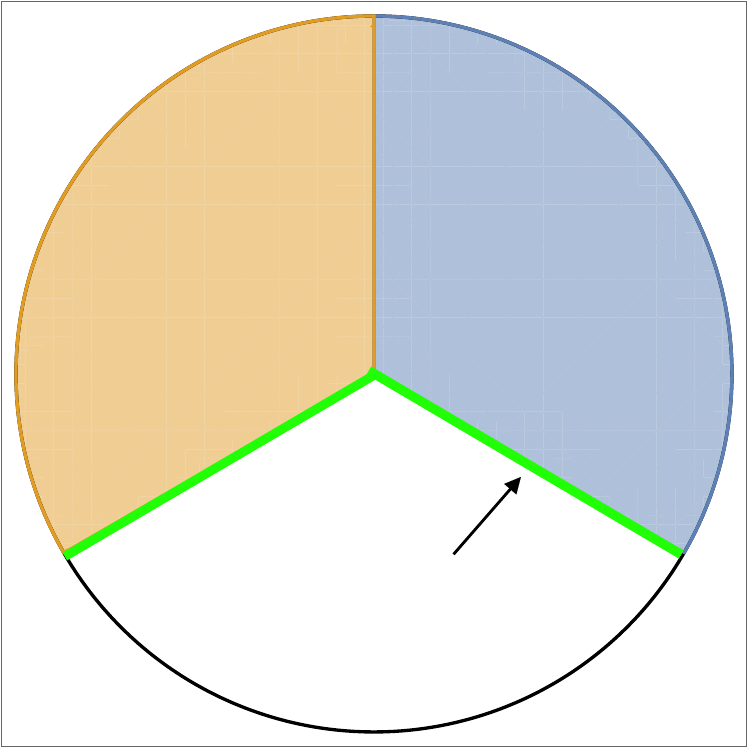} \put(25,60){$g_A$} \put(70,60){$g_B$} \put(55,20){$h$} 
		\end{overpic}
		\begin{overpic}
			[width=0.3
			\textwidth]{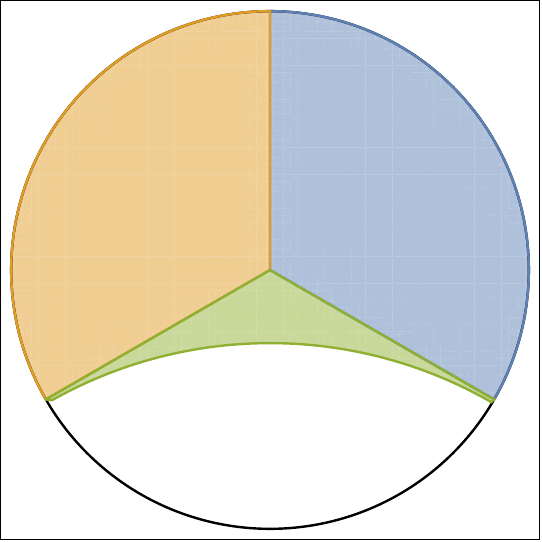} \put(25,60){$g_A$} \put(70,60){$g_B$} \put(48,40){$h$} \put(48,15){$e$} 
		\end{overpic}
	\label{fig:h_pocket} \end{equation}
	showing a a tripartite bulk configuration that can be deformed to include a new element $h$.
	For permutations $h$ that lie on \emph{both} geodesics $\Gamma(g_A,e)$ and $\Gamma(g_B,e)$, we have 
	\begin{align}
		\begin{split}
			d(g_A,h) + d(h,e) = d(g_A,e), \\
			d(g_B,h) + d(h,e) = d(g_B,e). 
		\end{split}
	\end{align}
	It then takes no energy to include a thin layer of $h$ at the domain wall where $g_A$ and $g_B$ transition to $e$. Next, imagine we pull out the thin layer to create a ``pocket'' of $h$. This will lower the overall energy since the new domain wall has shorter total length. It is easy to see that the change in the free energy will be proportional to $-d(h,e)$. So the optimal version of this process pulls out the element $h = X$: recall that by definition this element has the largest $d(h,e)$ amongst all $ h \in \Gamma(g_A,e) \cap \Gamma(g_B,e)$. We expect such a pocket to appear whenever the region $g_A$ and $g_B$ meet at an angle larger than $180$ degrees.
\end{enumerate}

Both phenomena together have an important effect on the bulk picture of the relevant phase. 
At $n>1$, the kink creates a non-flat angle at the EW cross-section, allowing a $h$-pocket to appear. Whenever this happens, it also shifts the domain wall from $g_{A/B}\leftrightarrow e$ to $g_{A/B}\leftrightarrow X$, therefore changing the sharpness of the kink. Since $d(X,g_{A/B})$ does not depend on $m$, the new kink angle is independent of $m$.

With this in mind, we can categorize different possible phases by the inclusion of an $X$-pocket or not. We consider four such possible saddles. In general it is a non-trivial problem to prove that these are the only saddles possible. We will simply assume this to be the case for now. The resulting physics we get from this assumption passes many tests, and so we strongly suspect these phases dominate in at least some finite window of parameter space.

Working with a connected entanglement wedge, with conformal cross ratio $x<1/2$, we now show possible minimal solutions to the hyperbolic network model and list their normalized free energies.
\begin{figure}
	[h!] \centering 
	\begin{subfigure}
		{0.35
		\textwidth} \subcaption{phase I} 
		\begin{overpic}
			[width=
			\textwidth]{kinkEW} \put(25,50){\large$g_A$} \put(70,50){\large$g_B$} 
		\end{overpic}
	\end{subfigure}
	\begin{subfigure}
		{0.35
		\textwidth} \subcaption{phase II} \label{fig:TN_phases} 
		\begin{overpic}
			[width=
			\textwidth]{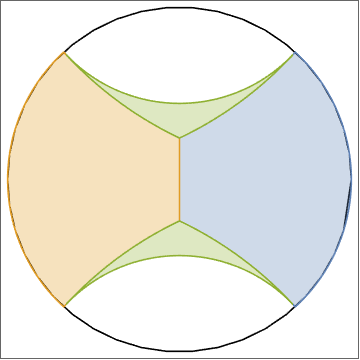} \put(25,50){\large$g_A$} \put(70,50){\large$g_B$} \put(46,65){\large$X$} 
		\end{overpic}
	\end{subfigure}
	\begin{subfigure}
		{0.35
		\textwidth} \subcaption{phase III} 
		\begin{overpic}
			[width=
			\textwidth]{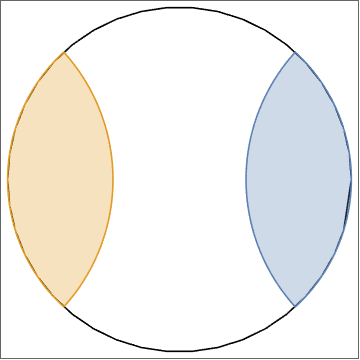} \put(15,50){\large$g_A$} \put(80,50){\large$g_B$} 
		\end{overpic}
	\end{subfigure}
	\begin{subfigure}
		{0.35
		\textwidth} \subcaption{phase IV} 
		\begin{overpic}
			[width=
			\textwidth]{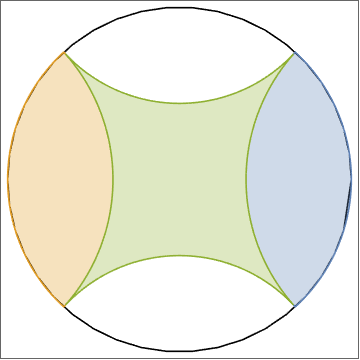} \put(15,50){\large$g_A$} \put(80,50){\large$g_B$} \put(46,50){\large$X$} 
		\end{overpic}
	\end{subfigure}
	\caption{Different phases that can dominate in the reflected entropy bulk configuration . \label{allphases}} 
\end{figure}
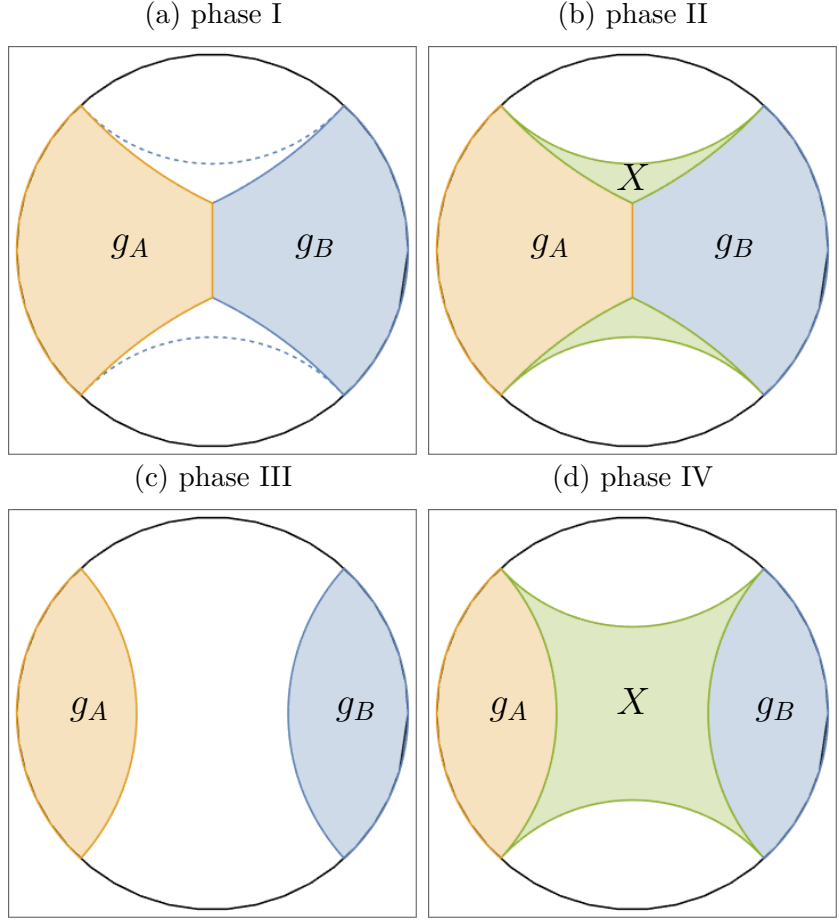
More specifically the $(m,n)$-R\'enyi Reflected entropies are $\ln \chi \times f$ for some minimal $f$. The candidate $f$'s arising from the phases in \figref{allphases} are:
 \begin{itemize}
 	\item phase I 
 	\begin{equation}
 		f_I = 4 (n-1) \ln \frac{ 1 + \sqrt{1 -x} }{ \sqrt{x}} + 4 n(m-1) \left(  H(p_+,p_-) - \ln 2 \right)
 		\label{phaseI}
 	\end{equation}
 	where $p_\pm =1/2 \left( 1 \pm \frac{(n-1)}{n(m-1)}\right)$ and $H(p_+,p_-) = - p_+ \ln p_+ - p_- \ln p_-$ is the Shannon entropy.
 	\item phase II 
 	\begin{equation}
 		f_{II} = 4 (n-1) \ln \frac{ 1 + \sqrt{1 -x} }{ \sqrt{x}} + 4 n \left( H\left(1-\frac{1}{2n},\frac{1}{2n}\right) - \ln 2 \right)
 		\label{phaseII}
 	\end{equation}
 	\item phase III 
 	\footnote{For $x>1/2$ we are always in the disconnected phase corresponding to phase III. The free energy for $x>1/2$ differs by an amount $2n(m-1)\ln((1-x)/x)$ compared to that shown here, due to the phase transition in the normalization of the reflected R\'enyi entropies. }  
 	\begin{equation}
 		f_{III} = 2n(m-1)\ln \frac{1-x}{x} 
 	\end{equation}
 	\item phase IV 
 	\begin{equation}
 	\label{phaseIV}
 		f_{IV} = 2n \ln \frac{1-x}{x} 
 	\end{equation}
\end{itemize}

Note that $x$ is the conformal cross-ratio for end points of the intervals on the boundary, such that $0\leq x\leq 1$ and $x \rightarrow 0$ gives the phase with a connected entanglement wedge. We have assumed that we can approximate the network geodesics by a continuum geometry where there is a conformal symmetry, this is a crude approximation that is sufficient for our purposes. 

We now analytically continue in $(n,m)$. Our first approach is naive and will fail for $1<m<2$. We simply analytically continuing the expressions in Eq.~(\ref{phaseI}-\ref{phaseIV}), and then compare all four free energies at each point in parameter space. 
For example, we plot the resulting phase diagram in the $(m,n)$ plane for several different values of $x$ in Figure \ref{fig:RTNphase}. 
\begin{figure}
	[h!] \centering 
	\begin{subfigure}
		{0.4
		\textwidth} \subcaption{$x=0.04$} 
		\begin{overpic}
			[width=
			\textwidth]{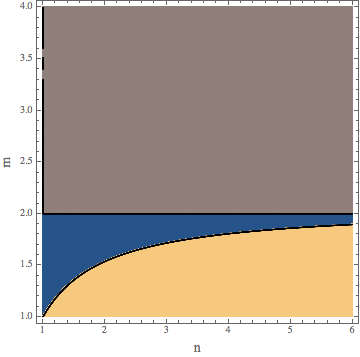} \put(17,27){I} \put(50,55){II} \put(65,20){III} 
		\end{overpic}
	\end{subfigure}
	\begin{subfigure}
		{0.4
		\textwidth} \subcaption{$x=0.11$} 
		\begin{overpic}
			[width=
			\textwidth]{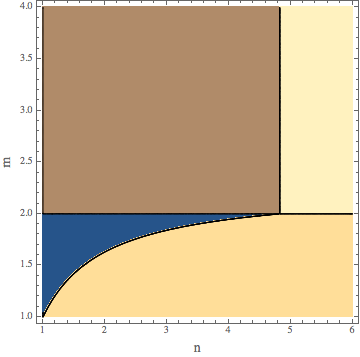} \put(17,27){I} \put(40,55){II} \put(65,20){III} \put(85,55){IV} 
		\end{overpic}
	\end{subfigure}
	\begin{subfigure}
		{0.4
		\textwidth} \subcaption{$x=0.33$} 
		\begin{overpic}
			[width=
			\textwidth]{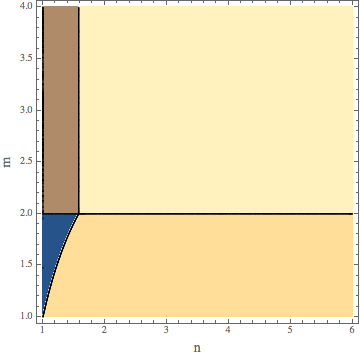} \put(13,25){I} \put(13,55){II} \put(50,20){III} \put(60,60){IV} 
		\end{overpic}
	\end{subfigure}
	\begin{subfigure}
		{0.4
		\textwidth} \subcaption{$x=0.56$} 
		\begin{overpic}
			[width=
			\textwidth]{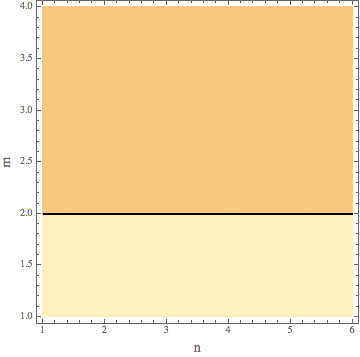} \put(50,60){III} \put(50,20){IV} 
		\end{overpic}
	\end{subfigure}
	\caption{The phase diagrams for the four phases introduced above. The horizontal axis is $n$ and the vertical is $m$. Note the flip in the two phases in the last diagram.} \label{fig:RTNphase} \end{figure}
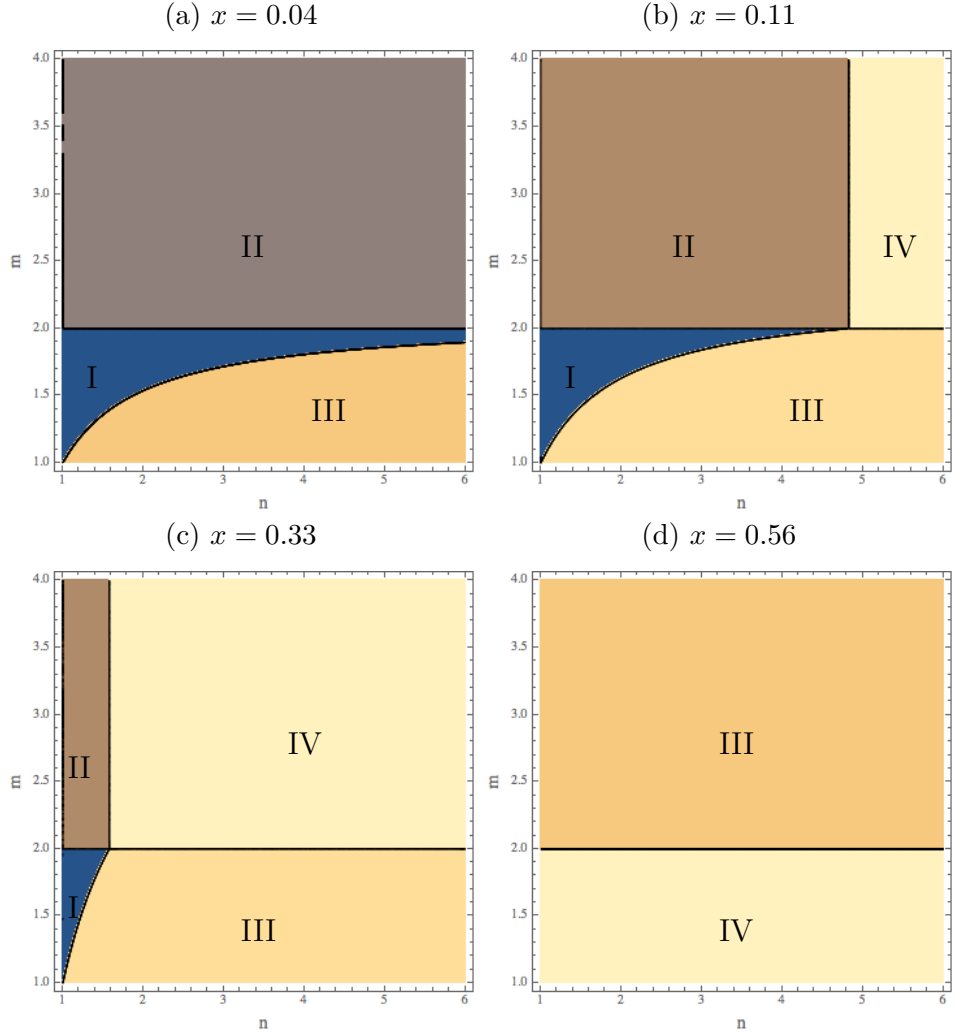
As we can see there is always a transition line at $m=2$, and above the transition line the phases are independent of $m$. The reason for this transition is that the tension of the $X \leftrightarrow e$ domain wall $d(X,e) = n(m-2)$ becomes negative. Indeed this was exactly the same as with the single tensor case. So rather than follow the above naive approach we apply our prescription from the single tensor case (which has already passed many checks) and this implies that the phase diagram in \figref{fig:RTNphase} is incorrect for $m<2$. Instead we should analytically continue the phase transitions above $m\geq 2$ all the way to $m=1$, without re-minimizing over the different phases. The resulting phase diagram is then independent of $m$ for all $m\geq 1$.

We now compute the full entanglement spectrum for this model away from the phase transition point. Analytically continuing from $m>2$ gives for all $m$:
\begin{align}
\label{clmix}
	S_R = -\frac{1}{n-1}
	\begin{cases}
		0, \quad &x>1/2\\
		\min\{ f_{II}, f_{IV} \} , \quad &x<1/2 
	\end{cases}
\end{align}
While the expression in \Eqref{phaseII} seems rather complicated there is an interesting way to write $f_{II}$ at large $\chi$:
\begin{equation}
\label{eq:f_II_z}
	S_{R,II}= - \frac{1}{n-1} \ln \sum_{ a <z< 1} \exp \left\{ -\ln\chi \left(- 4(n-1) \ln z + 4 n \ln \frac{1}{2} \left( \frac{a}{z} + \frac{z}{a} \right)\right) \right\} 
\end{equation}
where $a =\sqrt{x}/(1 + \sqrt{1-x})$ determines the un-pinched EW cross-section: $-2\ln a$, and $-2\ln z$ is the area of pinched cross-section, see \figref{fig:EWX_setup}.
\begin{figure}
    \centering
    \includegraphics[width=.6\textwidth]{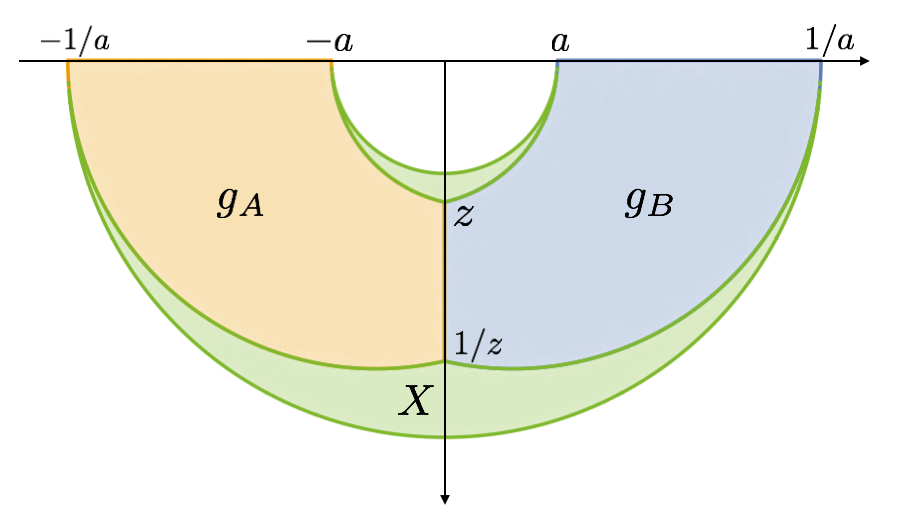}
    \caption{The bulk setup used to calculate \Eqref{eq:f_II_z}. We work within Poincar\'e hyperbolic coordinate with moving downward as going deeper into the bulk. The position of various end points are shown in terms of $a$ and $z$.}
    \label{fig:EWX_setup}
\end{figure}

We are being schematic about the sum. One is tempted to make this an integral, but recall that the network in question is really discrete and so there will indeed be a discrete set of cuts for the cross section. This form agrees with the previous one since we have to evaluate the sum in a saddle point approximation at large $\chi$. In particular \figref{fig:EWX_setup} represents how we computed \Eqref{phaseII} in the first place.

  Following the discussion in Section~\ref{sec:eff} for the single tensor case,
 we now interpret \Eqref{clmix} as arising from an effective description of the canonical purification as a superposition of wavefunctions. Thes wavefunctions then live in approximate superselection sectors when reduced to $AA^\star$.
It is clear the different sectors are associated to tensor networks with different/pinched cross-sectional areas. We think of this pinching as being determined by an area operator that is now allowed to fluctuate. 
 The area is $ A(z) \equiv -2 \ln z$. 
 
 Given this discussion we can read off from \Eqref{clmix} the probabilities of each sector arising:
\begin{equation}
	P(A) = \left( 
	\cosh((A-A_0)/2)
	\right)^{- 4\ln \chi }\, \qquad 0 \leq A \leq A_0
\end{equation}
where $A_0 = A(a) = - 2 \ln a$, the unpinched cross section. 
These probabilities are exponentially small except when $A = -2 \ln a$ where the probability goes to $1$. In general, there will be some perturbative corrections that help maintain the normalization condition. Thus we write our effective model for the canonical purification state using a new set of doubled and glued random tensor networks labelled by $A$: 
\begin{equation}
	\big| \rho_{AB}^{1/2} \big> \sim \sum_A P(A)^{1/2} \left| \Psi(A) \right> 
\end{equation}
where $\left| \Psi(A) \right>$ is defined as a random tensor network state as follows. Consider the pinched entanglement wedge consisting of the vertices with group elements $g_A,g_B$ in  \figref{fig:EWX_setup}. We construct a new tensor network by doubling and gluing along the $X \leftrightarrow g_A,g_B$ domain wall.
 This domain wall is slightly pinched relative to the original entanglement wedge.
 As before we can pick the random tensors on the $AB$ entanglement wedge to match those of the original tensor network, while we pick the tensors on the $(AB)^\star$ to be independent and random.
 
 Note that $2A(z)$ will represent a true minimal cut, homologous to $AA^\star$, through this new doubled network. In particular there is an associated $AA^\star$ entanglement wedge consisting of the region outside of this minimal cut/cross-section.  Since the entanglement wedges for these random tensor networks associated to the boundary region $AA^\star$ are very different for different $A(z)$ we expect the density matrices reduced to $AA^\star$ to be approximately orthogonal. This is also true for $BB^\star$.

It then follows that ${\rm Tr}_{AA^\star} \left| \Psi(A) \right> \left< \Psi(A') \right| \propto \delta_{A, A'} $ and $ {\rm Tr}_{BB^\star} \left| \Psi(A) \right> \left< \Psi(A') \right| \propto \delta_{A, A'} $ up to small non-perturbative corrections, and our results now parallel the results discussed in Section~\ref{sec:eff}. In particular an area operator on the physical Hilbert space emerges, by using the approximately orthogonal supports of $\rho_{AA^\star}(A(z))$ and $\rho_{BB^\star}(A(z))$. This area operator then determines the R\'enyi reflected entropy:
\begin{equation}
e^{-(n-1)S_{R,II}} = \sum_{A} P(A)^n e^{ - (n-1) A}  
\end{equation}
which agrees with \Eqref{eq:f_II_z}.
Evaluating the sum in the saddle point leads back to \Eqref{phaseII}. The dominant area $A$ shifts as a function of $n$ and in particular the $n \rightarrow 1$ limit gives back the entanglement wedge cross section since $P(A) \rightarrow 1$ in this limit. 

For large enough $n$ the dominant phase actually becomes the disconnected phase IV which is the same $n$ dependent phase transition that occured for the single random tensor model. The final effective description of the canonical purification, that capures all these effects, is shown pictorially here:
\begin{equation}
	\includegraphics[width=0.8
	\textwidth]{fig/CP_effective_2.pdf} 
\end{equation}

\section{Discussion}
\label{sec:disc}

While the main goal of this paper was simple, to compute reflected entropy in random tensor networks, we ended up discovering some surprises along the way. We summarize these here and also comment on some possible extensions of these results.

\subsection{Effective description of canonical purification}

Given a state $\rho_{AB}$, we have suggested a recipe for computing the reflected entanglement spectrum in random tensor networks. 
First construct the canonical purification $\ket{\sqrt{\rho_{AB}}}_{A A^\star B B^\star}$ using a generalized gluing construction - allowing for a superposition of networks for different possible entanglement wedges. 
Then compute the spectrum of the density matrix on $A A^\star$ using the fact that the individual density matrices on $AA^\star$ are approximately orthogonal. 

Let us comment in more detail how this relates to the gravitational gluing construction used in Ref.~\cite{Dutta:2019gen}. 
There, the canonical purification $\ket{\sqrt{\rho_{AB}}}_{A A^\star B B^\star}$  was described as dual to a particular bulk geometry, the one formed by gluing two copies of the $AB$ entanglement wedge together.
If we naively follow this procedure for a random tensor network we arrive at the following picture. Firstly the entanglement spectrum of $\rho_{AB}$ is flat in this case, so up to normalization we can replace $\rho_{AB}^{1/2}$ by $\rho_{AB}$. Then seemingly the gluing procedure automatically follows, since we know there is an isometry from the bulk legs at the $AB$ RT surface to the $AB$ boundary legs and $\rho_{AB}$ already contains two copies of the tensor network for the bra and the ket \cite{Marolf:2019zoo}. 

To see how this works, consider again the single tripartite tensor example.
The canonical purification $\ket{\sqrt{\rho_{AB}}}_{A A^\star B B^\star}$ following this procedure corresponds to a tensor network with two tensors $T_{ABC}$ and $T^*_{ABC}$ contracted along leg $C$.
This however is \emph{not exactly} the analog of gluing two copies of the entanglement wedge together. The difference is that these two tensors are now correlated. Normal tensor network analogs of any given geometry discretize that geometry with a collection of tensors that are all chosen independently at random. Hence the random tensor network analogous to the glued entanglement wedges from gravity \cite{Dutta:2019gen} would have completely uncorrelated tensors in the two copies!

How big is this difference?
It is certainly somewhat important: while the network of two uncorrelated tensors would have a completely flat entanglement spectrum, the canonical purification does not. As computed in Section \ref{sec:single_tensor}, the spectrum of $AA^\star$ involves two peaks. These two peaks can trade dominance as a function of R\'enyi parameter $n$. For example, when the entanglement wedge is in the connected phase, the single eigenvalue peak is subdominant, and the spectrum of $AA^\star$ is approximately that of two uncorrelated tensors-- yet for large enough $n$, the single eigenvalue peak begins to dominate. 
This phase transition as a function of $n$ is completely absent in two uncorrelated tensors.
Hence the canonical purification in this tensor network example is not entirely described by the analog of gluing of two entanglement wedges. 

That said, there is a simple fix, an updated effective description of the canonical purification that does capture this more complicated spectrum.
As we described in Sections \ref{sec:intro} and \ref{sec:single_tensor}, the spectrum of $AA^\star$ appears analogous to that obtained by summing, with appropriate weights, two tensor networks as in \Eqref{eq:CP_effective_1TN}, one with two uncorrelated tensors and one with no tensors at all.
And note, for the purposes of describing $\rho_{AA^\star}$, this effective description is quite consistent with the original, naive gluing of the entanglement wedge, as long as we are far from any $n$-dependent phase transition, i.e. as long as one peak dominates and the other is an exponentially small correction to the state.
Therefore this effective description is nice for at least two reasons: (1) it gives a good description of the reduced density matrix of $A A^\star$, and (2) it clarifies the sense in which we should trust the doubled-and-glued description of the canonical purification.

How good is this effective description for more general purposes? Can we use it to compute things besides the R\'enyi entropies of $AA^\star$, such as the spectrum of $AB$?
At least far from phase transitions, this seems roughly correct:  as long as $|AB| \ll |C|$ or $|AB| \gg |C|$. In the former case, $\rho_{AB}$ is approximately maximally mixed, which is indeed the density matrix in the effective description of that regime ($AB$ maximally entangled with $A^\star B^\star$). In the latter case, $\rho_{AB}$ is approximately maximally mixed on a random dimension $|C|$ subspace, which is indeed the density matrix of the effective description in that regime (two uncorrelated tensors contracted across a dimension $|C|$ leg).
Hence this effective description seems valuable, capturing the far-from-transition physics as well as at least some of the physics near transitions.

We note one interesting subtlety: the two networks in the effective description \Eqref{eq:CP_effective_1TN} are not generally orthogonal. So while it is approximately correct for some purposes to view them as defining separate superselection sectors, their overlap is not always ignorable. 

How does this effective description generalize beyond the single tensor example, to hyperbolic networks?
As we described in Sections \ref{sec:intro} and \ref{sec:general_TN}, it seems the spectrum of $AA^\star$ is well-modeled by as a sum of many tensor networks, each formed by doubling and gluing the tensor network along a different, possibly kinked candidate entanglement wedge, with all tensors chosen independently at random. See \Eqref{eq:CP_effective_hyperbolic}.

What do these lessons say about reflected entropy in a gravitional theory?
Because something is clearly missed in the random tensor network by the naive doubling and gluing of the entanglement wedges, we can expect that the same is true in gravity. Albeit we still expect the Reflected entropy/EW cross-section duality to hold, just not away from $n=1$. 

Likely there is some effective description that improves upon this naive doubling, analagous to the tensor network case sketched in \Eqref{eq:CP_effective_hyperbolic}.
That said, it's not entirely clear how to interpret such an effective description in gravity. For example it is not obvious that the `kinked' geometries in \Eqref{eq:CP_effective_hyperbolic} correspond to any saddles in gravity.
We leave a precise investigation of this question to future work.

\subsection{Non-flat spectrum and building geometry from RTNs}

This effective description offers an interesting possibility.
Perhaps the canonical purification is a useful tool for constructing, out of tensor networks, a geometry with gravity-like area fluctuations.

To explain this, let us recall some background.
While impressive in many ways, the precise relationship between random tensor networks and AdS/CFT is not yet fully understood.
So far, the best understanding is that tensor networks resemble so-called fixed-area states, which have approximately flat R\'enyi spectra just like random tensor networks \cite{Akers:2018fow, Dong:2018seb}.
This matching is quite good; both fixed-area states and random tensor networks have exactly the same non-perturbative corrections to the R\'enyi entropies \cite{Penington:2019kki}.

However, it would be nice to have a random tensor network model of something more realistic than a fixed-area state.
Our results suggest an interesting method for obtaining such a model.
The procedure is as follows.
First, start with some conventional random tensor network, say $\ket{\psi}_{ABC}$.
Second, find the canonical purification $\ket{\sqrt{\rho_{AB}}}$.
As we've shown, $\rho_{AA^*} = \tr_{BB^*}(\ket{\sqrt{\rho_{AB}}}\bra{\sqrt{\rho_{AB}}})$ has a non-flat spectrum.
That said, this is only a partial success. Not all factors of the canonical purification have a non-flat spectrum. In particular, $\rho_{AB}$ is the same as in the original tensor network. To fix this we can iterate, now finding the canonical purification of the canonical purification, this time canonically purifying, say, $\rho_{AA^*}$.
Then all factors have non-flat spectra.
If we like, we can continue to iterate, building up increasingly sophisticated superpositions over tensor networks and associated spectrum. In this way we might even build up complicated tensor networks describing higher dimensional space times just starting from a single random tensor.
This seems a bit like the Eguchi-Kawai mechanism which grows extra dimensions out of large-$N$ matrices. See for example Ref.~\cite{Shaghoulian:2016xbx}.

Having said this, this is not the only method for building up a tensor network with a non-flat spectrum.
Another possibility is to simply add degrees of freedom to each of the legs connecting the tensors, as in \cite{Donnelly:2016qqt}.
This raises the question, is there any reason one might prefer this iterative canonical purification method for building a tensor network with a non-trivial entanglement spectrum?

Here's one possible reason.
In conventional random tensor networks, even those with degrees of freedom on each of the links as in \cite{Donnelly:2016qqt}, the area operators associated to two crossing cuts commute. 
That is, if you consider one cut through the tensor network, and then a second cut that crosses it (but nowhere overlaps it), the `areas' associated to those cuts can be simultaneously fixed. 
This cannot happen for overlapping cuts in AdS/CFT \cite{Bao:2018pvs}.
There, crossing areas do not commute because the area is conjugate to the boost angle across the surface, and fixing one area makes the geometry of the Cauchy surface highly uncertain.
Hence simply adding link degrees of freedom does not capture this subtle behavior seen in gravity.

Very speculatively, perhaps this \emph{is} captured by the iterative canonical purification geometry.
If so, this would be an interesting reason to take these seriously as toy models of holography.
We leave such an investigation to future work. 

\acknowledgments

We would like to thank Xi Dong, Jonah Kudler-Flam, Don Marolf and Henry Maxfield for useful discussions. PR would especially like to thank Xi Dong, Don Marolf and Amir Tajdini for their whole-hearted support and patience through the completion of this project. PR is supported in part by a grant from the Simons Foundation, and by funds from UCSB. 
CA is supported by the Simons foundation as a member of the It from Qubit collaboration. This material is based upon work supported by the Air Force Office of Scientific Research under award number FA9550-19-1-0360. This work benefited from the Gravitational Holography program at the KITP and we would like to thank the KITP, supported in part by the National Science Foundation under Grant No. NSF PHY-1748958.

\appendix 

\section{Symmetric Group and Set Partitions}
\label{app:perm}

In this appendix, we will summarize certain aspects of the symmetric group $S_N$, the group of all permutations on $N$ elements. 
These are closely related to set partitions $P_N$, the set of all partitions of $N$ elements.
We start with a quick review of some well-known theorems of non-crossing permutations $NC_N$.
These theorems can be given a topological meaning via the relation to set partitions, 
which allows us to generalize the theory of non-crossing permutations to multiple disjoint boundaries.
They will play an important role in the proof of the form of the phase diagram for reflected entropy. 

\subsection{Cayley distance}

The Cayley distance $d(g,h)$ is a metric on $S_N$ defined by the minimal number of transpositions, i.e., swaps of two elements, required to go from $g$ to $h$.
Any permutation can be decomposed into disjoint cycles, where a cycle of $n$ elements is represented by the notation $(i_1\,i_2\dots i_N)$.
We first note that conjugation of an element $g$ by an element $h$ results in an element $h g h^{-1}$ with the same structure of cycles as $g$ but with a relabelling of the entries in each cycle dictated by $h$ as
\begin{align}
    (i_1\,i_2\dots i_n) &\mapsto (h(i_1)\,h(i_2)\dots h(i_N)).
\end{align}
Using this it is easy to show that the Cayley metric is both left and right invariant, i.e.,
\begin{align}
    d(g,h) &= d(g x, h x) = d(x g, x h).
\end{align}
The right invariance follows from the definition, while left invariance uses the fact that conjugation of a product of transpositions by $x$ is still a product of transpositions.

In particular, this means we can reduce all calculations of distance between two elements to the distance of an element from the identity element $e$, i.e.,
\begin{align}
    d(g,h) &= d(e,g^{-1} h) = d(e, h g^{-1}).
\end{align}
A special example is a cycle of $k$ elements for which the distance from $e$ is simply $k-1$, as can be seen by constructing an optimal decomposition
\begin{align}
    (i_1\,i_2\dots i_N) &= (i_n\,i_1)\dots (i_3\,i_1)(i_2\,i_1).
\end{align}
Using the above fact and the decomposition of an arbitrary permutation $g$ into $k$ disjoint cycles of size $n_k$, we have
\begin{align}
    d(e,g) &= \sum_{k}\, (n_k -1)= N-\#(g),
\end{align}
where we use the notation $\#(g)=k$ to denote the cycle counting function.\footnote{Note that elements that map to themselves are counted as cycles of size 1.}
This makes it manifest that $d(e,g)=d(e,h g h^{-1})$, i.e., distance from the identity is invariant under conjugation.

\subsection{Non-crossing permutations}
\label{app:NC}
The Cayley metric like any other metric satisfies a triangle inequality
\begin{align}
    d(g,h)+d(h,r) &\geq d(g,r),
\end{align}
where equality implies $h$ lies on a geodesic between $g$ and $r$, and the elements lying on the geodesic are collectively denoted $\Gamma(g,r)$.
In the case where one of the elements is the identity, the set $\Gamma(g,e)$ can be easily classified using the theory of non-crossing permutations as we review below.

\begin{definition}[Non-crossing permutations]
Let $g\in S_N$. Consider a disk with $N$ cyclically ordered (say clockwise) marked points on its boundary and connect the points with directed lines according to the cycle decomposition of $g$.
Then $g$ is non-crossing if and only if every directed cycle are also clockwise oriented and can be drawn in the interior of the disk without ever crossing each other.
The set of non-crossing permutations is denoted $NC_N$.
\end{definition}

\begin{theorem}\label{thm:biane}(Biane \cite{Biane97})
An element $g$ lies on the geodesic between $e$ and maximal cyclic permutation $\tau \equiv (12\cdots N)$, i.e. , satisfies
\begin{align}
\label{eq:biane_cond}
    d(e,g)+d(g,\tau)&=N-1,
\end{align}
or equivalently,
\begin{align}
    \#(g)+\#(\tau g^{-1}) &=N+1,
\end{align}
if and only if it is a non-crossing permutation.
\end{theorem}

\noindent We need the following lemma.
\begin{lemma}\label{swap}
Suppose $g$ is a permutation such that $\#(g)=k$ and $\sigma$ is a transposition. Then $\#(g \sigma)=k+1$ if and only if the elements exchanged by $\sigma$ are in the same cycle of $g$, else $\#(g \sigma)=k-1$.
\end{lemma}

\begin{proof}
Suppose the elements transposed by $\sigma$, labelled $i_1$ and $i_2$, are part of the same cycle in $g$, then $g \sigma$ splits into two cycles of the form
\begin{align}
    (i_1,\,g(i_2),\dots,g^{-1}(i_1))(i_2,\,g(i_1),\dots,g^{-1}(i_2)).
\end{align}
On the other hand if $i_1$ and $i_2$ were not part of the same cycle in $g$, then $g \sigma$ couples them into a single cycle of the form
\begin{align}
    (i_1,\,g(i_2),\dots,\,i_2, \,g(i_1),\dots).
\end{align}
\end{proof}

\begin{figure}
    \centering
    \includegraphics[scale=0.6]{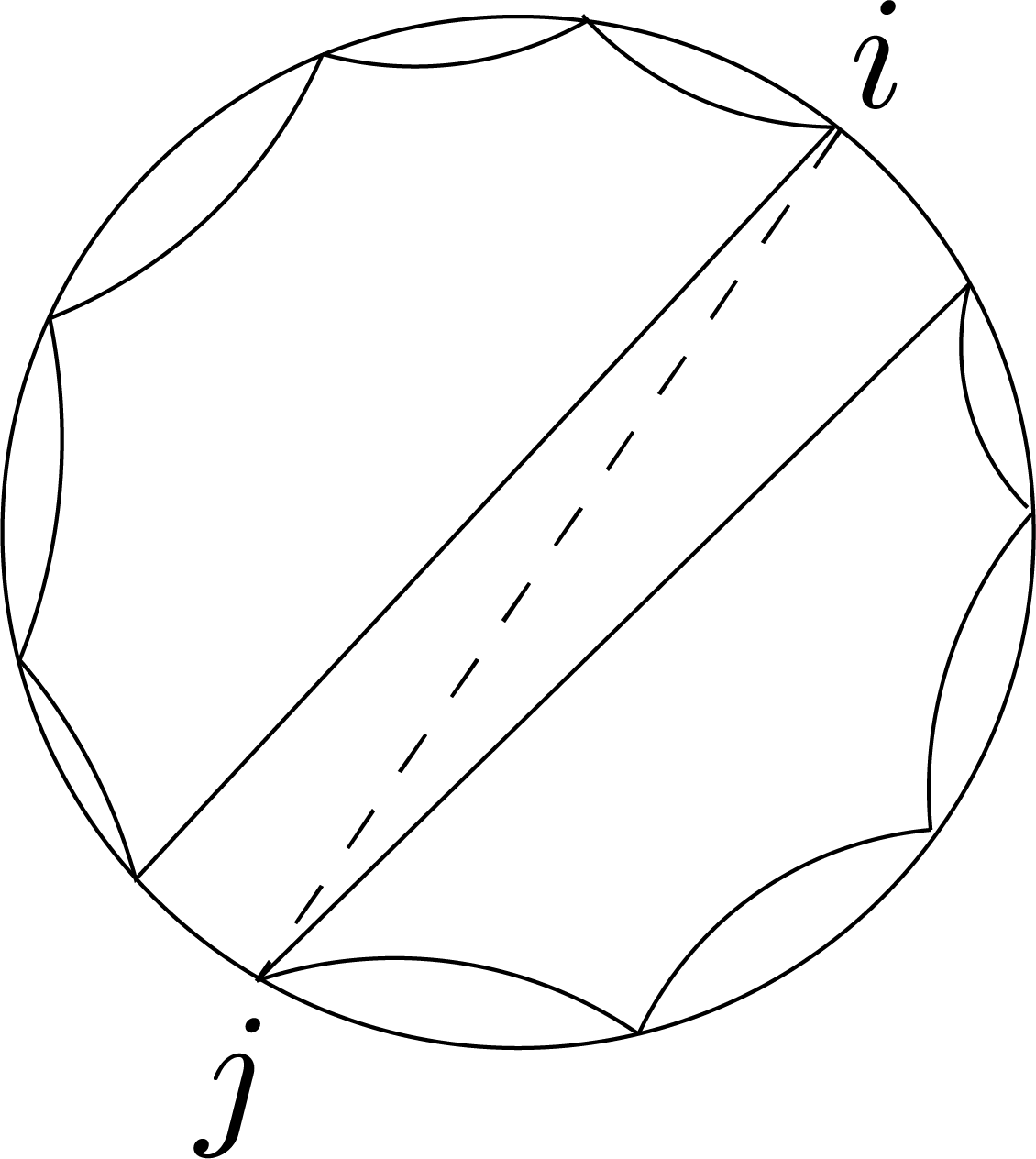}
    \caption{The transposition $(i,j)$ breaks the cyclic permutation $\tau$ into two cycles divided across the chord joining points $i$ and $j$. In order to satisfy the geodesic condition, all further transpositions should act on elements within the given cycles. Thus, we end up with a non-crossing permutation.}
    \label{fig:NCn}
\end{figure}

\begin{proof} (of Theorem \ref{thm:biane}).
Consider an element $g$ such that $d(g,e)=k$ or $\#(g) = N-k$.
Using Lemma~\ref{swap}, we see that $\tau g^{-1}$ has at most $k+1$ cycles and this precisely happens when the transpositions generating $g^{-1}$ break the cycles of $\tau$ at each step.
The action of a transposition $(i\,,j)$ is to break $\tau$ into cycles $(1,2,\dots,i\,,j+1,\dots,N)$ and $(i+1,\dots,j)$.
This can be represented on a disk as two cycles divided by the chord joining elements $i$ and $j$ as seen in \figref{fig:NCn}.
This process can be repeated $k$ times, while ensuring that transpositions always act on elements within the same cycle.
The figure makes it clear that the element $\tau g^{-1}$ obtained this way is a non-crossing permutation.

To prove the other direction, it is useful to introduce a topological representation as shown in \figref{fig:double}, where the green blocks represent the group action of $g$ while the orange blocks represents the action of $\tau g^{-1}$.
Theses blocks can never cross as implied from the non-crossing condition.
Now reinterpreting the figure as a graph we identify $2N$ vertices, $3N$ edges and the faces represent a cycle of either $g$ or $\tau g^{-1}$.
Thus, using the Euler formula we have
\begin{align}
    &V-E+F= \#(g)+\#(\tau g^{-1}) -N =2-2 G - B\\
    &\implies \#(g)+\#(\tau g^{-1})= N+1 -2G,
\end{align}
where $G$ is the genus of the Riemann surface and we have used $B=1$ for the number of boundaries.
Since the graph is planar, $G=0$, and we find that the geodesic condition is satisfied for non-crossing permutations.
\end{proof}

\begin{figure}
    \centering
    \includegraphics[width=.5\textwidth]{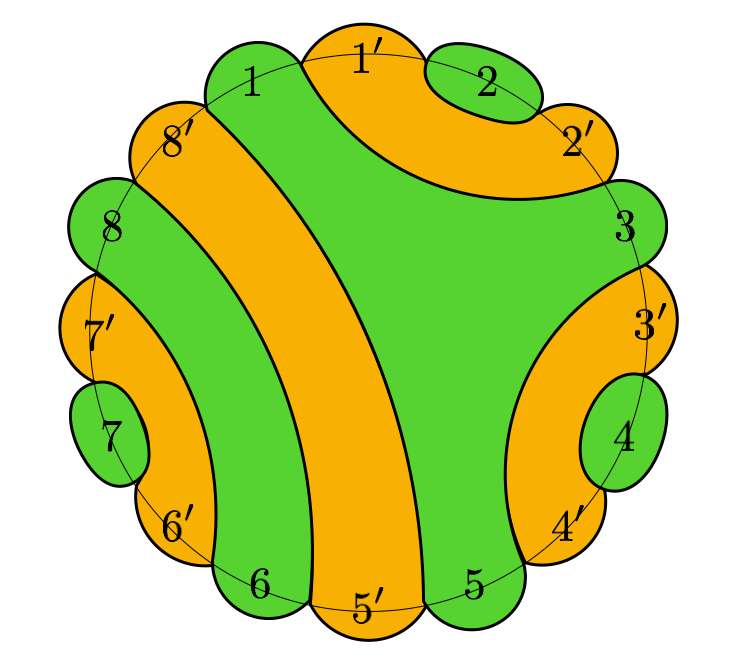}
    \caption{The group action of a non-crossing permutation $g$ is represented by its orbits (green) acting on the elements $i\in{1,2,\dots,N}$, where $N=8$. The action of $\tau g^{-1}$ is represented by its orbits (orange) acting on the same elements denoted $i'$ for clarity.}
    \label{fig:double}
\end{figure}

We note a nice corollary of theorem \ref{thm:biane}.
\begin{corollary}\label{thm:biane2}
Suppose $g$ is a permutation with a decomposition into disjoint cycles ${B_1,B_2,\dots, B_k}$. Then an element $h$ lies on the geodesic between $e$ and $g$ if and only if  $\forall j$ $B_j$ is a union of disjoint cycles of $h$, and $h$ restricts to a non-crossing permutation on each $B_j$. In this case, non-crossing is defined with respect to the orientation of cycles of $g$.
\end{corollary}
\begin{proof}
This follows from a simple application of Theorem~\ref{thm:biane} for each $B_j$ separately.
\end{proof}

\noindent \emph{Remark.}
The topological representation (\figref{fig:double}) relates a permutation to a graph defined on $2N$ vertices. The converse is also true and there is a bijection between such graphs and permutation elements $g\in S_N$.
However the corresponding graph is only planar iff $g\in NC_N$.
In general one can think of embedding this graph onto some Riemann surface and make this statement much more refined by relating the genus
directly to the failure to comply the geodesic condition, i.e.
For $g,h\in S_N$,
\begin{equation}
\label{prop:topology}
    d(e,g) + d(g,h) = d(e,h) + 2(C - B) + 2G 
\end{equation}
where $B$ is the number of boundaries in the graph which is set by $\#(h)$ , $C$ is the number of connected components of the graph and $G$ is the genus of the surface.
\footnote{The meaning of $C$ and $G$ is ambiguous at this point. In particular the genus of the embedding surface is not the same as the common definition of the genus of the graph itself since our surface has nontrivial boundaries. We will clarify what we mean by these numbers when we set out to prove this proposition.}
For example when $h=\tau$, we have $B=C=1$ and we recover the geodesic condition \Eqref{eq:biane_cond}.
This allows us to study the group geodesic in a more general setting and have a better handle on the resolvent calculation.
We will give a proof for this proposition in appendix \ref{app:perm_topology}.

In the calculation of various entropic quantities we will encounter group summations over non-crossing permutations, weighted by the individual cycle counts.
The outcome of these summations can be expressed in terms of \emph{q-Catalan numbers}. 
\begin{definition}[q-Catalan numbers]
Given any positive integer $n\in \mathbb{N}$ and $q \in \mathbb{C}$, the \emph{q-Catalan number} $C_n(q)$ is defined by the following sum
\begin{equation}
    C_n(q) = \sum_{g\in NC_n} q^{\#(g)} = \sum^n_{k=0} q^k N(n,k),
\end{equation}
where $N(n,k)=\begin{pmatrix}
      n \\ p
    \end{pmatrix}
    \begin{pmatrix}
      n \\ p-1
    \end{pmatrix}$
    are called the \emph{Narayama Numbers}.
    \end{definition}
    $N(n,k)$ counts the number of distinct non-crossing permutations with exactly $k$ cycles. Thus $C_n(q)$ is also the generating function for Narayama Numbers.
    The q-Catalan numbers can be expressed in terms of the Hypergeometric functions
  \begin{equation}
    C_n(q) = q \;_2F_1 (1-n,-n;2;q)
  \end{equation}
  For positive integer $n>0$ we have the following relationship:
  \begin{equation}
    C_n(1/q) = q^{-n-1}C_n(q)
  \end{equation}
  This is not true for non-integer $n$.
  Because of the ambiguity in Hypergeometric function due to the branch cut at $q=1$ we give the following analytic continuation in $n$ at fixed $q$:
  \begin{definition}[analytic continuation of q-Catalan numbers]
  For $n>0$ and $q>0$,
  \begin{equation}
    C_n(q) \equiv
    \begin{cases}
      q \;_2F_1 (1-n,-n;2;q) , \quad &q\le 1 \\
      q^n \;_2F_1 (1-n,-n;2;1/q), \quad &q>1
    \end{cases}
  \end{equation}
  \end{definition}

\subsection{Set partitions}
\label{app:set_partitions}

We now establish the relation between symmetric group $S_N$ and set partitions $P_N$.
We then show that it is naturally equipped with a lattice structure which defines a way to compare elements in $P_N$.

A partition $p \in P_N$ is a disjoint set of subsets, or blocks, whose disjoint union is $\mathbb{Z}_N$. Given a $g \in S_N$ we can use the cycles to produce a partition $P : S_N \rightarrow P_N$. 
This map is surjective but not injective. 
For example, if $g = (132)(45)\in S_5$ then $P(g) = \{\{1,2,3\},\{4,5\}\}$.
We can similarly define the counting function $\#(p)$ as the number of blocks in the partition. So $\#(g) = \#(P(g))$. We denote the finest partition by $e = \{ \{1\},\{2\}, \ldots ,\{N\}\}$
and the coarsest by $\{\mathbb{Z}_N\}$. 

There is a natural partial order on such partitions given by a refinement of the partitions or sub-partitions. That is $p_1 \leq p_2$ if for every block $ c \in p_2$ there is a subset of blocks in $p_1$ that forms a partition of $c$.
It turns out $P_N$ satisfies nicer properties that makes it a lattice.
We review the definition and basic properties of lattices below.
\begin{definition}[Lattice]
A lattice is a partially ordered set $L$ in which each two elements $a,b\in L$ always have a meet and join, where:
\end{definition}
\begin{definition}[Meet and Join]
Let $P$ be a partially ordered set and $a,b \in P$. The \emph{join} of $a$ and $b$, denoted $a \vee b$, is the least upper bound of $a$ and $b$, i.e. $a\vee b \le x$ for every $x$ that simultaneously satisfies $x>a$ and $x>b$. Conversely, the \emph{meet} of $a$ and $b$, denoted $a\wedge b$, is the greatest lower bound of $a$ and $b$, i.e. $a\wedge b \ge y$ for all $y$ such that $y<a$ and $y<b$.
\end{definition}
The join satisfies certain properties: 
\begin{align}
    a \vee b &= b \vee a \quad &\text{(commutativity)} \\
    a \vee (b \vee c) &= (a \vee b) \vee c \quad  &\text{(associativity)} \\
    a \vee a &= a, \quad &\text{(idempotent)}
\end{align}
 and similarly for the meet.  If $ a \leq b$ then $a \vee b = b$ and $a \wedge b = a$. They commute with the order, e.g. $a \leq b \implies a \vee c \leq b \vee c$. Note however in general they do not satisfy the distributive property that is familiar from set intersections and unions, i.e. $a\vee (b\wedge c) \neq (a\vee b) \wedge (a\vee c)$.

We now show that the set partition $P_N$ is indeed a lattice by explicitly constructing $p_1 \vee p_2$ and $p_1\wedge p_2$.
The latter $p_1 \wedge p_2$ is simply the set of pairwise non-empty intersections taken over all blocks in $p_1, p_2$. While the former can be found recursively via: $p_1 \vee p_2 = p_1 \vee (d_1) \vee (d_2) \ldots $ where $p_2 = (d_1) (d_2) \ldots$ is 
the block decomposition. Then for a single block: $q \vee (d_1)$ is simply $(\cup_{d_1 \cap c_i \neq \emptyset} c_i) q'$ with $(\cup_{d_1 \cap c_i \neq \emptyset} c_i) \equiv$ the single block formed by all $q$'s that intersect $d_1$
and $q' = \prod_{d_1 \cap c_i = \emptyset}  (c_i)$ are the remaining blocks that do not. 

 It turns out this lattice is also graded and (upper) semimodular.
 \begin{definition}[Grading of a lattice]
 A lattice $L$ is \emph{graded} if there exists a map
 $\rho: L\to \mathbb{N}$ such that for $a,b \in L$:
 \begin{itemize}
    \item It is compatible with the ordering of the lattice:
      we have $\rho(a) > \rho(b)$ when $a>b$.
    \item It is compatible with the covering condition:
      $\rho(a) = \rho(b) +1 $ iff $a$ covers $b$ (denoted $a :> b$), i.e. $a > b$ and there is no other $c\in L$ such that $a > c > b$.
 \end{itemize}
 The map $\rho$ is called a $grading$ of the lattice $L$.
 \end{definition}
 
For the lattice at hand the grading $\rho(q) = N - \#(q)$ satisfies these properties. 
We can check that $\rho$ is graded by the fact that when $p_1 > p_2$ and there is no other $p_3$ such that $p_1 > p_3 > p_2$, then $p_1$ must to be formed from $p_2$ by merging two blocks implying that $\rho(p_1) = \rho(p_2) + 1$. 

A graded lattice $L$ is semimodular iff the grading $\rho$ satisfies the following property:
\footnote{For a general lattice $L$ the semimodular condition is that $\forall a,b \in L, a \wedge b <: a$ implies $b <: a \vee b$. The equivalence of two definitions can be found on standard textbooks on lattice theory, e.g. Birkoff~\cite{birkhoff1967lattice}.}
\begin{equation}
\label{semmod}
 \rho( a \vee b) + \rho( a \wedge b) \leq \rho(a) + \rho(b) \quad \forall a, b \in L
\end{equation}
We give a proof that set partitions satisfy this inequality.

\begin{lemma}[Semimodularity of $P_N$]
\label{aproof}
For all $p_1,p_2 \in P_N$ then:
\begin{equation}
\#(p_1 \vee p_2) + \#(p_1 \wedge p_2) \geq \#(p_1) + \#(p_2)
\end{equation}
with equality iff the graph, formed by $\#(p_1) + \#(p_2)$ vertices, and connected via $\#(p_1 \wedge p_2)$ edges in the natural way, is a disconnected union of tree graphs.
\end{lemma}
\begin{proof}
Without loss of generality we can consider the case where $ \#( p_1 \vee p_2)  = 1$ since the different blocks in $p_1 \vee p_2$ 
contribute independently to the inequality. So the general case can be written as a sum over this case.  

Consider a bi-partitle graph formed by black vertices for each block in $c_1^i \in p_1$ and white vertices for each block in $c_2^i \in p_2$.
Connect the vertices by edges for each non-trivial intersection $c_1^i \cap c_2^j \neq \emptyset$. There must be exactly $E=\#( p_1\wedge p_2)$ edges by definition.
The graph must also be connected by $p_1 \vee p_2 = \mathbb{Z}_N$. A connected graph has the simple bound:
\begin{equation}
E \geq V -1 
\end{equation}
which is easy to prove by induction on the number of vertices. But $V = \#(p_1) + \#(p_2)$ and we are done. 
Saturation of the inequality implies the graph is a tree (with no cycles), and this follows from the induction proof mentioned above. 
\end{proof}

\subsection{Annular non-crossing and topology}
\label{app:perm_topology}
We now make precise the statement in \eqref{prop:topology} and provide a proof for it. 
In particular we will define what we mean by the embedding surface and various topological quantities associated to it.
The statement of proposition gives rise to an inequality which is saturated when the genus $G=0$.
This marks the proposition as a generalization to the geodesic conditions of non-crossing permutations and we will denote the set of the elements that saturates the inequality as \emph{multi-annular non-crossing permutations}.
\footnote{This terminology comes from \cite{Kim12annular} who first studied the case of two boundaries, where the embedding surface of the graph is annular.}

Fix a special element $g_0\in S_N$ with cycle decomposition $g_0= \prod_i c_i^0$. 
From corollary \ref{thm:biane2} we know that geodesics to the identity $g \in \Gamma(e,g_0)$ are products of non-crossing permutations of length $| c^0_i |$ for each cycle. We will
denote these simply as $NC_{g_0} \equiv NC_{ |c^1_0|, |c^2_0| \ldots}$. 
There is a geometric picture of these elements as curves living on a disjoint union of $\#(g_0)$ disconnected discs.

Beyond non-crossing permutations we now discuss multi-annular non-crossing permutations that allow for connections between the different $g_0$ cycles but continue to have zero genus.
We will need a way to describe how the different cycles in $g_0$ are connected through the action of $g$.
\begin{definition}[Connectedness]
Given two elements $g,g_0 \in S_N$. 
We define the \emph{connectedness} of $g$ over $g_0$, written as $q_{g_0}(g)$, to be the set quotient ($P(g) \vee P(g_0))/P(g_0)$, which is itself a partition $q_{g_0}(g) \in P_{\#(g_0)}$.
\end{definition}
Note that the quotient is well defined since $P(g_0) \le P(g)\vee P(g_0)$.
We will often drop the subscript $g_0$ when it is clear which base permutation we refer to.
The connectedness $q_{g_0}(g)$ measures how the different orbits of $g_0$ are joined by actions of $g$.
Besides connectedness there is another quantity we can define that measures the number of connected components.
\begin{definition}
Given $g,g_0 \in S_N$. We denote $\#(g\vee g_0)$ to be the number of orbits of where  $\mathbb{Z}_N$ is split under the joint action of $g$ and $g_0$. By the ``joint action'' we mean the action on $\mathbb{Z}_N$ by the subgroup generated by $\braket{g,g_0}$.
\end{definition}
Note that $\#(g\vee g_0) = \#(P(g) \vee P(g_0)) = \#(q_{g_0}(g))$.
This number equals $\#(g \vee g_0) = \#(g_0)$ for $g \in \Gamma(g_0, e)$. More generally, we have to discuss the topology of permutation elements. 

\begin{definition}
Given elements $g, g_0 \in S_N$ we define an \emph{admissible surface} $\Sigma$ for $g$ (based over $g_0$) as a disjoint union of oriented two dimensional Riemann surfaces 
with a total of $\#(g_0)$ connected boundaries,  and such that it is possible to decorate this surface as follows:
\begin{enumerate}
\item Each connected boundary has $| (c_i^0) |$ marked points where $i = 1, \ldots  ,\#(g_0)$ labels the boundary.
Each point is ordered and marked using an element from the cycle $(c_i^0) $. The orientation of the ordering is fixed by the orientation of the surface.

\item The permutations in $g$ are represented by oriented curves on $\Sigma$ that pass between two different marked points on the boundaries according to $g$. Each marked point has a curve entering and leaving, and we locally pick these in the same direction as the cycles in $c_0$ (see Figure below.)

\item The curves are all mutually \emph{non-crossing}. 
\end{enumerate}

\begin{equation}
\nonumber
\centering
\includegraphics[scale=.8]{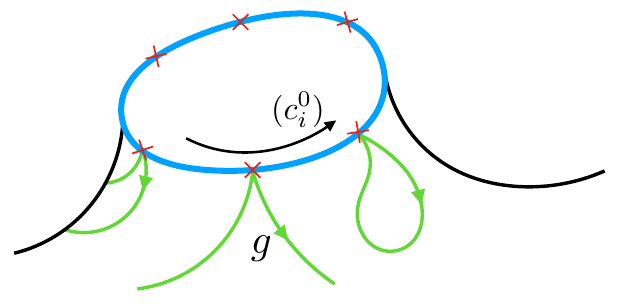}
\end{equation}

\end{definition}

\begin{theorem}
\label{thmgen}
For all $g \in S_{N}$ write:
\begin{equation}
\label{defG}
d(g_0,g) + d(g, e) = d(g_0, e) + 2( \#(g_0) - \#(g \vee g_0) ) + 2 G_{g_0}(g)
\end{equation}
then $G_{g_0}(g) \geq 0$. Furthermore there exists an \emph{admissiable} surface for $g$ that has genus $G_{g_0}(g)$ and $\#(g \vee g_0) $ connected components. This is the minimal possible genus and maximal possible number of connected components. 
\end{theorem}

\begin{proof}
It is clear we can always work with surfaces that have $\#(g \vee g_0) $ connected components and this is the maximal number.
Without loss of generality we can now assume that $\#(g \vee g_0) = \#(q_{g_0}(g)) = 1$, since the more general case is then just a sum over the partitions in $q_{g_0}(g)$. 

Firstly there always exists at least one \emph{admissible} surface (not necessarily with minimal genus) since we can simply thicken the lines defined by $g$ into tubes and connect these tubes onto  $\#(g_0)$ disks near the marked points on the boundaries of the disk and according to $g$. Then each curve segment simply pass through their respective tubes and do not cross each other. 

Consider some \emph{admissible} surface $\Sigma$. Let $c_i \in g$ be a cycle. The corresponding line segments are $\mathcal{L}^i_k$ where $k = 1, \ldots | c_i |$. Consider a closed curve $C_i$ that hugs tightly to the line segments in this cycle - following the direction of the cycle and strictly ``inside'' the boundary anchored curve $\cup_k \mathcal{L}^i_k$. Locally ``inside'' is defined such that the curve does not intersect the boundary. It also does not intersect any of the $g$ curves by the non-crossing condition. Perform surgery on this closed curve $C_i$. That is cut along the curve, and insert two disks to close up the surface along the cuts. Discard any boundaryless Riemann surface that gets disconnected under this process - by the non-crossing condition the result is a single surface with the original $\#(g_0)$ boundaries. There can only be one boundaryless surface $\mathcal{B}$ that is discarded. This produces a new Riemann surface $\Sigma'$ that 
is also \emph{admissible}. The genus of this new surface must decrease:
\begin{equation}
\label{geni}
G(\Sigma') \leq G(\Sigma)
\end{equation}
This is because the Euler character under connected sum must decrease by an amount corresponding to the Euler character of a sphere: 
\begin{equation}
 \chi(\Sigma)  = \chi(\Sigma'  \sqcup \mathcal{B})  - 2
\end{equation}
The Euler character of $\mathcal{B}$ is bounded above by $2$ and contributes additively:
\begin{equation}
 \chi(\Sigma) = \chi(\Sigma') + \chi(\mathcal{B}) -2 \leq \chi(\Sigma')
\end{equation}
Since the number of simple boundaries in $\Sigma$ and $\Sigma'$ is the same we get \Eqref{geni}. On $\Sigma'$ the curve $C_i$ is now contractible. 

Continue this process for all cycles in $g$ such that the corresponding curves  $\cup_k \mathcal{L}^i_k$ for all $i$ are contractible. Similarly by including line segments between adjacent marked points on the boundaries, oriented opposite to the $c_0's$, we can cut along closed cycles on the ``outside'' of the cycles. It is easy to see that these can be represented by the cycles
in $(g_0) g^{-1}$. Applying surgery to all of these cycles gives the final surface that we call $\Sigma_0$. We now give a triangulation of the surface $\Sigma_0$. 
We have edges corresponding to the curves defined by $g$. There are $N$ of these. There are also edges on the boundaries of $\Sigma_0$ between the marked points. There are also $N$ of these. So we have $ E = 2N$. The vertices have 4 lines meeting at the boundary for each marked point on the boundary. 
There are $ V = N$ of these. See for example:
\begin{equation}
\nonumber
\centering
\includegraphics[scale=.7]{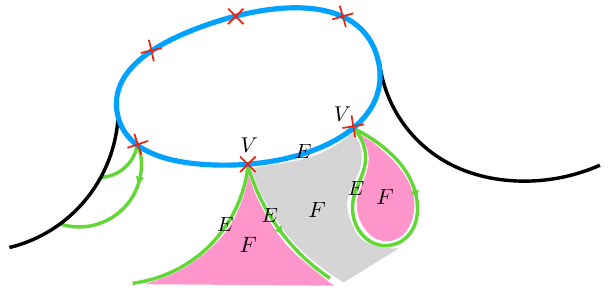}
\end{equation}

The faces are the interiors of the cycles in $g$ and $ g g_0^{-1}$ so $F = \#(g) + \#(g (g_0)^{-1})$. Thus:
\begin{align}
\chi &= 2 - 2G(\Sigma_0) - \#(g_0)  \\ &= V - E + F = -N + \#(g) + \#(g g_0^{-1})
=  N - d(g, e) - d( g, g_0)
\end{align}
Thus:
\begin{equation}
2G(\Sigma_0) = 2 + d(g, e) + d( g, g_0)  - d(g_0,e) - 2 \#(g_0)
\end{equation}
Which gives $G(\Sigma) \geq G(\Sigma_0) = G_{g_0}(g)$. In particular this implies that $\Sigma_0$ has the minimal genus, since for any other
valid surface $\widetilde{\Sigma}$ with lower genus $G(\widetilde{\Sigma}) < G_{g_0}(g)$ we run the above surgery argument and arrive at a contradiction: $G(\widetilde{\Sigma}) \geq G_{g_0}(g)$ . \end{proof}

Finally we give a definition of the multi-annular non-crossing permutations.
\begin{definition}
\label{def:anc}
A multi-annular non-crossing permutation $g$ for $g_0$ with $\#(g_0) > 1$ is a group element with $G(g) = 0$ as defined in \Eqref{defG}
and such that $\#(g\vee g_0) = 1$ (fully connected).  We will denote these $ANC_{g_0} = ANC_{ |c^1_0|, |c^2_0| \ldots}$. 
\end{definition}

\section{Proof of the 1-site phase diagram}
\label{app:phase_proof}
In this appendix we give a proof for the single site phase diagram \figref{fig:1site_phase} appeared in \secref{sec:phase_nm}.
This proof draws heavily from the results in Appendix~\ref{app:perm}, in particular the theorem given in Appendix~\ref{app:perm_topology}.

To set up the problem, we consider the group $S_{N} = S_{nm}$ and the elements $g_A$ and $g_B$ with $n$ cycles each and defined as
\begin{align}
g_B &= (1 \ldots m) (m+1 \ldots 2m) \ldots (nm -m+1  \ldots nm) \\
g_A & = (m/2+1 \ldots 3m/2)(3m/2+1 \ldots 5m/2) \ldots (nm-m/2+1 \ldots m/2)
\end{align}
where the cycles contain all element that appear in between the numbers shown
and are cyclicly ordered.
These are defined for $n \geq 1$ and $m/2 \geq 1 $ integers.
These have the property that $\Gamma(g_A, e) \cap \Gamma(g_B,e)$  is non-trivial (contains more than the identity), although the triple intersection
$\Gamma(g_A, e) \cap \Gamma(g_B,e)  \cap \Gamma(g_A,g_B) = \emptyset$. This is the main origin of ``frustration'' in the problem below, and distinguishes the reflected entropy from negativity. There is a unique element we call $X$ that satisfies the property $X \in \Gamma(g_A,e) \cap \Gamma(g_B,e)$ and it minimizes $d(X,g_{A,B})$.
It has $2n$ cycles:
\begin{equation}
X = (1 \ldots m/2) (m/2+1 \ldots m) \ldots (nm-m/2+1 \ldots nm)
\end{equation}
Note that $P(X) = P(g_A)\wedge P(g_B)$.
If $m=2$ then $X=e$. 

The free energy function we wish to minimize is:
\begin{equation}
\label{eq:f(g)_simple}
f(g) = x_A d(g,g_A) + x_B d(g,g_B) + d(g,e) 
\end{equation}
for  $x_{A,B} \geq 0$. We wish to prove:
\begin{theorem}
\label{thmmin}
For all $x_A, x_B > 0$ then the following minimum is achieved on a simple four element subset:
\begin{equation}
f_{\rm min} \equiv \min_{g\in S_{mn}} f(g) = \min_{g \in \{ e, X, g_A, g_B \} } f(g) 
\end{equation}
\end{theorem}
\noindent\emph{Remark.} This theorem only states that the minimum of $f(g)$ can be achieved for $g\in{e,X,g_A,g_B}$, but does not exclude the possibility of other elements also saturating the minimum.
In fact there exists other minimal elements $g\in S_{mn}$ that lives at the phase boundaries of the phase diagram \figref{fig:1site_phase}, and they are crucial for smoothing out the phase transition near $x_A + x_B = 1$.
We will investigate the detailed form of these elements in Appendix \ref{app:2nd_resum}.

\noindent \emph{Preliminary.} We firstly note there are regions in the phase diagram that are easy to deal with:
\begin{align}
x_A + x_B < 1: \qquad f(g) & \geq  d(g,e) ( 1- x_A -x_B)  + f(e)  \geq f(e)
\end{align}
where we have used the triangle inequality: $d(g_{A,B},g) + d(g,e) \geq d(e, g_{A,B})$.
So we have equality iff $d(g,e) = 0$. For $x_A + x_B = 1$ then we have quality for $g \in \Gamma(g_A, e) \cup \Gamma(g_B, e)$. 
Similarly:
\begin{align}
x_A > x_B + 1: \qquad f(g) & \geq (x_A - x_B - 1) d(g, g_A) + f(g_A) \geq f(g_A)
\end{align}
where we have used the triangle inequality: $d(g_{A},g) + d(g,g_B) \geq d( g_{A},g_B)$ and $d(g_{A},g) + d(g,e) \geq d(e, g_{A})$.
Equality is achieved iff $g = g_A$. For $x_A = x_B + 1$ we still only have $g_A$ as the minimal element since
the intersection $\Gamma(g_A, g_B) \cup \Gamma(g_A, e) = \{ g_A\}$. Similarly for $A \leftrightarrow B$. Thus the \emph{non-trivial region}  is
$x_A + x_B > 1$, $x_A < x_B +1$ and $x_B < x_A +1$. Indeed the phase diagram is convex:
\begin{lemma}
\label{lem:convexity}
If the minimum for $f(g)$ is achieved for some $g_\star$ at two locations in the $(x_A,x_B)_{1,2}$ phase diagram then $g_\star$ is also minimal
at: 
\begin{equation}
(x_A,x_B)_\lambda = \lambda  (x_A,x_B)_{1} + (1- \lambda) (x_A,x_B)_{2}\, \qquad 0 \leq \lambda \leq 1
\end{equation}
\end{lemma}
\begin{proof}
Note that (in hopefully clear notation) for $g \in S_{mn}$:
\begin{equation}
f_\lambda(g) = \lambda f_1(g) + (1- \lambda) f_2(g) \geq \lambda f_1(g_\star) + (1- \lambda) f_2(g_\star) = f_\lambda(g_\star)
\end{equation}
\end{proof}
Thus, for the general Theorem~\ref{thmmin} we can limit ourselves to the line $x_A = x_B > 1/2$. Convexity will do the rest, since all four elements $\{ g_A,g_B, X , e \}$ are already represented
somewhere on the phase diagram away from the \emph{non-trivial region}. For $m=2$ there is a much simpler proof than the proof discussed below, we present this in Appendix~\ref{app:m=2_proof}.

\begin{proof} (of Theorem~\ref{thmmin})
\label{mainthm}
As discussed above we need only consider $x_A = x_B = x \geq 1/2$.
Consider the topological discussion of Appendix \ref{app:perm_topology} for $g_0 = X$. We classify all elements in $S_{mn}$ using their connectedness $q_X(g)$
  over $X$.
Recall that $q_X(g) = (P(g)\vee P(X)) / P(X)$.
The quotient is a partition $ q_X(g) \in P_{2n}$.

Set $q = q_X(g)$ and  $n_1 = \#_1(q)$, where $\#_1(q)$ counts the number of length $1$ blocks in $q$.
It satisfies $0 \leq n_1 \leq 2n$.  
Now write:
\begin{align}
\label{fbelow}
f(g) &= f(e) + (1-2x) d( e, X)  +  n - \lfloor  n_1/2 \rfloor   - 2 x \delta_{n_1,0} \\ &
\label{pos1}
+ 2 x( G_A  - G_X) + 2 x(G_B -  G_X) + 2 ( G_X)   \\ \label{pos2}&
+  (2 x - 1)  \left( \vphantom{ \bigoplus } d(g,X) -  2n + \#( q) \right) 
 \\ \label{pos3} &+ 2x \left( \vphantom{ \bigoplus } \#(q) - \#(q \vee t_A) - \#(q \vee t_B) + \delta_{n_1,0} \right)\\ & \label{pos4}+ \left( \vphantom{ \bigoplus } n + \lfloor  n_1/2 \rfloor -\#(q) \right) 
\end{align}
where $\lfloor \cdot \rfloor$ is the floor function, $G_{A,B} \equiv G_{g_{A,B}}(g)$ and $G_X \equiv  G_X(g)$ were defined in Lemma~\ref{thmgen}. And 
$t_{A,B} = q(g_{A,B})$ 
or more specifically: $t_B = (12)(34) \ldots (2n-1, 2n)$
and $t_A = (23)(34) \ldots (2n,1)$.  To arrive at the formula above we have used:
\begin{equation}
\#( g_A \vee g)  = \#( t_A \vee q(g) )\,, \qquad \#( g_B \vee g)  = \#(t_B \vee q(g))
\end{equation}
which follows from $\#( g_A \vee g)  = \#( P(g_A) \vee P(g) )$
and since $P(g_A) \geq P(X)$ this implies that $P(g_A) \vee P(g) \geq P(X)$ and so we can take the set quotient  $ (P(g_A) \vee P(g))/ P(X) =
(P(g_A) \vee (P(g) \vee P(X)))/ P(X) =  t_A \vee q(g)$. And the number of sets is the same under the quotient $\#( P(g_A) \vee P(g) ) = \#(t_A \vee q_g)$. Similarly for $\#(g_B \vee g)$.

We aim to show that each bracketed terms in lines (\ref{pos1}-\ref{pos4}) are all positive. That is we wish to establish the estimate:
 \begin{equation}
 f(g) \geq  f(e) + (1-2x) d(e, X) +  (n - \lfloor  n_1/2 \rfloor )  - 2 x \delta_{n_1,0}
 \end{equation}
 Assuming this is the case then:
 \begin{equation}
 \label{tolast}
 \min_{g \in S_{mn}} f(g) \geq f(e) + (1-2x) d(e, X) + \min_{0 \leq n_1 \leq 2n} \left( (n - \lfloor  n_1/2 \rfloor )  - 2 x \delta_{n_1,0}
 \right) = \min_{g \in {e,X, g_A,g_B}} f(g)
 \end{equation}
which ,together with:
\begin{equation}
 \min_{g \in S_{mn}} f(g) \leq \min_{g \in {e,X, g_A,g_B}} f(g)
\end{equation}
proves the theorem.  The last step in \Eqref{tolast} is by direct computation. It also follows since the bounds  (\ref{pos1}-\ref{pos3})  that we derive below
are all tight for the elements $g = X, g_A, g_B$. 

We need the bound $G_A \geq G_X$ in line \Eqref{pos1}. This follows from the construction of the genus. We use the surface $\Sigma_0$ for $g$ based over $g_A$
which has genus $G_A$. We then
deform the $n$ boundaries of this surface into $2n$ boundaries by pinching - dividing the $m$ marked points into two sets of $m/2$ marked points on the new boundary (see figure below).
This deformed $\Sigma_0'$ is an \emph{admissible surface} of genus $G(\Sigma_0)$ for $g$ based over $X$ since we can do this deformation without touching any of the curves. 
By Theorem~\ref{thmgen} we must have $G_A = G(\Sigma_0) = G(\Sigma_0')  \geq G_X$.  Similarly $G_B \geq G_X$ and $G_X \leq 0$ always.
\begin{equation}
\nonumber
\centering
\includegraphics[scale=.7]{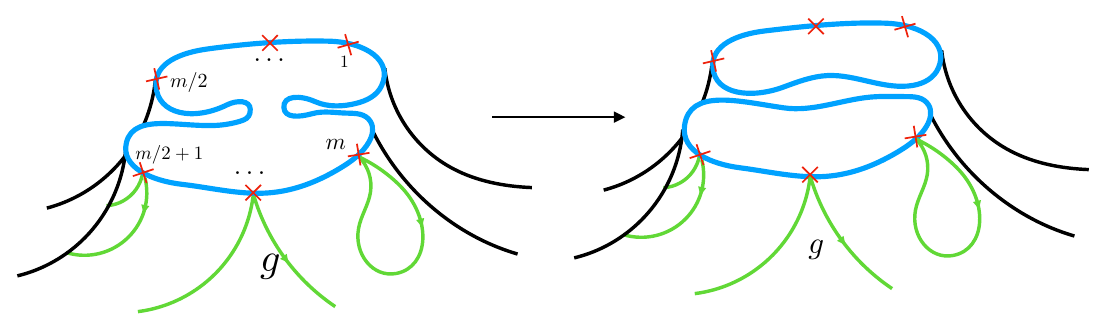}
\end{equation}

Moving to the next line we also need the lower bound \Eqref{pos2}:
\begin{equation}
d(g,X) \geq ( 2n - \#(q_X(g)))
\end{equation}
Note that $q_X(X^{-1} g) = q_X(g)$ since $\left< X, g \right> = \left< X, X^{-1} g \right>$. 
So we can equivalently prove 
$ d(g,e) \geq ( 2n - \#(q_X(g))) $ for all $g$. This follows from Lemma~\ref{lemmo} applied to $g_0 = X$ and $N = nm$.

We also need to bound \Eqref{pos3}:
\begin{equation}
\label{finest}
\#(q) - \#(q \vee t_A) - \#(q \vee t_B)  + \delta_{\#_1(q),0} \geq 0
\end{equation}
 To do this we use  Lemma~\ref{lempart}, which simply follows from the semimodular condition \Eqref{semmod}.
We apply this Lemma with $N=2n$, $q=q$, $s = t_B$ and $t = t_A$, and use the fact that
$t_A \vee t_B =  (\mathbb{Z}_{2n})$ and $t_A\wedge t_B = e$, giving the estimate:
\begin{equation}
\#(q) - \#(q \vee t_A) - \#(q \vee t_B) +1 \geq 0
\end{equation}
We can improve this estimate as follows.
If there is at least one block in $q$
of length-$1$ then we remove one of the double blocks in $t_B$ where this length-$1$ $q$-block would overlap. That is we split this double block into two single blocks
to give a new $t_B' < t_B$. After doing this we still have:
\begin{equation}
\label{tbp}
t_A \vee t_B' = (\mathbb{Z}_{2n}) \, \qquad t_A \wedge t_B' = e\, \qquad \#(q \vee t_B') = \#(q \vee t_B) + 1
\end{equation}
such that $\#(q) - \#(q \vee t_A) - \#(q \vee t_B) \geq 0$. Together the final estimate is \Eqref{finest}. 

Finally line \Eqref{pos4} is positive since  $\#(q) \leq n + \left \lfloor{\#_1(q)/2}\right \rfloor $ (the floor) which comes simply from maximizing $\#(q)$ by splitting the remaining non length $1$ blocks into pairs if $\#_1$ is even, or pairs and a triplet if $\#_1$ is odd. 

\end{proof}

Above we needed the following results:

\begin{lemma}
\label{lemmo}
For all $g_0,g \in S_N$, then:
\begin{equation}
d(g,e) \geq  \#(g_0) - \#( g \vee g_0) 
\end{equation}
\end{lemma}
\begin{proof}
Map to the set of partitions $P_N$ and consider:
\begin{align}
d(g,e) - (\#(g_0)  - \#( g \vee g_0) 
) &= N - \#(P(g)) - \#(P(g_0)) + \#(P(g) \vee P(g_0)) \\& \hspace{-2cm} \geq \#(P(g) \wedge P(g_0)) - \#(P(g)) - \#(P(g_0)) + \#(P(g) \vee P(g_0)) \geq 0
\end{align}
where the first equality follows since for all partitions $\#(P) \leq N$ and the second inequality uses the semi-modularity property of $\rho$ on $P_N$. 
Saturation requires that $P(g) \wedge P(g_0) = e$ and also that for all $p \leq P(g_0)$ then:
\begin{equation}
P(g_0) \wedge ( p \vee P(g)) = (P(g_0) \wedge P(g)) \vee p = p
\end{equation}
\end{proof}

\begin{lemma} Given three partitions $q, t, s \in P_N$ then:
\label{lempart}
\begin{equation}
\#(q \vee (t \wedge s)) - \#(q \vee t) - \#(q \vee s) + \#(q \wedge (t \vee s) \geq 0
\end{equation}
\end{lemma}
\begin{proof}
Since set partitions form a $\rho$-graded semimodular lattice with: $\rho(q) = N - \#(q)$
we know that the grading satisfies:
\begin{equation}
\rho(q_1 \vee q_2) + \rho(q_1 \wedge q_2) \leq \rho(q_1) + \rho(q_2)
\end{equation}
Set $q_1 =q \vee t$ and $q_2 =q \vee s$, then $q_1\vee q_2 =q \vee s\vee t$ and $q_1\wedge q_2  \geq q\vee (t\wedge s)$ (since $t \geq t\wedge s
$ implies that $q \vee t \geq q \vee (t\wedge s)$.) Thus:
\begin{equation}
\rho(q \vee s \vee t)+\rho(q \vee (t\wedge s)) \leq \rho(q_1 \vee q_2)+\rho(q_1 \wedge q_2) \leq \rho(q \vee t)+\rho(q \wedge s) 
\end{equation}
as required.
\end{proof}

\subsection{Simpler proof at $m=2$}
\label{app:m=2_proof}
There is an independent proof of Theorem~\ref{thmmin} for $m=2$, where there is only a three element subset on the right hand side. A quick sketch:
\begin{proof}{(For $m=2$.)} Set $x_A = x_B = x > 1/2$. Write:
\begin{equation}
f_{\rm min} = \min_{g \in S_{2n}} f'(g) \, \quad f'(g) = f( g_A g) =   x( d(g, g_A^{-1}g_B) + d(g,e)) + d(g, g_A^{-1}) 
\end{equation}
where it is easy to see that $ g_A^{-1}g_B$ is made of two cycles of length $n$. 
We have the more general bound 
\begin{equation}
f'(g) \geq x ( d ( g_A^{-1} g_B,e ) + 2( \#(g_A^{-1} g_B) - \#( g_A^{-1} g_B \vee g) )) + d(g, g_A^{-1})
\end{equation}
where recall $ 1\leq \#( g_A^{-1} g_B \vee g) \leq  \#(g_A^{-1} g_B) = 2 $ is the number of connected components discussed above.
We minimize over the partition of $S_{mn}$ defined by the integer $\#( g_A^{-1} g_B \vee g) = 1,2$:
\begin{equation}
f_{\rm min,1} =  \min_{g \in S_{2n}: \#( g_A^{-1} g_B \vee g) = 1 } f'(g)  \geq x ( d ( g_A^{-1} g_B,e )  + 2) + d(g,g_A^{-1}) \geq x 2n
= f'(g_A^{-1}) 
\end{equation}
with equality iff $g = g_A^{-1}$.  And also
\begin{align}
f_{\rm min,2} =  \min_{g \in S_{2n}: \#( g_A^{-1} g_B \vee g) = 2 } f'(g)  & \geq x  d ( g_A^{-1} g_B,e )  + d(g,g_A^{-1}) 
\\ &=x d ( g_A^{-1} g_B,e )  +(2n - \#(k q))  
\geq f'(e) 
\end{align}
where we have used the fact that all elements with $\#( g_A^{-1} g_B \vee g) = 2 $ must take the form $g = k_1 q_2$ where $k, q\in S_n$ and acts
on the respective elements in the two cycles of $g_A^{-1} g_B$.
The minimal is achieved for $f_{\rm min,2} $ iff $g \in \Gamma(g_A^{-1} g_B, e)$ and $k = q^{-1}$. 
Minimizing over the two different sets gives:
\begin{equation}
f_{\rm min} = {\rm min}_{g \in e, g_A, g_B} f(g)
\end{equation}
where the minimum is achieved iff $g \in  \{ g_A k_1 (k^{-1})_2 : k \in \Gamma(\tau_n , e) \} \subset \Gamma(g_A, g_B)$
or $g = e$.
\end{proof}

\section{Reflected resolvent via direct group summation}
\label{app:2nd_resum}
In this appendix we provide a parallel approach for finding the reflected entropy resolvent \Eqref{eq:ref_spectrum}.
We consider factorization of $g\in S_{mn}$ into different $q \equiv q_X(g)$ sectors.
In each sector we will find conditions that must be satisfied for $g$ that minimizes the free energy. 
By restricting the full permutation group to these special elements we are able to arrive at an expression of the reflected entropy resolvent that matches the form given in main text.

\subsection{Minimal elements in a fixed sector}
The problem of finding elements that minimizes the free energy function \Eqref{eq:f(g)_simple} factorizes into two parts. 
Firstly, we seek for minimal elements for a fixed $\#(q)$ sector that saturates the two conditions \Eqref{pos3} and \Eqref{pos4}:
\begin{align}
    \#(q) - \#(q\vee t_A) - \#(q\vee t_B) + \delta_{n_1,0} = 0, \quad
    n + \lfloor n_1/2 \rfloor - \#(q) = 0
\end{align}
Secondly, for a fixed $q\in S_{2n}$, it should be possible to find all the minimal $g\in S_{mn}$ elements where:
 \begin{equation}
 \label{sat0}
 G_A = G_B = G_X = 0\,, \qquad d(g,X) =  2n - \#( q)
 \end{equation}
 Note that at $x = 1/2$ (which is at the vicinity of reflected entropy phase transition at $n=1$) we can drop the latter condition in which case the answer can be written in terms of multi-annular non-crossing elements.
 These elements marks the contribution of the new dominant saddles that smooth out the phases transition.  
 We will return to this problem when we have a better handle
 on the minimal $q$ elements discussed next.

Let's begin with the first minimization problem.
We will look for minimal elements for fixed $\lceil n_1/2 \rceil = 0,\ldots n$ sector.
Note that we have, somewhat arbitrarily chosen to fix the ceiling of $n_1(q)/ 2$ since this removes the odd case.
\footnote{Recall that the saturation of $\#(q)\leq n + \lfloor n_1/2 \rfloor$ enforces the non length-1 blocks of $q$ into pairs if $n_1$ is even, or pairs and a triplet (which we do not have a good handle on) if $n_1$ is odd. If we instead fix the ceiling then one can show that the inequality will never be saturated when $\#_1(q)$ is odd.
This will likely not change our conclusion of this section and we suspect that the effect of odd $n_d$ sector will only serve as a correction to the resolvent.}
That is fix $n_d = 0, 1 \ldots n$ then look for $q$'s such that:
\begin{equation}
\label{equal}
\#_1(q) = 2 n_d, \qquad
\#(q) = n + n_d, \qquad \#(q \vee t_A) + \#(q \vee t_B) =   n + n_d + \delta_{n_d,0}  
\end{equation}

\begin{theorem}
\label{partq}
For fixed integer $0 \leq n_d \leq n$ all partitions which satisfy \Eqref{equal} take the form $q = P( \hat{h} g_B|_{m=2})$ where $\hat{h} \in S_{2n}$ and is best described as the following diagram:

\begin{figure}[H] 
\nonumber
\centering
\includegraphics[scale=1.3]{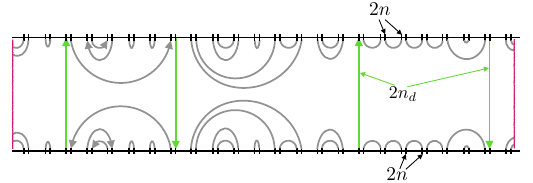}
\label{fig:annular-nc}
\end{figure}
\vspace{-.5cm}
In words, form the diagram by placing $2n$ points on the top (representing  element $135\ldots (2n-1)$) and $2n$ points on the bottom (representing element $246\ldots (2n)$). 
The left and right ends of the strip are further identified (so it is really a annulus).
$\hat{h} \in ANC_{n,n} \subset S_{2n}$ is then drawn by connecting the points pairwise without crossing as per the permutations in $\hat{h}$, where we split each of the $2n$ elements into two, each representing the incoming and outgoing lines of $\hat{h}$.
Pick $2n_d$ vertically aligned defect elements that is unique up to cyclic translation.
$\hat{h}$ is then constrained to take the form that it directly connects the $2n_d$ pairs of vertical defects (one-way) and otherwise only connects within the top/bottom half strip as an inverse with respect to each other.

\end{theorem}

\begin{proof}

 Fix $g_{A,B} = g_{A,B}|_{m=2}$ the $m=2$ version of these group elements. That is $g_B = (12)(34) \ldots (2n-1,2n)$ and $g_A = (23)(34) \ldots (2n,1)$.
 The first two conditions \Eqref{equal} require $q$ to be composed of pairs or singlets. In this case there is a canonical map to a permutation element $h_q \in S_{2n}$ since the order does not matter for cycles of length $2$ or $1$. Given this, we can easily compare $\#( h_q g_B)$ and  $\#( q \vee t_B) $ since the latter is either made of ``closed even blocks'' (a block that does not include any singlets) or ``$n_d$ blocks'' that start and end on a singlet.
 The counting is doubled in $\#( h_q g_B)$ for the closed blocks in $\#( q \vee t_B) $ (since the block factors into an even and odd orbits under the action of $h_q g_B$) while it is not for the blocks containing singlets. Thus:
\begin{equation}
\#( h_q g_A) = 2  \#( q \vee t_A) - n_d \,, \qquad \#( h_q g_B) = 2  \#(  q \vee t_B)  -n_d
\end{equation}
Substituting these relations into \Eqref{equal} and defining $ \hat{h} = h_q g_B$ one can show that the conditions \Eqref{equal} are equivalent to
\begin{equation}
\label{ncc}
\#( \hat{h} ) + \#( \hat{h} g_B g_A) = 2n + 2 \delta_{n_d,0}\,, \qquad \#( \hat{h} g_B )  = n+ n_d\,, \qquad \#_1( \hat{h}  g_B) = 2 n_d
\end{equation}

Define $\tau \equiv g_A g_B =(135 \ldots 2n-1) (246 \ldots 2n)^{-1}$ . We will use the notation $(k)_1(q)_2$ for $k,q\in S_n$ to mean permuting
$(1,3,5 \ldots 2n-1)$ according to $k$ and  the elements $(2,4,6\ldots 2n)$ according to $q$. So $\tau = (\tau_n)_1 (\tau_n^{-1})_2$.
We then recognize the first condition in \Eqref{ncc} as the problem of finding annular non-crossing permutations: $ANC_\tau = ANC_{n,n}$, i.e.
\begin{equation}
    d(\hat{h},e) + d(\hat{h},\tau) = d(\tau,e) + 2(\#(\tau)-\#(\hat{h}\vee\tau))
\end{equation}
which implies that the genus $G_\tau$ must be zero, see Definition~\ref{def:anc}.

For $n_d =0$  we can only have factorized $\hat{h} = (k)_1 (k^{\prime-1})_2$ on the two cycles in $g_A g_B$ where $k,k' \in NC_n$ are non-crossing permutations. Now also $\#(h_q) = \#(\hat{h} g_B) = \#(k k^{\prime-1})= 2n$  implying that $k = k'$. 
Thus:
\begin{equation}
q= P ((k)_1 (k^{-1})_2 g_B)
\end{equation}

For $n_d\neq 0$  we have annular non-crossing permutations. In particular we can only get one-cycles in $ \hat{h}  g_B$  iff there are \emph{straight} crossings between the two sectors. Thus we are looking for a certain class of annular non-crossings with $2 n_d$ straight crossings -- where we count either direction of crossings. 

Draw $\hat{h}$ as shown in the diagram in the statement of the theorem.
Such a diagram can be intepreted as an operator acting on the Hilbert space $(\mathcal{H}_\chi \otimes \mathcal{H}_\chi^\star)^{\otimes n}$ of dimension $\chi^{2n}$.
We define this corresponding operator as $D(\hat{h})$:
\footnote{The relevant non-crossing diagrams can be understood as arising from the \emph{affine Temperley-Lieb} (TL) algebra on $2n$ strands, which in turn has a representation acting on this Hilbert space. Indeed there is a well known correspondence between the annular non-crossing permutations and the affine TL algebra. }
\begin{equation}
\Braket{e_{i_1 j_1} \otimes e_{i_2 j_2} \ldots \otimes e_{i_n j_n}|D(\hat{h})| e_{j_{2n} i_{2n} } \otimes e_{j_{2n-1} i_{2n-1} } \ldots \otimes  e_{j_{n+1} i_{n+1} }}
\equiv \prod_{k=1}^{2n} \Braket{i_k| j_{\hat{h}(k)}} 
\end{equation}
where $e_{ij} = \left| i \right> \left< j \right|$ for some basis $ \left| i \right> $ on the $\mathcal{H}_\chi $ Hilbert space.

Given $\hat{h} \in ANC_{\tau}$ or it's corresponding $D(\hat{h})$, we can define $s(\hat{h})$ as the number of straight crossings and $t(\hat{h})$ as the total number of crossings. These are both even numbers. 
We can form $\hat{h}$ by considering two non-crossing permutations $k_1,k_2\in NC_{n+n_d}$ each with $t$ defect elements placed cyclically together and such that $k$ is constrained to connect all the $n_d$ defect elements to and from other non-defect elements.
The first defect is further constrained to connect the non-defect element directly next to it, as shown in \figref{fig:cut_glue}.
$\hat{h}\in ANC_{n,n}$ can be constructed by cutting open the connections at the defects of $k_1$ and $(k_2)^{-1}$ and glue the open connections in order, which we denote by $\hat{h} = k_1 \#_t k_2^{-1}$.
We have the obvious bound $s(\hat{h}) \leq t(\hat{h})$. 

\begin{figure}[h]
    \centering
    \includegraphics[width=.8\textwidth]{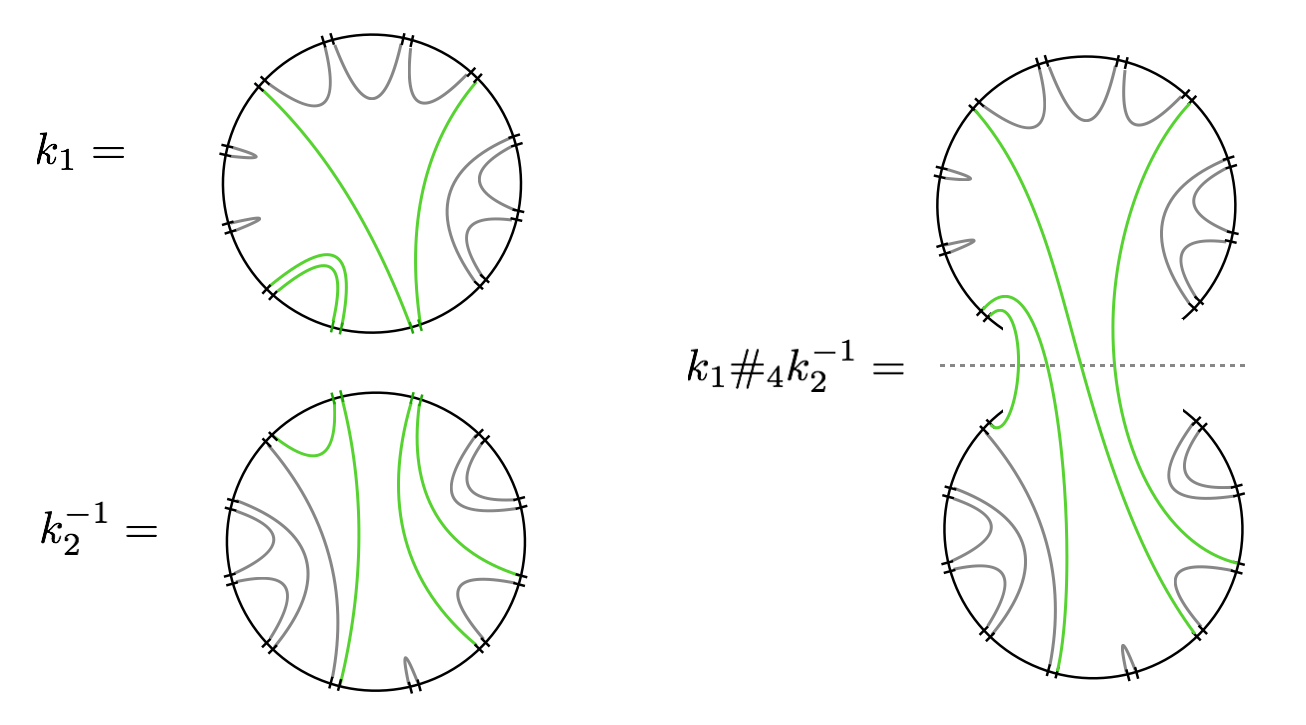}
    \caption{The procedure for constructing $k_1\#_t k_2^{-1}$. The green lines represent the would-be crossings.}
    \label{fig:cut_glue}
\end{figure}

Now we seek elements $\hat{h} \in ANC_{\tau}$, with $s(\hat{h}) = 2n_d$ such that
\begin{equation}
\chi^{\#(\hat{h} g_A)} =  {\rm Tr}_{\mathcal{H}_\chi^{\otimes 2n}} D( \hat{h}) = \chi^{n + n_d}
\end{equation}
In fact we have the following inequality:
\begin{equation}
\label{bdbd}
{\rm Tr}_{\mathcal{H}_\chi^{\otimes 2n}}D( \hat{h} ) 
 \leq \chi^{ n+ s(\hat{h})/2}
\end{equation}
and this is saturated iff $\hat{h} = k \#_t k^{-1}$ for some $k \in NC'_{n+t/2}$ in which case $s=t = 2 n_d$. This proves the theorem (after applying an arbitrary rotation by conjugating by powers of
to $(\tau_n)_1 (\tau_n)_2$.) 
We prove \Eqref{bdbd} and the saturation condition just used in Lemma~\ref{lembd}. 
\end{proof}

\begin{lemma}
\label{lembd}
For any $\hat{h} \in NC_{2n}$ then we have the estimates:
\begin{equation}
\label{twoineq}
{\rm Tr}_{\mathcal{H}_\chi^{\otimes 2n}}D( \hat{h} ) \leq \chi^{ n+ s(\hat{h})-t(\hat{h})/2}  \,, \qquad {\rm Tr}_{\mathcal{H}_\chi^{\otimes 2n}}D( \hat{h} ) \leq \chi^{ n+ s(\hat{h})/2} 
\end{equation}
Furthermore the later inequality is saturated iff $\hat{h} = k \#_t k^{-1}$ for some $k \in NC'_{n+t/2}$ where $t = t(\hat{h})$. 
\end{lemma}
\begin{proof}
Note that:
\begin{equation}
{\rm Tr}_{\mathcal{H}_\chi^{\otimes 2n}}D( k \#_t k^{-1}) = \chi^{ n+ t/2} 
\end{equation}
by direct computation. 

Consider now the first inequality \Eqref{twoineq}. We can remove any \emph{straight crossings} since they factor out trivially from both sides
of this inequality. So wlog consider only elements 
with $s(\hat{h}) = 0$. We can then write:
\begin{equation}
{\rm Tr}_{\mathcal{H}_\chi^{\otimes 2n}} D( k_1 \#_t k_2^{-1} )
= \left< k_1' \right| \Sigma \left| k_2' \right> \, \qquad k_{1,2}' \in NC_{n - t}
\end{equation}
defined as follows (see Figure below). Here $ k_{1,2}'$ are constructed from $k_1$ and $k_2$ by removing
the defects/crossing lines and considering the remaining elements as a non-crossing permutation on $n-t$ points. 
These can then be interpreted as pure states (non-normalized maximally entangled states) 
in the Hilbert space of dimension $\chi^{ 2(n-t)}$. Finally $\Sigma$ is a unitary permutation  on this Hilbert space
formed by following crossing lines in $k_1 \#_t k_2^{-1}$ around the trace. 
\begin{equation}
\nonumber
\centering
\hspace{-.1cm} \includegraphics[scale=.58]{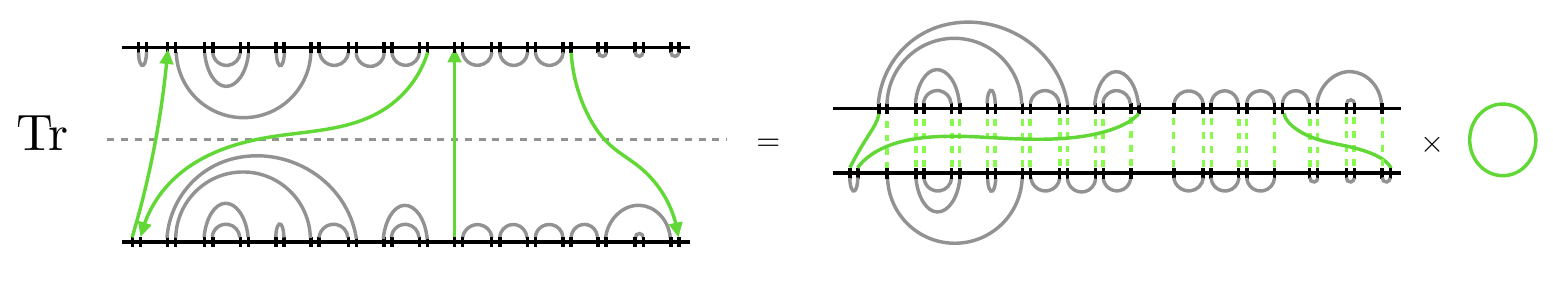}
\end{equation}
(where we have included a single straight crossing, that may be removed trivially.)

Then Cauchy-Schwarz gives:
\begin{align} 
|{\rm Tr}_{\mathcal{H}_\chi^{\otimes 2n}} D( k_1 \#_t k_2^{-1} )|^2
&= |\left< k_1' \right| \Sigma \left| k_2' \right>|^2  
\leq  \left< k_1' \right| \Sigma^\dagger \Sigma \left| k_1' \right>  \left< k_2' \right. \left| k_2' \right>\\
&=   \left< k_1' \right. \left| k_1' \right>  \left< k_2' \right. \left| k_2' \right> 
\\ &= \chi^{ -2 t} {\rm Tr}_{\mathcal{H}_\chi^{\otimes 2n}} D( k_1 \#_t k_1^{-1})  {\rm Tr}_{\mathcal{H}_\chi^{\otimes 2n}} D( k_2 \#_t k_2^{-1}) = \chi^{2n -t}
\end{align}
as required. The second inequality in \Eqref{twoineq} is trivial now. Saturation of this later inequality requires $t(\hat{h}) = s(\hat{h})$ and saturation
of the Cauchy-Schwarz inequality above (now with no $\Sigma$) requires $k_1' = k_2'$ implying the correct condition. 
\end{proof}

We now consider the saturation condition for $g$ in \Eqref{sat0}. To make progress we consider only the form of $q$ given by Theorem~\ref{partq}. 
\begin{theorem}
\label{thm:g(q)}
Fix an $n_d$ and a corresponding $q$ satisfying \Eqref{equal}. Then there is a unique $g\in S_{mn}$ that satisfies \Eqref{sat0}. This element is:
\begin{equation}
g(q) = \iota(\hat{h}) g_B
\end{equation}
where $\hat{h} \in S_{2n}$ is the unique element that satisfies $q = P( \hat{h} g_B|_{m=2})$ and where  $\iota $ embeds the subgroup $S_{2n}$ into $S_{mn}$. The subgroup acts only on elements $j m/2 + 1$ for $j =0 \ldots 2n-1$ fixing all other elements. 

If we relax the second condition in \Eqref{sat0} then all the dominant elements are those satisfying $g' :  P(g') \leq P(g(q))$ where $g(q)$ was defined above, and also $P(g') \nleq P(g(q'))$ for any $q' < q$.  There are:
\begin{equation}
(C_{m} - C_{m/2}^2)^{n - n_d} (C_{m/2}^{2})^{n_d}
\end{equation}
of these, where $C_m$ are the Catalan numbers.
\end{theorem}
\begin{proof}
We first consider the case where we ignore second part of \Eqref{sat0}.  
Given a connectivity fixed by $q$ satisfying \Eqref{equal} first consider the blocks of length $1$ in $a_i \in q$ . These contribute to all $G_X, G_A, G_B$ for a given $g'$ independently, and so we have zero genus for $g'$ restricted to these blocks, iff these $g'$ forms a NC permutation within the corresponding $m/2$ block. We call these NC permutation $\alpha_i \in NC_{m/2}$. It is not hard to see that $ \iota(\hat{h}) g_A$ contains a cycle $\tau_{m/2}$ on this same $m/2$ block. Thus the statements of the theorem are true for the $2n_d$ unit blocks. 

Now consider a block $b_i \in q$ of length $2$. From the form of $q$ given in Lemma~\ref{partq} we can see that such a block sits either in a closed block of $q \vee t_A$ or a closed block of $q \vee t_B$ (or both). 
Assume without loss of generality it is the former. Now recall the ``pinching argument'' in Theorem~\ref{thmmin}. Start with the zero genus surface $\Sigma_A$ describing $g'$ with resect to $g_A$.
Now pinch to form the surface with respect to $X$. Since $0\leq G_X \leq G_A  = 0$ the pinched surface must also be an admissible surface $\Sigma_X$ for $g'$ (with respect to $X$) that has minimal genus $0$. Focus on the curves connecting the two $m/2$ boundaries on $\Sigma_X$ associated to $b_i$ that describe a disconnected genus $0$ surface  with two boundaries $\Sigma_i$. These must form an annular non-crossing permutations in $S_m \subset ANC_{m/2,m/2}$
that we call $\beta_i$. We now show that these must be a special subset of annular non-crossing permutations. 

We will show that we can add to $\Sigma_i$
a curve that starts between the marked points $1$ and $m/2$ on the first boundary and ends between the marked points $1$ and $m/2$ on the second boundary, and does not cross any other curves on $\Sigma_i$. We construct this as follows (see Figure below). Since we found $\Sigma_X$ by pinching $\Sigma_A$ between the marked points $1$ and $m/2$  and $m/2+1$ and $m$ on the boundaries of $\Sigma_A$, we can  add \emph{new} non-crossing curves to $\Sigma_X$ that pass between pairs of boundaries described by the blocks in $t_A$ that start and end between the $1$ and $m/2$ marked points on each $X$ boundary.  On $\Sigma_X$ shrink all boundaries except for the two described by $b_i$ to arrive at $\Sigma_i$ plus
some bulk points where the boundaries shrunk, and non-crossing curves between the bulk points.  Since $q \vee t_A$ contains the block $b_i$ in a ``closed block'' we can find a curve on $\Sigma_X$ that follows alternatively the \emph{new} curves that we constructed using $t_A$ and the curves in $g'$ \emph{not} associated to $b_i$. The result connects, via the bulk points, the two remaining boundaries as required.

\begin{figure}[h!] 
\nonumber
\centering
\includegraphics[scale=.7]{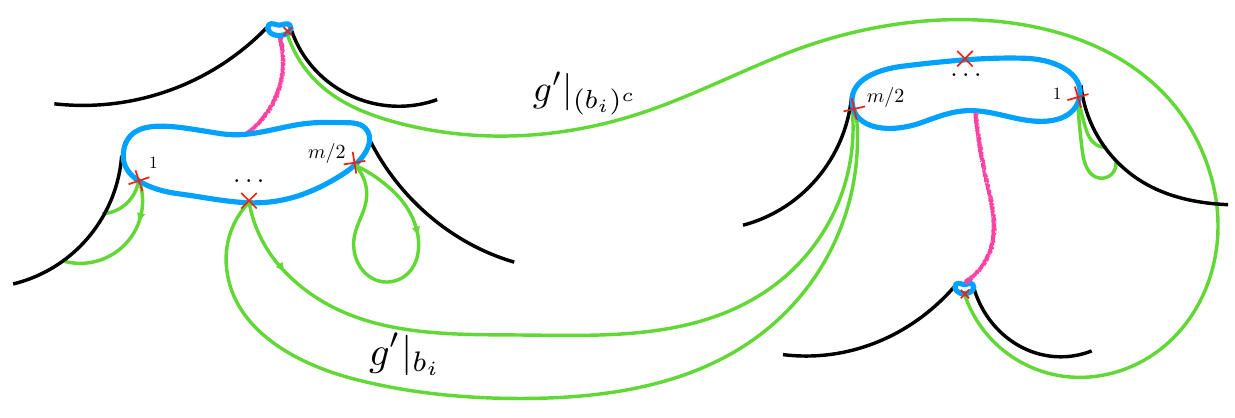}
\caption{The surface $\Sigma_X$ is constructed from $\Sigma_A$ in such away that we can insert non-crossing curves on $\Sigma_X$, shown in pink.
These curves pass between the two boundaries associated to the blocks in $t_A$.  We can use these curves and some other non-crossing
curves, labelled $g'|_{b_i^c}$, to pass between the two boundaries associated to the block $b_i \in q$.  }
\end{figure}

An annular non-crossing permutation ($ANC_{m/2,m/2}$) does not cross a line between the two boundaries iff it is a non-crossing permutation $NC_m$ associated to the joined boundaries, where we join the boundaries through the non-crossing line. This is a special subset $NC_m \subset  ANC_{m/2,m/2}$. It is not hard to see that in our case these are defined as non-crossing permutations with respect to the remaining cycles  $\tau_m$ of length $m$ in $g(q) \equiv \iota(\hat{h})) g_B$. Thus we solve the condition $G_X = G_A = G_B =0$ by demanding $P(g') \leq P(g(q))$ (such that the connectivity of $g'$ is still described by $q$ which requires that $P(g') \nleq P(g(q'))$ for $q' < q$.) 

For the first statement of the theorem (including the second condition in \Eqref{sat0}), we now simply need to compute the following for $g' \leq g(q)$:
\begin{align}
d(X, g') & = mn - \sum_i \#(\tau_{m/2} (\alpha_i)^{-1} ) - \sum_j \#(\tau_{m/2} \times \tau_{m/2} (\beta_j)^{-1} ) \\
& = mn - \sum(m/2+1- \#(\alpha_i)) - \sum ( m - \#(\beta_j) ) \\
& = (n-n_d) + \sum( \#(\alpha_i)-1) + \sum ( \#(\beta_j) -1)
\end{align}
where $n-n_d = 2 n - \#(q)$. So we get equality for the second part of \Eqref{sat0} iff $\alpha_i$ and $\beta_j$ have one cycle - namely they equal $\tau_{m/2}$ or $\tau_m$
within their respective blocks. 
Again it is not hard to check that the unique element that does the job is $g(q)$. 
\end{proof}

The surfaces that describe $g'$ with respect to $g_B$, are disconnected $NC_{m/2}$ discs and, the $k$-fold branched coverings of the disk. We give some pictures to describe the dominant saddles. 
\begin{figure}[H] 
\centering
 \includegraphics[scale=.58]{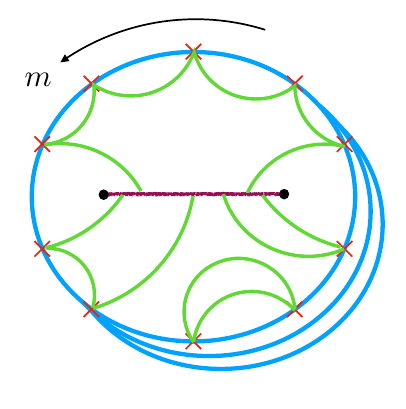}  
 \caption{A part of $g'$ corresponding to single block in $\#(q \vee t_B)$ and represented as a surface with respect to $g_B$. The boundary represented by the cycles of $g_B$ is shown in blue. The purple line can be thought of as a ``twist operator'' which relates the different branches of the block. 
 }
\end{figure}

Lastly we mention how these elements make their appearances in the phase diagram.
The minimal elements are constituted by $g(q) \in S_{mn}$ described in Theorem \ref{thm:g(q)} but also saturating \Eqref{pos1}, which in turn forces either $n_d=n$ or $n_d=0$.
They are:
\begin{align}
\label{eq:min_elements}
    \begin{cases}
    \{e\}, \quad & x<1/2, \\
    \{g : P(g) \leq P(X)\} = \Gamma(g_A,e)\cap \Gamma(g_B,e) , \quad & x=1/2, \\
    \{X\}, \quad & 1/2<x\le n/2, \\
    \{ k^{[1]_m}(k^{-1})^{[m/2+1]_m}g_B : k\in NC_n\} \subset \Gamma(g_A,g_B), \quad & x\ge n/2,
    \end{cases}
\end{align}
in terms of the notation of Sec. \ref{sec:SR_RTN}. 
We see that the $X$ element comes from $n_d=n$, whereas the $x\ge n/2$ elements come from $n_d=0$.
Intuitively one can also regard the elements with $0<n_d<n$ as smoothly interpolating between $\{X\}$ and $\{k^{[1]_m}(k^{-1})^{[m/2+1]_m}g_B\}$.
They are essential around the reflected entropy phase transition $x=n/2=1/2$, where all the four regions in \Eqref{eq:min_elements} come on top of each other.
The fact that phase transition occurs at $x=1/2$ also implies that we drop the second condition of \Eqref{sat0}, that is we must sum over all the elements prescribed in theorem \ref{thm:g(q)}.

\subsection{Generating functions}

Let us compute the generating function for elements that  saturate the $q$ conditions in \Eqref{equal}.
Define:
\begin{equation}
Z(z,w,p) = \sum_{n = 0}^{\infty} \sum_{n_d = 0}^{n} w^{n_d} z^{n}   \sum_{ q \in P_{2n} : 
\Eqref{equal} \checkmark } p^{-\#(t_A \vee q) + \#(t_B \vee q) }
\end{equation}
where we \emph{pick} the $n_d = n = 0$ term to equal $p$ (which then violates 
the seeming $p \rightarrow 1/p$ symmetry.) 
We use the result in Lemma~\ref{partq}. Let us separate out the $n_d =0$ contribution. For $n =n_d$ we simply get the generating function of the $q$-Catalan numbers:
\begin{equation}
Z|_{n_d = 0} = \sum_{n=0}^{\infty} \sum_{ a \in NC_n} z^n p^{-2 \#(a)+n+1} \equiv C(z,p) = - \frac{\sqrt{ (1 - z/z_+)(1-z/z_-)} }{2z} + \frac{1}{2z} - \frac{1}{2 \sqrt{z_+ z_-}}
\end{equation}
where $z_\pm = p (1 \pm p)^{-2}$ and we pick the cut such that $\sqrt{z_+ z_-} = p/(1-p^2)$ which is the term that violates the symmetry $p \rightarrow 1/p$.  This function satisfies:
\begin{equation}
C(z,1/p) = C(z,p) + (p^{-1}-p) \,, \qquad z C(z,p)C(z,1/p) = C(z,p)-p
\end{equation}
 
Focus now on the rest.
Consider the integer partition $n-n_d =\sum_{k=1}^{2n_d} \ell_k$
where $0 \leq \ell_k  \leq n-n_d$, then diagram represented in Theorem~\ref{partq} is described by $2n_d$ group elements $b_k \in NC_{\ell_k}$ and is weighted
in the partition function above by:
\begin{equation}
w^{n_d} z^n p^{\sum_{k}(-1)^k (2 \#(b_k) - \ell_k -3/2)}
\end{equation}
where we pick $\#(b_k)  = 0$ for $\ell_k = 0$. 
We also need to account for translations of the diagram up to an amount of $1+\ell_1 + \ell_2$. By cyclic symmetry we can pick any two $\ell$'s and sum
$(\sum_{k=1}^{2 n_d}(1/2+ \ell_k))/n_d = n/n_d$. 
Thus we need to compute:
\begin{equation}
Z_{n_d \neq 0} =  \sum_{n=1}^{\infty} \sum_{n_d=1}^{n} \frac{n}{n_d}  \sum_{ \{\ell_k = 0, \ldots n-n_d \}_k }   \sum_{ \{ b_k \in NC_{\ell_k}\} } \delta_{ \sum_{k=1}^{2n_d} \ell_k, n- n_d} w^{n_d} z^n p^{\sum_{k}(-1)^k (2 \#(b_k) - \ell_k)}
\end{equation}
We introduce a contour integral to extract the correct power of $y$:
\begin{equation}
\delta_{ \sum_{k=1}^{2n_d} \ell_k, n- n_d}  = \oint_{C} \frac{dy}{ 2\pi i y}  y^{(n_d -n) + \sum_k \ell_k}
\end{equation}
where $C$ encircles the origin $y=0$. This also
allows us to extend the sum $ \sum_{\ell_k =0 }^{2 n_d} \rightarrow \sum_{\ell_k=0}^{\infty}$. 
Computing this:
\begin{align}
Z_{n_d \neq 0} =& z \partial_z \oint_{C} \frac{dy}{ 2\pi i y}  \sum_{n=1}^{\infty} \sum_{n_d=1}^{n} \frac{1}{n_d} w^{n_d} z^n y^{(n_d -n)} \left(  C(y, p) C(y,1/p) \right)^{n_d}
\end{align}
Shifting the sums to $n_d = 1, \ldots \infty$ and $\tilde{n} =  n - n_d = 0 , \ldots \infty$ allows us to do the sums:
\begin{align}
Z_{n_d \neq 0} = &  z \partial_z \oint_{C} \frac{dy}{ 2\pi i y}  \frac{1}{(1 - z/y)}  \ln ( 1 - w z C(y, p)  C(y, 1/p)   ) 
\end{align}
converging for small $w, z$. 
Then doing the contour integral to pickup the pole at $z = y$ (we can pick the contour $C$ to be at sufficiently small $y$ to avoid non-analyticities from the
 $C(y,p)$ terms above, but large enough to encircle this pole):
\begin{align}
Z_{n_d \neq 0} =    z \frac{ \partial}{\partial z} \ln ( 1 - w z C(z, p)   C(z, 1/p)  )
\end{align}
while this derivation required small $z,w$ we can then analytically continue this answer away from this regime. 

Putting it together we find:
\begin{equation}
Z(z,w,p) = C(z,p) + \frac{ z \partial_z   C(z,p)}{ C(z,p)-p - w^{-1} }
\end{equation}
Note that the second term has a pole at $z = \frac{p w}{(1 + pw)( p + w)}$ with unit residue. 

Now consider the problem of the minimal elements for $g \in S_{mn}$ in a fixed $n_d$ sector, satisfying \Eqref{sat0}.
\begin{equation}
Z(z,w,p) =  \sum_{n = 0}^{\infty} \sum_{n_d = 0}^{n} w^{n_d} z^{n}   \sum_{ g \in S_{mn} : \left\{ \substack{ \hspace{-.7cm}  \Eqref{sat0}\checkmark \\  \Eqref{equal}|_{q = q_g} \checkmark } \right. } p^{-\frac{1}{2} (\#(g g_A^{-1} ) - \#(g g_B^{-1})) }
\end{equation}
and this gives the same generating function.

We finally drop the second condition in \Eqref{sat0} such that we have a new degree of freedom.
Motivating the more complicated generating function:
\begin{equation}
\label{eq:gen_fun_Y}
Y(z,w,p,r) =  \sum_{n = 0}^{\infty} \sum_{n_d = 0}^{n} w^{n_d} z^{n}   \sum_{ g \in S_{mn} : \left\{ \substack{  G_A = G_B = 0\checkmark \\  \Eqref{equal}|_{q = q_g} \checkmark } \right. } r^{-\#(g)} p^{ \frac{1}{2}( \#(g g_A^{-1} )-  \#(g g_B^{-1}) )}
\end{equation}
In a fixed $q$ sector we have to sum over all elements $g \leq \iota(\hat{h}_q) g_A$ which simply gives:
\begin{equation}
\rho_1(r)^{n-n_d} \rho_0(r)^{n_d} \,, \qquad \rho_1(r) =  C_{m}(r^{-1}) - C_{m/2}(r^{-1})^2
\qquad \rho_0 (r)= C_{m/2}(r^{-1})^{2}
\end{equation}
where $C_m(r) = \sum_{g \in NC_m} r^{\#(g)}$, with $C_1(r) = 1$. 
So that:
\begin{equation}
Y(z,w,p,r)  = Z(\rho_1(r) z, \rho_0(r) \rho_1(r)^{-1} w, p) 
\end{equation}

\subsection{Reflected resolvent}

The resolvent can be expressed using the generating functions we found in previous subsection. We set:
\begin{equation}
\chi_A = p^{1/2} q^{1/2} \chi^{1/2}  \, \qquad \chi_B = p^{-1/2} q^{1/2} \chi^{1/2}
\qquad \chi_C = \chi 
\end{equation}
and take $\chi$ large. This zooms into the $x_A = x_B = 1/2$ part of the phase diagram.
\begin{align}
R(\lambda) = \sum_{n=0}^{\infty}    \lambda^{-n-1} \frac{\overline{ {\rm Tr} \sigma_{AA^\star}^n}}{
\overline{ \left( {\rm Tr} \rho_{AB}^{m} \right)}^{n}  } 
&\approx \lambda^{-1}  \hat{Y}( (\lambda \chi q  (\rho_1(q) + \rho_0(q)))^{-1},  \chi q  ,p,q)
\\
& =\lambda^{-1}  \hat{Z}( (\lambda \chi q)^{-1} p_1(q),   \chi q p_0(q)/p_1(q) ,p)
\end{align}
Where $\hat{Y}(z,w,p,q) =  (\chi q -1 )Y(z,0,p,q)  + Y(z,w,p,q)  $ and similarly for $\hat{Z}$. We have defined:
\begin{equation}
p_{0,1}(q) = \frac{\rho_{0,1}(q)}{\rho_0(q) + \rho_1(q)}
\end{equation}
At large $\chi$ there are multiple scales in the resolvent, which makes it difficult to prove a clean statements on the large $\chi$ answer.
We have the main part of the resolvent which describes the MP distribution given by:
\begin{align}
R(\lambda)|_{\lambda \sim 1/\chi} &= \lambda^{-1} \chi q C( (\lambda \chi q)^{-1} p_1(q),p) \\
& =\chi^2 \left(  - \frac{\sqrt{ (\lambda \chi - z'/z_+)(\lambda \chi -z'/z_-)} }{2z' \lambda \chi} + \frac{1}{2z'} - \frac{1}{2 \lambda \chi \sqrt{z_+ z_-}} \right)
\end{align} 
and is found near eigenvalues $\lambda \sim 1/\chi$ where we have defined $z' = p_1(q)/q$ and $z_\pm = z_\pm(p)$ was defined previously. Then there is a single pole approximately at $z w \approx 1$ (taking $ w $ large and $z w $ held fixed.) That is $\lambda \approx p_0(q)$. For $\lambda \sim \mathcal{O}(1)$ we have $z$ small so we can approximate $C(z,p) \approx p + z + \ldots$
\begin{equation}
 R(\lambda)|_{\lambda \sim 1} \approx  \frac{1}{\lambda - p_0(q)}
\end{equation}
In particular the naively leading MP distribution is not relevant because it has zero imaginary part for this range of
$\lambda$. 

Note the general location of the pole is:
\begin{equation}\label{eq:shifted_pole}
\lambda =  p_0 \left(p + \frac{p_1}{p_0 \chi q} \right)\left( \frac{1}{p} + \frac{p_1}{p_0 \chi q} \right) =  p_0 \left(1 +  \frac{p_1}{p_0 \chi_{A}^2} \right)\left( 1 + \frac{p_1}{p_0 \chi_{B}^2}\right) 
\end{equation}
although there is no reason to expect there are not other $1/\chi$ corrections that possibly even compete with the corrections
implied by the above formula. 
The more general resolvent that is represented by \Eqref{eq:gen_fun_Y}  is not accessible in this physical quantity since it lies outside of the domain of validity for the approximation that we are using. This is due to the large $w \sim \chi$ value that we must use in the generating function. In order to access this we need a way to hold fixed $w$ without changing $q$ or $p$.

\section{Full solution to $k=2$ SD-equation}
\label{app:SD-R}
In section \ref{sec:SD} we have the following set of matrix equations for the $2\times2$ matrix $R$ which determines the resolvent at $\lambda\sim O(\chi_C^0)$:
\begin{align}
  \label{eq:SD-R}
  \begin{split}
      \begin{pmatrix}
    R_{11} & R_{12} \\
    R_{21} & R_{22}
  \end{pmatrix} &= 
             \begin{pmatrix}
               \lambda - F_{11} & -F_{12} \\
               -F_{21} & \lambda- F_{22}
             \end{pmatrix} ^{-1} \\
  \begin{pmatrix}
    F_{11} & F_{12} \\
    F_{21} & F_{22}
  \end{pmatrix} &=
             \sqrt{\lambda}\begin{pmatrix}
               0 & \widehat{D}_m \\
              \widehat{D}_m & 0 
             \end{pmatrix}
             + \frac{\lambda\widehat{E}_m}{(\chi_A\chi_B)^2}
             \begin{pmatrix}
               R_{22} & R_{21} \\
               R_{12} & R_{11}
             \end{pmatrix}
                        + \frac{\lambda\widehat{B}_m}{(\chi_A\chi_B)^2}
                        \begin{pmatrix}
                          S_{22}(\chi^2_B-1) & 0 \\
                          0 & S_{11}(\chi^2_A-1)
                        \end{pmatrix} 
                      \end{split}
\end{align}
With $S_{11}$ and $S_{22}$ given by
\begin{align}
   (\chi_A^2 - 1) S_{11} , (\chi_B^2 - 1) S_{22} = \frac{ \chi_A^2 \chi_B^2}{2 \widehat{B}_m \lambda } \left( \lambda - \sqrt{ (\lambda - \lambda_+) (\lambda - \lambda_-)}  \right) \pm \frac{ \chi_A^2 - \chi_B^2}{2 \lambda}
\end{align}
We now attempt to find a complete solution for matrix $R$.
Define
\begin{align}
  M &=
      \begin{pmatrix}
        \lambda & -\sqrt{\lambda}\widehat{D}_m \\
        -\sqrt{\lambda}\widehat{D}_m & \lambda
      \end{pmatrix}
                       -\frac{\lambda \widehat{B}_m}{(\chi_A\chi_B)^2}
                       \begin{pmatrix}
                         S_{11} (\chi^2_A-1) & 0 \\
                         0 & S_{22} (\chi^2_B-1)
                       \end{pmatrix}
    &\equiv
      \begin{pmatrix}
        M_{11} & M_{12} \\
        M_{21} & M_{22}
      \end{pmatrix}
\end{align}
  Note that $M$ is symmetric, $M_{12} = M_{21}$.
  We can write \Eqref{eq:SD-R} as
  \begin{align}
    T^{-1}RT = \left( M-\frac{\lambda \widehat{E}_m}{(\chi_A\chi_B)^2}R \right)^{-1}
  \end{align}
  where
  \begin{align}
    T = T^{-1} =
    \begin{pmatrix}
      0 & 1 \\
      1 & 0
    \end{pmatrix}
  \end{align}
  is the $\chi_A\leftrightarrow \chi_B$ basis flip matrix. 
  Multiply both sides by $T$ from the left we obtain
  \begin{align}
    RT =\left( MT-\frac{\lambda \widehat{E}_m }{(\chi_A\chi_B)^2}RT  \right)^{-1}
  \end{align}
  What we have shown is that $RT$ and $MT$ are simultaneously diagonalizable.
  Our strategy is then straightforward: we simply go to the basis where both matrices are diagonalized.
  Doing so allows us to solve for the eigenvalues of $RT$ easily, then we transform back to the familiar basis and undo the effect of $T$.
  Start by diagonalizing $MT$:
    \begin{align}
    MT &=
    \begin{pmatrix}
      M_{12} & M_{11} \\
      M_{22} & M_{12}
    \end{pmatrix}
               = Q
               \begin{pmatrix}
                 d_+ & 0 \\
                 0 & d_-
               \end{pmatrix}
                     Q^{-1}
\nonumber \\
       &=
         \begin{pmatrix}
           \sqrt{M_{11}} & -\sqrt{M_{11}} \\
           \sqrt{M_{22}} & \sqrt{M_{22}}
           \end{pmatrix}
      \begin{pmatrix}
        M_{12}+\sqrt{M_{11}M_{22}} & 0 \\
        0 & M_{12}-\sqrt{M_{11}M_{22}}
      \end{pmatrix}
            \begin{pmatrix}
              \sqrt{M_{11}} & -\sqrt{M_{11}} \\
              \sqrt{M_{22}} & \sqrt{M_{22}}
            \end{pmatrix}^{-1}
  \end{align}
  We find that the eigenvalues of $MT$ are
  \begin{align}
    d_{\pm} =M_{12}\pm\sqrt{M_{11}M_{22}}
  \end{align}
  so we have
  \begin{align}
    R &= Q
    \begin{pmatrix}
        r_+ & 0 \\
        0 & r_-
      \end{pmatrix}
            Q^{-1}T \nonumber \\
    &=
      \frac{1}{2}
          \begin{pmatrix}
        \sqrt{\frac{M_{11}}{M_{22}}} (r_+-r_-) & r_++r_-\\
        r_++r_- & \sqrt{\frac{M_{22}}{M_{11}}} (r_+-r_-)
      \end{pmatrix} 
  \end{align}
  with $r_\pm$ given by
  \begin{align}
    r_{\pm} = \frac{(\chi_A\chi_B)^2}{2\lambda \widehat{E}_m }\left( d_\pm-\sqrt{d^2_\pm-\frac{4\lambda \widehat{E}_m}{(\chi_A\chi_B)^2}} \right)
  \end{align}
  In calculating the resolvent we only need the $R_{11}$ component. It is
  \begin{dmath}
    \label{eq:R11mess}
    R_{11}(\lambda) = 
    \sqrt{\frac{M_{11}}{M_{22}}}\frac{(\chi_A\chi_B)^2}{4\lambda \widehat{E}_m}
    \left( 2\sqrt{M_{11}M_{22}} +\sqrt{(M_{12}-\sqrt{M_{11}M_{22}})^2-\frac{4\lambda \widehat{E}_m}{(\chi_A\chi_B)^2}} - \sqrt{(M_{12}+\sqrt{M_{11}M_{22}})^2-\frac{4\lambda \widehat{E}_m}{(\chi_A\chi_B)^2}}\right)
  \end{dmath}
    The form of $R_{11}(\lambda)$ seems to predict two ``mini-MP'' peaks with width $\propto \sqrt{\widehat{E}_m}/(\chi_A\chi_B)$, centered at $\lambda_\star$ where the condition
  \begin{align}
    M_{12} = \pm \sqrt{M_{11}M_{22}}
  \end{align}
  is satisfied.
  However since $M_{12} \propto \sqrt{\lambda}$, from the branch cut of square root we see that solutions for these two equations lie immediately above and below real axis, respectively.
  In fact they represent the same peak, as we must have $R(\lambda+i\epsilon) = -R(\lambda-i\epsilon)$ for any resolvent $R(\lambda)$.
  Solving for $\lambda_\star$ gives 
  \begin{align}
    \begin{split}
      \lambda_\star &= \frac{(\widehat{B}_m (1-\chi_A^{-2})  \chi_B^{-2} + 
      \widehat{D}_m^2) (\widehat{B}_m(1-\chi_B^{-2})  \chi_A^{-2} + \widehat{D}_m^2 )}{\widehat{D}_m^2 } \\
    &= \widehat{D}^2_m + \left( \frac{1}{\chi^2_A} + \frac{1}{\chi^2_B}\right) \widehat{B}_m + O(\chi^{-4})
    \end{split}
  \end{align}
  We can calculate the width of this mini-MP peak via solving the relation
  \begin{equation}
    (M_{12}\pm\sqrt{M_{11}M_{22}})^2 = \frac{4\lambda\widehat{E}_m}{(\chi_A\chi_B)^2}
  \end{equation}
  We find
  \begin{align}
    \begin{split}
      \delta\lambda &= \frac{8\widehat{D}_m\sqrt{\widehat{E}_m}}{\chi_A\chi_B}\left( 1+\frac{\widehat{B}^2_m(\chi^2_A-1)(\chi^2_B-1)}{(4\widehat{E}_m-\widehat{D}^2_m\chi^2_A\chi^2_B)^2} \right)     \\
      &=\frac{8\widehat{D}_m\sqrt{\widehat{E}_m}}{\chi_A\chi_B} + O(\chi^{-4})
    \end{split}
  \end{align}
  Note that:
  \begin{equation}
    \lambda_\star-\lambda_{\pm}=   \left( \widehat{D}_m^2 \mp \widehat{B}_m \frac{(\chi_A^2-1)^{1/2}(\chi_B^2-1)^{1/2}}{(\chi_A \chi_B)^2}  \right)^2
  \end{equation}
  so the $\lambda_\star$ pole never reaches the main MP peak. It rather bounces when 
  \begin{equation}
    \widehat{D}_m^2 = \widehat{B}_m \frac{(\chi_A^2-1)^{1/2}(\chi_B^2-1)^{1/2}}{(\chi_A \chi_B)^2} 
  \end{equation}
  which is outside of the validity of our approximations. 

\bibliographystyle{jhep}
\bibliography{mybibliography}

\providecommand{\href}[2]{#2}\begingroup\raggedright\begin{thebibliography}{10}

\bibitem{tHooft:1993dmi}
G.~'t~Hooft, {\it {Dimensional reduction in quantum gravity}},  {\em Conf.
  Proc. C} {\bf 930308} (1993) 284--296,
  [\href{http://arxiv.org/abs/gr-qc/9310026}{{\tt gr-qc/9310026}}].

\bibitem{Susskind:1994vu}
L.~Susskind, {\it {The World as a hologram}},  {\em J. Math. Phys.} {\bf 36}
  (1995) 6377--6396, [\href{http://arxiv.org/abs/hep-th/9409089}{{\tt
  hep-th/9409089}}].

\bibitem{Ryu:2006bv}
S.~Ryu and T.~Takayanagi, {\it {Holographic derivation of entanglement entropy
  from AdS/CFT}},  {\em Phys. Rev. Lett.} {\bf 96} (2006) 181602,
  [\href{http://arxiv.org/abs/hep-th/0603001}{{\tt hep-th/0603001}}].

\bibitem{Ryu:2006ef}
S.~Ryu and T.~Takayanagi, {\it {Aspects of Holographic Entanglement Entropy}},
  {\em JHEP} {\bf 08} (2006) 045,
  [\href{http://arxiv.org/abs/hep-th/0605073}{{\tt hep-th/0605073}}].

\bibitem{Hubeny:2007xt}
V.~E. Hubeny, M.~Rangamani, and T.~Takayanagi, {\it {A Covariant holographic
  entanglement entropy proposal}},  {\em JHEP} {\bf 07} (2007) 062,
  [\href{http://arxiv.org/abs/0705.0016}{{\tt arXiv:0705.0016}}].

\bibitem{Engelhardt:2014gca}
N.~Engelhardt and A.~C. Wall, {\it {Quantum Extremal Surfaces: Holographic
  Entanglement Entropy beyond the Classical Regime}},  {\em JHEP} {\bf 01}
  (2015) 073, [\href{http://arxiv.org/abs/1408.3203}{{\tt arXiv:1408.3203}}].

\bibitem{Almheiri:2014lwa}
A.~Almheiri, X.~Dong, and D.~Harlow, {\it {Bulk Locality and Quantum Error
  Correction in AdS/CFT}},  {\em JHEP} {\bf 04} (2015) 163,
  [\href{http://arxiv.org/abs/1411.7041}{{\tt arXiv:1411.7041}}].

\bibitem{Jafferis:2015del}
D.~L. Jafferis, A.~Lewkowycz, J.~Maldacena, and S.~J. Suh, {\it {Relative
  entropy equals bulk relative entropy}},  {\em JHEP} {\bf 06} (2016) 004,
  [\href{http://arxiv.org/abs/1512.06431}{{\tt arXiv:1512.06431}}].

\bibitem{Dong:2016eik}
X.~Dong, D.~Harlow, and A.~C. Wall, {\it {Reconstruction of Bulk Operators
  within the Entanglement Wedge in Gauge-Gravity Duality}},  {\em Phys. Rev.
  Lett.} {\bf 117} (2016), no.~2 021601,
  [\href{http://arxiv.org/abs/1601.05416}{{\tt arXiv:1601.05416}}].

\bibitem{Hayden:2016cfa}
P.~Hayden, S.~Nezami, X.-L. Qi, N.~Thomas, M.~Walter, and Z.~Yang, {\it
  {Holographic duality from random tensor networks}},  {\em JHEP} {\bf 11}
  (2016) 009, [\href{http://arxiv.org/abs/1601.01694}{{\tt arXiv:1601.01694}}].

\bibitem{Akers:2018fow}
C.~Akers and P.~Rath, {\it {Holographic Renyi Entropy from Quantum Error
  Correction}},  {\em JHEP} {\bf 05} (2019) 052,
  [\href{http://arxiv.org/abs/1811.05171}{{\tt arXiv:1811.05171}}].

\bibitem{Dong:2018seb}
X.~Dong, D.~Harlow, and D.~Marolf, {\it {Flat entanglement spectra in
  fixed-area states of quantum gravity}},  {\em JHEP} {\bf 10} (2019) 240,
  [\href{http://arxiv.org/abs/1811.05382}{{\tt arXiv:1811.05382}}].

\bibitem{Dong:2019piw}
X.~Dong and D.~Marolf, {\it {One-loop universality of holographic codes}},
  {\em JHEP} {\bf 03} (2020) 191, [\href{http://arxiv.org/abs/1910.06329}{{\tt
  arXiv:1910.06329}}].

\bibitem{Dutta:2019gen}
S.~Dutta and T.~Faulkner, {\it {A canonical purification for the entanglement
  wedge cross-section}},  \href{http://arxiv.org/abs/1905.00577}{{\tt
  arXiv:1905.00577}}.

\bibitem{Maldacena:2001kr}
J.~M. Maldacena, {\it {Eternal black holes in anti-de Sitter}},  {\em JHEP}
  {\bf 04} (2003) 021, [\href{http://arxiv.org/abs/hep-th/0106112}{{\tt
  hep-th/0106112}}].

\bibitem{Hayden:2021gno}
P.~Hayden, O.~Parrikar, and J.~Sorce, {\it {The Markov gap for geometric
  reflected entropy}},  \href{http://arxiv.org/abs/2107.00009}{{\tt
  arXiv:2107.00009}}.

\bibitem{Takayanagi:2017knl}
T.~Takayanagi and K.~Umemoto, {\it {Entanglement of purification through
  holographic duality}},  {\em Nature Phys.} {\bf 14} (2018), no.~6 573--577,
  [\href{http://arxiv.org/abs/1708.09393}{{\tt arXiv:1708.09393}}].

\bibitem{Lewkowycz:2013nqa}
A.~Lewkowycz and J.~Maldacena, {\it {Generalized gravitational entropy}},  {\em
  JHEP} {\bf 08} (2013) 090, [\href{http://arxiv.org/abs/1304.4926}{{\tt
  arXiv:1304.4926}}].

\bibitem{Wall:2012uf}
A.~C. Wall, {\it {Maximin Surfaces, and the Strong Subadditivity of the
  Covariant Holographic Entanglement Entropy}},  {\em Class. Quant. Grav.} {\bf
  31} (2014), no.~22 225007, [\href{http://arxiv.org/abs/1211.3494}{{\tt
  arXiv:1211.3494}}].

\bibitem{Akers:2019lzs}
C.~Akers, N.~Engelhardt, G.~Penington, and M.~Usatyuk, {\it {Quantum Maximin
  Surfaces}},  \href{http://arxiv.org/abs/1912.02799}{{\tt arXiv:1912.02799}}.

\bibitem{Akers:2019gcv}
C.~Akers and P.~Rath, {\it {Entanglement Wedge Cross Sections Require
  Tripartite Entanglement}},  \href{http://arxiv.org/abs/1911.07852}{{\tt
  arXiv:1911.07852}}.

\bibitem{Cui:2018dyq}
S.~X. Cui, P.~Hayden, T.~He, M.~Headrick, B.~Stoica, and M.~Walter, {\it {Bit
  Threads and Holographic Monogamy}},  {\em Commun. Math. Phys.} {\bf 376}
  (2019), no.~1 609--648, [\href{http://arxiv.org/abs/1808.05234}{{\tt
  arXiv:1808.05234}}].

\bibitem{Page:1993df}
D.~N. Page, {\it {Average entropy of a subsystem}},  {\em Phys. Rev. Lett.}
  {\bf 71} (1993) 1291--1294, [\href{http://arxiv.org/abs/gr-qc/9305007}{{\tt
  gr-qc/9305007}}].

\bibitem{Akers:2020pmf}
C.~Akers and G.~Penington, {\it {Leading order corrections to the quantum
  extremal surface prescription}},  {\em JHEP} {\bf 04} (2021) 062,
  [\href{http://arxiv.org/abs/2008.03319}{{\tt arXiv:2008.03319}}].

\bibitem{Penington:2019kki}
G.~Penington, S.~H. Shenker, D.~Stanford, and Z.~Yang, {\it {Replica wormholes
  and the black hole interior}},  \href{http://arxiv.org/abs/1911.11977}{{\tt
  arXiv:1911.11977}}.

\bibitem{Shapourian:2020mkc}
H.~Shapourian, S.~Liu, J.~Kudler-Flam, and A.~Vishwanath, {\it {Entanglement
  negativity spectrum of random mixed states: A diagrammatic approach}},
  \href{http://arxiv.org/abs/2011.01277}{{\tt arXiv:2011.01277}}.

\bibitem{Marolf:2019zoo}
D.~Marolf, {\it {CFT sewing as the dual of AdS cut-and-paste}},  {\em JHEP}
  {\bf 02} (2020) 152, [\href{http://arxiv.org/abs/1909.09330}{{\tt
  arXiv:1909.09330}}].

\bibitem{Kusuki:2019evw}
Y.~Kusuki and K.~Tamaoka, {\it {Entanglement Wedge Cross Section from CFT:
  Dynamics of Local Operator Quench}},
  \href{http://arxiv.org/abs/1909.06790}{{\tt arXiv:1909.06790}}.

\bibitem{westcoast}
C.~Akers, T.~Faulkner, S.~Lin, and P.~Rath, {\it The page curve for reflected
  entropy}, .

\bibitem{harrow2013church}
A.~W. Harrow, {\it The church of the symmetric subspace},  {\em arXiv preprint
  arXiv:1308.6595} (2013).

\bibitem{Engelhardt:2017aux}
N.~Engelhardt and A.~C. Wall, {\it {Decoding the Apparent Horizon:
  Coarse-Grained Holographic Entropy}},  {\em Phys. Rev. Lett.} {\bf 121}
  (2018), no.~21 211301, [\href{http://arxiv.org/abs/1706.02038}{{\tt
  arXiv:1706.02038}}].

\bibitem{Engelhardt:2018kcs}
N.~Engelhardt and A.~C. Wall, {\it {Coarse Graining Holographic Black Holes}},
  {\em JHEP} {\bf 05} (2019) 160, [\href{http://arxiv.org/abs/1806.01281}{{\tt
  arXiv:1806.01281}}].

\bibitem{Dong:2021oad}
X.~Dong, S.~McBride, and W.~W. Weng, {\it {Replica Wormholes and Holographic
  Entanglement Negativity}},  \href{http://arxiv.org/abs/2110.11947}{{\tt
  arXiv:2110.11947}}.

\bibitem{Vardhan:2021npf}
S.~Vardhan, J.~Kudler-Flam, H.~Shapourian, and H.~Liu, {\it {Bound entanglement
  in thermalized states and black hole radiation}},
  \href{http://arxiv.org/abs/2110.02959}{{\tt arXiv:2110.02959}}.

\bibitem{Fischetti:2014zja}
S.~Fischetti and D.~Marolf, {\it {Complex Entangling Surfaces for AdS and
  Lifshitz Black Holes?}},  {\em Class. Quant. Grav.} {\bf 31} (2014), no.~21
  214005, [\href{http://arxiv.org/abs/1407.2900}{{\tt arXiv:1407.2900}}].

\bibitem{Dong:2020iod}
X.~Dong and H.~Wang, {\it {Enhanced corrections near holographic entanglement
  transitions: a chaotic case study}},  {\em JHEP} {\bf 11} (2020) 007,
  [\href{http://arxiv.org/abs/2006.10051}{{\tt arXiv:2006.10051}}].

\bibitem{Marolf:2020vsi}
D.~Marolf, S.~Wang, and Z.~Wang, {\it {Probing phase transitions of holographic
  entanglement entropy with fixed area states}},  {\em JHEP} {\bf 12} (2020)
  084, [\href{http://arxiv.org/abs/2006.10089}{{\tt arXiv:2006.10089}}].

\bibitem{Balasubramanian:2014hda}
V.~Balasubramanian, P.~Hayden, A.~Maloney, D.~Marolf, and S.~F. Ross, {\it
  {Multiboundary Wormholes and Holographic Entanglement}},  {\em Class. Quant.
  Grav.} {\bf 31} (2014) 185015, [\href{http://arxiv.org/abs/1406.2663}{{\tt
  arXiv:1406.2663}}].

\bibitem{Shenker:2013yza}
S.~H. Shenker and D.~Stanford, {\it {Multiple Shocks}},  {\em JHEP} {\bf 12}
  (2014) 046, [\href{http://arxiv.org/abs/1312.3296}{{\tt arXiv:1312.3296}}].

\bibitem{fannes1973continuity}
M.~Fannes, {\it A continuity property of the entropy density for spin lattice
  systems},  {\em Communications in Mathematical Physics} {\bf 31} (1973),
  no.~4 291--294.

\bibitem{audenaert2006sharp}
K.~M. Audenaert, {\it A sharp fannes-type inequality for the von neumann
  entropy},  {\em arXiv preprint quant-ph/0610146} (2006).

\bibitem{rastegin2011some}
A.~E. Rastegin, {\it Some general properties of unified entropies},  {\em
  Journal of Statistical Physics} {\bf 143} (2011), no.~6 1120--1135.

\bibitem{lubkin1978entropy}
E.~Lubkin, {\it Entropy of an n-system from its correlation with ak-reservoir},
   {\em Journal of Mathematical Physics} {\bf 19} (1978), no.~5 1028--1031.

\bibitem{liu2018entanglement}
Z.-W. Liu, S.~Lloyd, E.~Zhu, and H.~Zhu, {\it Entanglement, quantum randomness,
  and complexity beyond scrambling},  {\em Journal of High Energy Physics} {\bf
  2018} (2018), no.~7 1--62.

\bibitem{Jurkiewicz08}
J.~Jurkiewicz, G.~Lukaszewski, and M.~Nowak, {\it Diagrammatic approach to
  fluctuations in the wishart ensemble},  {\em Acta Physica Polonica B - ACTA
  PHYS POL B} {\bf 39} (04, 2008).

\bibitem{BREZIN1993613}
E.~Brézin and A.~Zee, {\it Universality of the correlations between
  eigenvalues of large random matrices},  {\em Nuclear Physics B} {\bf 402}
  (1993), no.~3 613--627.

\bibitem{BREZIN1995531}
E.~Brézin and A.~Zee, {\it Universal relation between green functions in
  random matrix theory},  {\em Nuclear Physics B} {\bf 453} (1995), no.~3
  531--551.

\bibitem{Kudler-Flam:2021efr}
J.~Kudler-Flam, V.~Narovlansky, and S.~Ryu, {\it {Negativity Spectra in Random
  Tensor Networks and Holography}},
  \href{http://arxiv.org/abs/2109.02649}{{\tt arXiv:2109.02649}}.

\bibitem{pagereflected}
C.~Akers, T.~Faulkner, S.~Lin, and P.~Rath, {\it {to appear}}, .

\bibitem{Harlow:2016vwg}
D.~Harlow, {\it {The Ryu\textendash{}Takayanagi Formula from Quantum Error
  Correction}},  {\em Commun. Math. Phys.} {\bf 354} (2017), no.~3 865--912,
  [\href{http://arxiv.org/abs/1607.03901}{{\tt arXiv:1607.03901}}].

\bibitem{Bao:2018pvs}
N.~Bao, G.~Penington, J.~Sorce, and A.~C. Wall, {\it {Beyond Toy Models:
  Distilling Tensor Networks in Full AdS/CFT}},  {\em JHEP} {\bf 11} (2019)
  069, [\href{http://arxiv.org/abs/1812.01171}{{\tt arXiv:1812.01171}}].

\bibitem{Qi:2017ohu}
X.-L. Qi, Z.~Yang, and Y.-Z. You, {\it {Holographic coherent states from random
  tensor networks}},  {\em JHEP} {\bf 08} (2017) 060,
  [\href{http://arxiv.org/abs/1703.06533}{{\tt arXiv:1703.06533}}].

\bibitem{Almheiri:2016blp}
A.~Almheiri, X.~Dong, and B.~Swingle, {\it {Linearity of Holographic
  Entanglement Entropy}},  {\em JHEP} {\bf 02} (2017) 074,
  [\href{http://arxiv.org/abs/1606.04537}{{\tt arXiv:1606.04537}}].

\bibitem{Faulkner:2020hzi}
T.~Faulkner, {\it {The holographic map as a conditional expectation}},
  \href{http://arxiv.org/abs/2008.04810}{{\tt arXiv:2008.04810}}.

\bibitem{Shaghoulian:2016xbx}
E.~Shaghoulian, {\it {Emergent gravity from Eguchi-Kawai reduction}},  {\em
  JHEP} {\bf 03} (2017) 011, [\href{http://arxiv.org/abs/1611.04189}{{\tt
  arXiv:1611.04189}}].

\bibitem{Donnelly:2016qqt}
W.~Donnelly, B.~Michel, D.~Marolf, and J.~Wien, {\it {Living on the Edge: A Toy
  Model for Holographic Reconstruction of Algebras with Centers}},  {\em JHEP}
  {\bf 04} (2017) 093, [\href{http://arxiv.org/abs/1611.05841}{{\tt
  arXiv:1611.05841}}].

\bibitem{Biane97}
P.~Biane, {\it Some properties of crossings and partitions},  {\em Discret.
  Math.} {\bf 175} (1997), no.~1-3 41--53.

\bibitem{birkhoff1967lattice}
G.~Birkhoff, {\em Lattice Theory}.
\newblock American Mathematical Society, Providence, 3rd~ed., 1967.

\bibitem{Kim12annular}
J.~S. Kim, S.~Seo, and H.~Shin, {\it Annular noncrossing permutations and
  minimal transitive factorizations},  {\em Journal of Combinatorial Theory
  Series A} {\bf 124} (01, 2012).

\end{thebibliography}\endgroup
 
\end{document}